\documentclass[11pt]{article}
\usepackage{amsmath, amssymb, amsfonts}
\usepackage{amsthm}
\usepackage{enumitem}
\usepackage{cite}
\newtheorem{lemma}{Lemma}       
\newtheorem{theorem}{Theorem}

\newtheorem{remark}{Remark}
\usepackage{helvet}
\usepackage{graphicx}
\usepackage{hyperref}
\usepackage{url}
\usepackage{geometry}
\usepackage{graphicx}
\usepackage{caption}
\usepackage{float} 

\DeclareCaptionLabelSeparator{pipe}{\ \textbar\ }
\usepackage{xcolor}

\renewcommand{\u}{\mathbf{u}}
\newcommand{\balpha}{\boldsymbol{\alpha}}
\newcommand{\bbeta}{\boldsymbol{\beta}}

\newcommand{\bPhi}{\boldsymbol{\Phi}}
\newcommand{\x}{\mathbf{x}}

\usepackage{amsmath, amssymb, amsfonts}
\usepackage{amsthm}
\usepackage{enumitem}
\usepackage{cite}
\usepackage{helvet}
\usepackage{graphicx}
\usepackage{hyperref}
\usepackage{url}
\usepackage{geometry}
\usepackage{graphicx}
\usepackage{caption}
\usepackage{float} 

\DeclareCaptionLabelSeparator{pipe}{\ \textbar\ }
\renewcommand{\thetable}{S\arabic{table}}
\usepackage{xcolor}
\renewcommand{\u}{\mathbf{u}}
\newcommand{\bv}{\mathbf{v}}

\usepackage{booktabs}
\usepackage{longtable}
\usepackage{array}
\usepackage{bibunits}

\title{Solving and Learning Navier-Stokes Dynamics with the Neural Basis Method}
\date{} 
\author{Yuhe Wang\thanks{Asia Pacific Technology, Inc., Houston, TX 77042, USA; Institute for Scientific Computation, Texas A\&M University, College Station, TX 77843, USA. \texttt{yuhe.wang@me.com}}, Zhu Wang\thanks{Department of  Mathematics, University of South Carolina, Columbia, SC 29208, USA. \texttt{wangzhu@math.sc.edu}}, Min Wang\thanks{Department of  Mathematics, University of Houston, TX 77204, USA. Email: \texttt{mwang50@uh.edu}}} 

\begin{document}

\makeatletter
\def\bibcite#1#2{%
  \global\@namedef{b@#1}{%
    \hyper@@link[cite]{}{cite.#1}{#2}%
  }%
}
\makeatother

\begin{bibunit}[unsrt]
\renewcommand{\thefootnote}{\fnsymbol{footnote}}

\maketitle
\vspace{-5ex}

\begin{abstract}
\noindent
Physics-informed machine learning has become a major direction in predictive modeling of fluid dynamics, combining governing physics with learned representations.
Yet prevailing formulations neither guarantee that small PDE-residual losses imply small solution errors nor scale effectively with increasing problem size.
Here we present the neural basis method (NBM) as a new approach for Navier--Stokes dynamics.
NBM is built on a ``representation-then-projection'' principle: fixed neural bases define the representation space, while projection of the governing equations determines the expansion coefficients.
The same bases define common coordinates for the parametric solution manifold, recasting operator learning as a parameter-to-coefficient mapping and enabling scalability by removing explicit spatial dependence.
For incompressible flow, NBM yields a velocity-only formulation whose residual provides an a posteriori velocity-error estimate, establishing a certified, data-free training objective for coefficient learning.
Theoretical analysis and numerical experiments demonstrate the accuracy, efficiency, generalization, and scalability of NBM, with substantial improvements over baselines.
\vspace{1em}  

\noindent\textbf{Keywords:} neural basis method, operator learning, physics-informed machine learning, parametric PDE, Navier-Stokes, incompressible flow and velocity-only formulation
\end{abstract}

\vspace{-15pt}
\section*{ }
\vspace{-15pt}
Computational methods for the Navier--Stokes systems translate fluid-motion laws into predictive tools for scientific inquiry and engineering practice~\cite{ferziger_computational_2020,girault2012finite,temam2024navier}.
Incompressible Newtonian flow is a canonical setting for theoretical and numerical studies of fluid motion and describes a broad class of practically important computational fluid dynamics (CFD) problems~\cite{elman2005finite,guermond2006overview,quartapelle2013numerical,buaria2020self}.
Established CFD approaches, including finite volume and finite element solvers, remain the gold standard, but each new parametric configuration generally requires a separate discretized solve~\cite{quarteroni2011certified,prud2002reliable,najm2009uncertainty}, 
creating a fundamental per-instance cost bottleneck in many-query settings such as uncertainty quantification, PDE-constrained optimization, and design~\cite{veroy2005certified,quarteroni2007numerical,benner2015survey}.

Physics-informed machine learning (PIML) has emerged as a compelling response to this bottleneck, seeking predictive models that retain contact with governing physics while exploiting the capabilities of learned representations~\cite{brunton2024promising,karniadakis2021physics,brunton2020machine}.
Incompressible Navier--Stokes systems are a central target and testbed for PIML, spanning physics-informed neural networks (PINNs) for individual PDE instances~\cite{raissi2019physics}, neural operators for reusable solution maps~\cite{lu2021learning,li2020fourier,kovachki2023neural}, reduced-order models for latent dynamics~\cite{rowley2017model,solera2024beta}, and hybrid approaches such as physics-informed neural operator learning (PINO)~\cite{li2024physics,wang2021learning,karumuri2026physics}.
Despite substantial progress, trade-offs remain among data requirements, physical consistency, optimization robustness, generalization, and computational cost~\cite{karniadakis2021physics,brunton2024promising,vinuesa2026decoding}, and no single framework has yet combined solver-level accuracy, physics-consistent prediction, scalable operator learning, and practical many-query efficiency~\cite{krishnapriyan2021characterizing,subedi2025operator}.

Two structural limitations are particularly important here. 
First, the governing equations are imposed externally through soft penalty losses over generic neural approximators, rather than built into the underlying neural formulation~\cite{karniadakis2021physics,li2024physics}.
Such losses generally act only as optimization surrogates during training, with no intrinsic guarantee that their reduction correspondingly reduces the solution error~\cite{de2024numerical,qiu2026variationally,bachmayr2025variationally}; the resulting optimization is also susceptible to ill-conditioning and training pathologies~\cite{wang2021understanding,wang2022and,krishnapriyan2021characterizing,chen2026training}.
Second, scalability remains fundamentally constrained by the explicit coupling to spatial resolution, with computational and memory costs increasing rapidly with spatial degrees of freedom or evaluation points~\cite{shukla2022scalable,subedi2025operator}.

The neural basis method (NBM), introduced in our previous work~\cite{wang2026solving}, is built on a ``representation-then-projection'' principle combining neural construction of the representation space~\cite{cybenko1989approximation,barron1993universal,mao2026approximation,takhanov2024multi} with the projection structure of classical numerical methods~\cite{babuvska1978posteriori,brenner2008mathematical}. 
Concretely, the representation is supported by fixed basis functions generated by a neural network, with physics embedded directly into the representation and discretization, while only the basis expansion coefficients are determined through projection.
The same basis space provides a common coordinate system for the parametric solution manifold, with each parameter instance mapped to the coordinates of its corresponding solution, i.e. the basis expansion coefficients, thereby recasting operator learning as a parameter-to-coefficient map.
This design provides a numerical foundation for rigorous convergence analysis and a posteriori error estimation, while removing the explicit dependence of operator learning on spatial resolution.
NBM is related to random feature methods through fixed neural bases and least-squares coefficient solves for individual PDE instances~\cite{huang2006extreme,chen2022bridging,rahimi2007random,dong2021local,zhang2024transferable,dwivedi2020physics,xu2026weak}, but differs in basis construction, PDE discretization, and more fundamentally, its operator learning viewpoint: rather than fit the operator through linear regression over parameter-dependent bases~\cite{nelsen2024operator,fabiani2025randonets}, NBM learns a general mapping from the parameters to the expansion coefficients in the fixed neural basis space.

Here, we present NBM as a new approach for solving and learning incompressible Navier--Stokes dynamics, addressing the two structural limitations identified above.
Specifically, we embed incompressibility a priori through a divergence-free velocity basis construction and cast Navier--Stokes into a conservative neural discretization compatible with classical finite-volume principles.
This naturally incorporates the Leray projection~\cite{leray2016motion} into NBM, yielding a velocity-only formulation in which pressure is no longer a primary unknown and avoiding the coupled velocity--pressure solution, a persistent source of difficulty in incompressible CFD~\cite{wang2009robust,guermond2006overview,brown2001accurate,john2017divergence}, while pressure can be recovered using proper postprocessing when needed. 
Building on NBM's numerical ground and the specific flow physics here, we provide a theoretical result on the convergence behavior characterizing the contributions of neural approximation, residual discretization, and algebraic solution to the overall error, and establish that the NBM residual provides a rigorous a posteriori estimate of velocity error.
This turns the residual from a surrogate loss into a certified operator-learning objective whose minimization directly controls velocity error without reference data, providing a training-loss-to-solution-error link generally absent from existing PIML approaches.
Meanwhile, by transforming operator learning as a parameter-to-velocity-coefficient map in a fixed divergence-free velocity basis space, NBM decouples learning from spatial resolution, addressing the second limitation.

We validate NBM on analytical benchmarks, boundary-driven flows, and shear-layer dynamics, obtaining accurate single-instance solutions with refinement behavior and residual--error correspondence consistent with the theoretical results.
For parametric problems, NBM achieve sub-percent velocity errors substantially below those of the baseline neural operators, and this advantage retains under generalization.
On a single NVIDIA H100 GPU, NBM achieves \(>10^{2}\times\) and \(>10\times\) measured speedups in training and inference, respectively, over these baselines.
Relative to the OpenFOAM finite volume solver, the inference speedup exceeds \(10^{4}\times\) and increases with prediction horizon and problem size.
We also show that the residual-based loss shapes the training landscape to align with the velocity error, mechanistically explaining the enhanced operator-learning performance.

\section*{Results}\label{results}
\vspace{-15pt}
We first present the NBM formulation for incompressible Newtonian flow, with convergence analysis, a posterior error estimate, and its implication for operator learning.
We then demonstrate single-instance PDE solutions and parametric accuracy, efficiency, robustness, and scalability, and investigate the mechanism underlying the improved performance.
\begin{figure}[h]
    \centering
    \includegraphics[width=\linewidth]{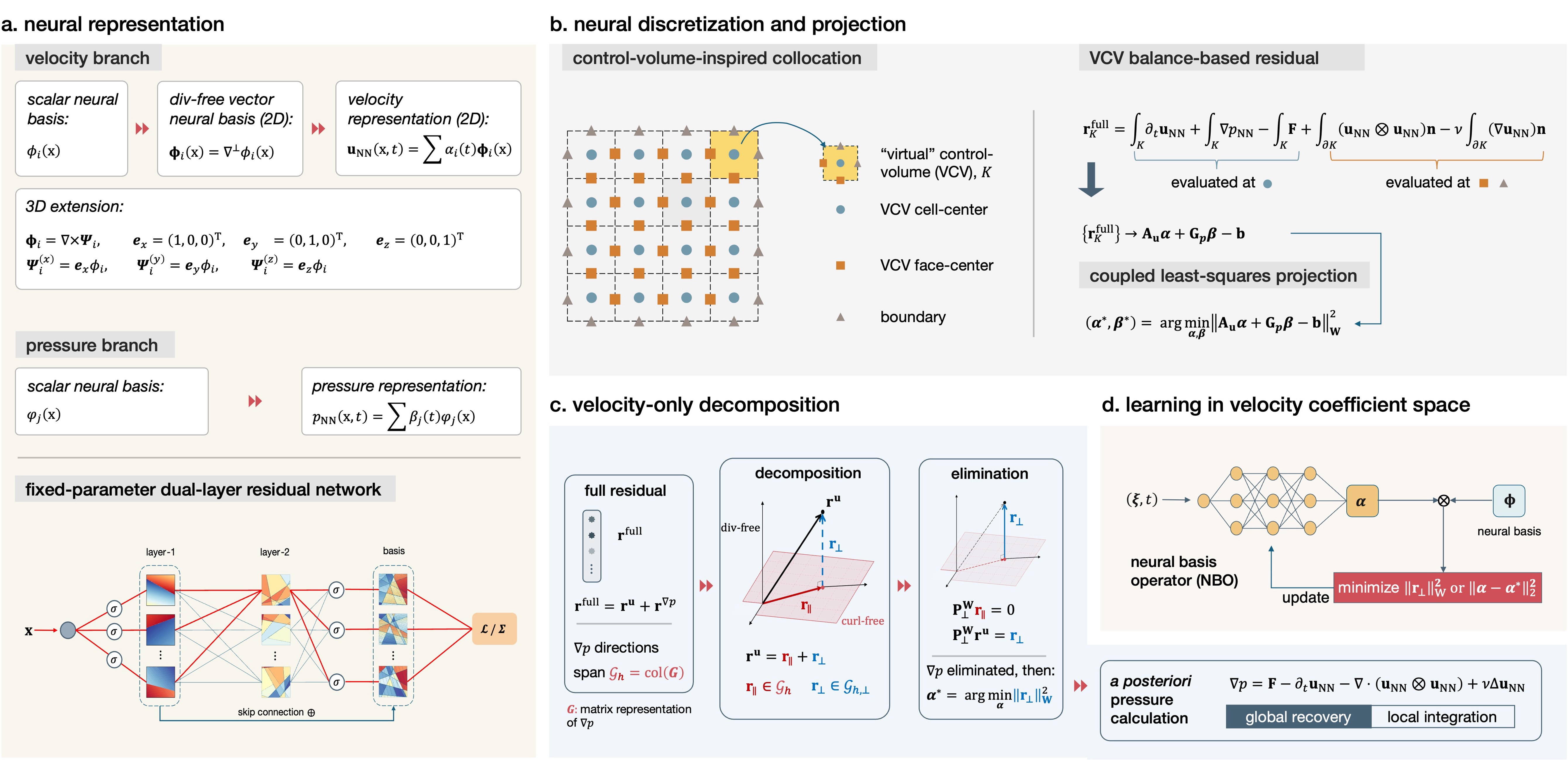}
    \caption{%
    \bfseries NBM formulation for incompressible Newtonian fluid dynamics.\,
    \bfseries a,
    \normalfont Neural representation. 
    Scalar neural bases are generated by fixed-parameter, dual-layer residual networks. 
    Divergence-free velocity bases are then constructed via rotated gradient in 2D, or from curls of vector-valued neural potentials in 3D. 
    Pressure field is represented using scalar neural bases.
    \bfseries b,
    \normalfont Neural discretization and projection.
    Residuals are evaluated on ``virtual'' control-volumes and assembled into a linearized least-squares system to determine the velocity and pressure coefficients.
    \bfseries c,
    \normalfont Velocity-only decomposition.
    The full momentum residual consists of velocity-dependent and
    pressure-gradient contributions,
    \(\mathbf{r}^{\mathrm{full}}=\mathbf{r}^{\mathbf{u}}+\mathbf{r}^{\nabla p}\),
    where
    \(\mathbf{r}^{\mathbf{u}}=\mathbf{A}_{\mathbf{u}}\boldsymbol{\alpha}-\mathbf{b}\)
    and
    \(\mathbf{r}^{\nabla p}=\mathbf{G}_p\boldsymbol{\beta}\).
    The velocity-dependent residual is decomposed into components
    parallel (\(\mathbf{r}_{\parallel}\)) and orthogonal
    (\(\mathbf{r}_{\perp}\)) to the discrete pressure-gradient space
    \(\mathcal{G}_h\), such that
    \(\mathbf{r}^{\mathbf{u}}=\mathbf{r}_{\parallel}+\mathbf{r}_{\perp}\).
    With incompressibility satisfied by the divergence-free
    velocity representation, applying the Leray projection
    eliminates the component parallel to the pressure-gradient
    space, leaving
    \(\mathbf{r}_{\perp}=\mathbf{P}_{\perp}^{\mathbf{W}}\mathbf{r}^{\mathbf{u}}\)
    and yielding the equivalent velocity-only least-squares formulation:
    \(\min_{\boldsymbol{\alpha},\boldsymbol{\beta}}
    \left\|\mathbf{r}^{\mathrm{full}}\right\|_{\mathbf{W}}^{2}
    =
    \min_{\boldsymbol{\alpha}}
    \left\|\mathbf{P}_{\perp}^{\mathbf{W}}
    \mathbf{r}^{\mathbf{u}}\right\|_{\mathbf{W}}^{2}\).
    Pressure can be recovered a posteriori.
    \bfseries d,
    \normalfont Learning in velocity-coefficient space.
    NBO maps parameter-time inputs to velocity coefficients.%
    }
    \label{fig:overview}
\end{figure}

\vspace{-15pt}
\subsection*{NBM formulation for incompressible Newtonian flow}
\vspace{-10pt}
We consider incompressible Newtonian flow governed by the Navier--Stokes equation in conservative form \eqref{eq:NS}, with incompressibility enforced by the continuity equation \eqref{eq:continuity}. 
\begin{equation}
\partial_t \mathbf{u}+\nabla\cdot(\mathbf{u}\otimes\mathbf{u})
-\nu\Delta\mathbf{u}+\nabla p=\mathbf{F},
\label{eq:NS}
\end{equation}
\begin{equation}
\nabla\cdot\mathbf{u}=0.
\label{eq:continuity}
\end{equation}
Here, \(\mathbf{u}\), \(p\), \(\mathbf{F}\), and \(\nu\) are the velocity, kinematic pressure, body force, and kinematic viscosity.
We present the formulation in 2D Cartesian coordinates for clarity;
it extends to 3D and other coordinate systems.
The formulation integrates a physics-conforming neural
representation, conservative ``virtual'' control-volume discretization, discrete pressure-gradient elimination, and operator learning directly in velocity-coefficient space (Fig.~\ref{fig:overview}).

\vspace{3pt}
\textbf{Neural representation (Fig.~\ref{fig:overview}a).}
We approximate velocity in a divergence-free neural basis space, satisfying incompressibility at the representation level.
Specifically, the velocity is represented as 
\(
\mathbf{u}_{\rm NN}(\mathbf{x},t)=\sum_i \alpha_i(t)\boldsymbol{\Phi}_i(\mathbf{x}),
\)
where \(\alpha_i(t)\) is the corresponding neural expansion coefficient. 
In 2D, each divergence-free vector basis $\mathbf{\Phi}_i$ is generated from a scalar neural basis \(\phi_i\) through
\(
\boldsymbol{\Phi}_i(\mathbf{x})
=
\nabla^\perp\phi_i(\mathbf{x})
=
\left(
\partial_y\phi_i(\mathbf{x}),
-\partial_x\phi_i(\mathbf{x})
\right)^{\top}
\) (See ~\cite{wang2026solving} for 3D extension).
Because \(\nabla\cdot\boldsymbol{\Phi}_i=0\), the represented velocity satisfies
\(
\nabla\cdot\mathbf{u}_{\rm NN}=0
\)
for any coefficient vector \(\boldsymbol{\alpha}\).
The kinematic pressure is represented independently as $p_{\rm NN}(\mathbf{x},t)=\sum_j \beta_j(t)\varphi_j(\mathbf{x})$, 
with scalar neural bases $\{\varphi_j\}$ and expansion coefficients $\{\beta_j(t)\}$.
All scalar neural bases use the fixed dual-layer ResNet architecture introduced in~\cite{wang2026solving}.
Once constructed, the basis spaces remain fixed; velocity and pressure depend linearly only on the coefficient vectors $\boldsymbol{\alpha}$ and $\boldsymbol{\beta}$, while the neural
bases retain nonlinear spatial expressivity.

\vspace{3pt}
\textbf{Neural discretization and projection (Fig.~\ref{fig:overview}b).}
With this neural representation, we construct a conservative neural discretization using a control-volume-inspired collocation scheme. 
Each collocation point defines a ``virtual'' control-volume (VCV), over which the integral balance of \eqref{eq:NS} is enforced. 
For each VCV \(K\), Gauss' theorem gives
\begin{equation}
\int_{K} \partial_t \mathbf{u}_{\rm NN}\, d\mathbf{x}
+
\int_{\partial K}
(\mathbf{u}_{\rm NN}\otimes\mathbf{u}_{\rm NN})
\mathbf{n}\, ds
-
\nu
\int_{\partial K}
\nabla\mathbf{u}_{\rm NN}\cdot\mathbf{n}\, ds
+
\int_{K}
\nabla p_{\rm NN}\, d\mathbf{x}
=
\int_{K}
\mathbf{F}\, d\mathbf{x}.
\label{eq:vlcv}
\end{equation}
The transient, pressure-gradient, and forcing terms are evaluated at VCV centers, whereas the convective and viscous fluxes are at VCV faces. 
Neural bases provide field values and spatial derivatives in \eqref{eq:vlcv} analytically.
The resulting residual retains the conservative flux structure of \eqref{eq:NS}, while the divergence-free representation enforces incompressibility. 
The flux form permits established schemes for transient and nonlinear convection terms (see Methods).
Assembling the residuals over all VCV and boundary collocation points yields an overdetermined system in the neural expansion coefficients. 
The NBM solution is then defined by weighted least-squares projection
\begin{equation}
\left(
\boldsymbol{\alpha}^{*},
\boldsymbol{\beta}^{*}
\right)
=
\arg\min_{\boldsymbol{\alpha},\boldsymbol{\beta}}
\left\|
\mathbf{A}_{\mathbf{u}}\boldsymbol{\alpha}
+
\mathbf{G}_{p}\boldsymbol{\beta}
-
\mathbf{b}
\right\|_{\mathbf{W}}^{2}.
\label{eq:nbm_ls_full}
\end{equation}
Here, $\mathbf W=\mathbf W^T\succ0$ specifies the discrete weighting, or
equivalently the discrete measure, on the VCV residuals, \(\mathbf{A}_{\mathbf{u}}\) collects velocity-dependent contribution, \(\mathbf{G}_{p}\) represents pressure-gradient contribution, and \(\mathbf{b}\) contains forcing, velocity boundary, and, for unsteady problems, known contributions from previous time levels. 
Unlike velocity, explicit pressure boundary condition is not required in the present formulation~\cite{sani2006pressure} and thus \(\mathbf{b}\) contains no pressure-dependent entries.
This deterministic convex least-squares problem determines the coupled velocity and pressure coefficients. 

\vspace{3pt}
\textbf{Velocity-only construction via pressure-gradient elimination (Fig.~\ref{fig:overview}c).}
In incompressible flow, pressure is not an independent variable but act as a Lagrange multiplier enforcing incompressibility.
Because the divergence-free velocity neural bases satisfy
incompressibility a priori, pressure need not be a primary NBM
unknown, echoing the classical Leray-projection
idea~\cite{leray2016motion}.
Let \(\mathbb{P}_{\perp}\) denote the projection onto a divergence-free space, equivalently the orthogonal complement of a gradient-field space. 
Applying \(\mathbb{P}_{\perp}\) to \eqref{eq:NS} eliminates the pressure-gradient term, yielding
\(
\mathbb{P}_{\perp}
\left[\partial_t \mathbf{u}
+\nabla\cdot(\mathbf{u}\otimes\mathbf{u})
-
\nu\Delta\mathbf{u}
-
\mathbf{F}
\right]=0
\)
and defining the continuous residual 
\(\mathcal R(\u) := 
\mathbb{P}_{\perp}
\left[\partial_t \mathbf{u}
+\nabla\cdot(\mathbf{u}\otimes\mathbf{u})
-
\nu\Delta\mathbf{u}
-
\mathbf{F}
\right].\)
At the discrete level, we construct a space to represent the pressure-gradient directions in the residuals over VCVs. 
This construction is not unique. 
We denote this discrete pressure-gradient space by \(\mathcal{G}_{h}\) and its matrix representation by \(\mathbf{G}\), whose column space spans these pressure-gradient directions (\(\mathbf{G}\) is not required to coincide with the pressure-gradient block \(\mathbf{G}_{p}\) in \eqref{eq:nbm_ls_full}). 
Let \(\mathbf{P}_{\perp}^{\mathbf{W}}\) denote the weighted orthogonal projection onto the complement of \(\mathcal{G}_{h}\), satisfying 
\(
\mathbf{P}_{\perp}^{\mathbf{W}}\mathbf{G}
=
\mathbf{0}.
\)
With
\(\|\mathbf{r}\|_{\mathbf{W}}^{2}
=
\mathbf{r}^{\top}\mathbf{W}\mathbf{r}\),
one realization of this projector is
\(
\mathbf{P}_{\perp}^{\mathbf{W}}
=
\mathbf{I}
-
\mathbf{G}
\left(
\mathbf{G}^{\top}\mathbf{W}\mathbf{G}
\right)^{\dagger}
\mathbf{G}^{\top}\mathbf{W},
\)
where \((\cdot)^{\dagger}\) denotes the pseudoinverse. 
The velocity coefficients follow from the velocity-only least-squares problem
\begin{equation}
\boldsymbol{\alpha}^{*}
=
\arg\min_{\boldsymbol{\alpha}}
\left\|
\mathbf{P}_{\perp}^{\mathbf{W}}
(\mathbf{A}_{\mathbf{u}}\boldsymbol{\alpha}
-
\mathbf{b})
\right\|_{\mathbf{W}}^{2}.
\label{eq:velocity_only_projection}
\end{equation}
We interpret this formulation as eliminating the pressure-gradient contributions from \eqref{eq:nbm_ls_full} by projecting out the component of the momentum residual in the discrete pressure-gradient space $\operatorname{Col}(\mathbf G)$, thereby yielding a least-squares problem in the velocity-coefficient space.
Indeed, for a fixed $\balpha$, minimizing over the pressure-gradient coefficients gives
\(
\min_{\bbeta}
\left\|
\mathbf A_{\u}\balpha
+
\mathbf G_p\bbeta
-
\mathbf b
\right\|_{\mathbf W}^{2}
=
\left\|
\mathbf P_\perp^{\mathbf W}
\left(
\mathbf A_{\u}\balpha-\mathbf b
\right)
\right\|_{\mathbf W}^{2}.
\)
Then,
\(
\min_{\balpha,\bbeta}
\left\|
\mathbf A_{\u}\balpha
+
\mathbf G_p\bbeta
-
\mathbf b
\right\|_{\mathbf W}^{2}
=
\min_{\balpha}
\left\|
\mathbf P_\perp^{\mathbf W}
\left(
\mathbf A_{\u}\balpha-\mathbf b
\right)
\right\|_{\mathbf W}^{2}.
\)
Pressure can be computed as postprocessing when required. 
This velocity-only construction brings two immediate advantages. 
The velocity-only projected least-squares system significantly reduces primary unknowns relative to the coupled velocity–pressure problem. For operator learning, predicting only velocity coefficients also greatly simplifies its formulation and analysis.

\vspace{3pt}
\textbf{Learning in the velocity coefficient space (Fig.~\ref{fig:overview}d).}
A single-instance NBM solve may cost \(O(N_cN_b^2)\), with \(N_b\) neural bases and \(N_c\) collocation VCVs.
Repeated solutions can quickly become demanding in many-query settings with varying physical parameters, boundary conditions or forcing. 
Building on our previous NBM-powered operator-learning framework~\cite{wang2026solving}, we use the present velocity-only construction to learn the parametric map directly in velocity-coefficient space. 
For a parameter--time input \(\boldsymbol{\zeta}=(\boldsymbol{\xi},t)\), the neural basis operator (NBO) approximates this map as \(\mathcal{N}_{\boldsymbol{\eta}}:\boldsymbol{\zeta}\mapsto\boldsymbol{\alpha}_{\boldsymbol{\eta}}(\boldsymbol{\zeta})\), where \(\boldsymbol{\eta}\) denotes the trainable opertor-network parameters. 
The corresponding velocity field is reconstructed in the fixed divergence-free neural basis as \(\mathbf{u}_{\rm NN}(\mathbf{x};\boldsymbol{\zeta})=\sum_i\alpha_{\boldsymbol{\eta},i}(\boldsymbol{\zeta})\boldsymbol{\Phi}_i(\mathbf{x})\), so that the velocity coefficients provide a finite-dimensional representation of the parametric flow dynamics. 
Training the NBO can be fully self-supervised if using the following loss function based on the pressure-eliminated residual in \eqref{eq:velocity_only_projection}:
\begin{equation}
\mathcal{L}_{\rm RC}(\boldsymbol{\eta})
=
\mathbb{E}_{\boldsymbol{\zeta}\in\mathcal{D}_{\rm tr}}
\left[
\left\|
\mathbf{P}_{\perp}^{\mathbf{W}}
\left(
\mathbf{A}_{\mathbf{u}}(\boldsymbol{\zeta})
\boldsymbol{\alpha}_{\boldsymbol{\eta}}(\boldsymbol{\zeta})
-
\mathbf{b}(\boldsymbol{\zeta})
\right)
\right\|_{\mathbf{W}}^{2}
\right].
\label{eq:nbm_ol_loss}
\end{equation}
Here, \(\mathcal{D}_{\rm tr}\) is the parameter--time training set, and the expectation $\mathbb{E}_{\boldsymbol{\zeta}\in\mathcal{D}_{\rm tr}}$ denotes its empirical average. 
We show below that this NBM residual-based loss provides a posteriori estimate of the velocity error and thus a certificate for monitoring solution accuracy and guiding the operator-learning process;
we call this approach residual-certified NBO (NBO-RC).
Alternatively, direct training on reference velocity-coefficient data yields the data-driven NBO (NBO-DD).

\vspace{-0pt}
\subsection*{NBM convergence analysis and a posteriori error estimation}
\vspace{-10pt}
We establish two theoretical results characterizing the convergence from neural approximation, VCV discretization, and algebraic solution, and providing a posteriori velocity-error estimate underpinning residual-certified operator learning.

\vspace{3pt}
\textbf{Triple-resolution convergence analysis.}
Let
\(V=\{\mathbf{v}\in H^1(\Omega)^d:
\nabla\cdot\mathbf{v}=0\}\)
denote the divergence-free velocity space equipped
with the norm
\(\|\cdot\|_V=\|\cdot\|_{H^1(\Omega)^d}\).
Let \(V_{N_b}\subset V\) be the neural basis space
spanned by \(N_b\) independently sampled
divergence-free neural basis functions, i.e.,
\(V_{N_b}:=\operatorname{span}
\{\boldsymbol{\Phi}_i,\ 1\leq i\leq N_b\}\).
We denote the exact discrete NBM solution by
\(
\u_{N_b,N_c}^{*}
:=
\sum_{i=1}^{N_b}\alpha_i^{*}\bPhi_i(\x),
\)
where $\balpha^{*}$ is the exact minimizer defined in
\eqref{eq:velocity_only_projection} subject to $N_c$ VCVs. The computed solution is denoted by
\(
\widetilde{\u}_{N_b,N_c}
:=
\sum_{i=1}^{N_b}\widetilde{\alpha}_i\bPhi_i(\x),
\)
where the computed coefficient vector $\widetilde{\balpha}$ is controlled by the least-squares optimality residual 
\(
\left\|
\mathbf{A}_{N_b,N_c}^{T}
\left(
\mathbf{A}_{N_b,N_c}\widetilde{\balpha}
-
\mathbf{b}_{N_c}
\right)
\right\|_2
\le
\varepsilon_{\mathrm{LS}},
\)
with $\varepsilon_{\mathrm{LS}}>0$ denoting the prescribed tolerance, 
$\mathbf A_{N_b,N_c} := \mathbf W^{1/2} \mathbf P_\perp^{\mathbf W} \mathbf A_{\u},$ and $\mathbf b_{N_c} := \mathbf W^{1/2} \mathbf P_\perp^{\mathbf W} \mathbf b.$
The minimizer $\balpha^*$ is unique if $\mathbf A_{N_b,N_c}$ has full column rank.

\begin{theorem}
\label{thm:nbm_error_estimate}
Assume that the continuous residual operator $\mathcal{R}(\cdot)$ is locally stable
and locally Lipschitz continuous in a neighborhood of the exact solution $\u$. Assume also the neural basis representation is uniformly stable in $V$, namely,
$\left\|
\sum_{i=1}^{N_b}\alpha_i\bPhi_i
\right\|_V
\le
C_{\mathrm{basis}} \|\balpha\|_2,$ for all $\balpha \in \mathbb{R}^{N_b}$, with $C_{\mathrm{basis}}>0$ independent of $N_b$.
Assume further that $\mathbf A_{N_b,N_c}$ has full column rank
and that the hypotheses of Lemmas~S1-S3 hold.
Let $q>0$ denote the order of accuracy of the VCV cell
and face quadrature rules.
Then, for any $0<\delta<1$ and sufficiently large $N_b$ and $N_c$, with
probability at least $1-\delta$, 
\[
\left\|
\u-\widetilde{\u}_{N_b,N_c}
\right\|_{V}
\leq
C_1
\exp\!\left[
-\alpha
\left(
\frac{N_b}{\log(N_b/\delta)}
\right)^{\beta}
\right]
+
C_2N_c^{-q/d}
+
C_3
\frac{\varepsilon_{\mathrm{LS}}}
{\left[\sigma_{\min}\!\left(\mathbf A_{N_b,N_c}\right)\right]^2},
\]
where \(C_1,C_2,C_3>0\), \(\alpha>0\), and \(0<\beta\leq1\) are independent of \(N_b\) and \(N_c\).
\end{theorem}

\begin{proof}
The proof including Lemmas S1-S3 are provided in Supplementary Information.
\end{proof}

This convergence estimate separates three NBM error components and suggests corresponding refinement: 1) 
Neural approximation error decreases at a spectral-like high-probability rate with increasing \(N_b\), with a logarithmic correction from the random basis coverage; 
2) VCV discretization error decreases algebraically with increasing \(N_c\). 
3) Least-squares error is controlled by solver tolerance and condition number. 

\vspace{3pt}
\textbf{Residual--error certificate.}
Under appropriate stability conditions, the NBM residual estimates velocity error a posteriori. 
The following result establishes a residual--error certificate for the velocity-only VCV formulation, with its continuous counterpart provided in Supplementary Section S4. 
Consider the residual associated with one linearized implicit time step,
\(
R_{\perp}(\boldsymbol{\alpha})
=
\mathbf{L}_{\perp}\boldsymbol{\alpha}
-
\mathbf{g}_{\perp},
\)
where \(\mathbf{L}_{\perp}\) is the linear system matrix and \(\mathbf{g}_{\perp}\) is the right-hand side containing all known contributions.
Let \(\boldsymbol{\alpha}^{\ast}\) denote the coefficient vector of the exact projected discrete solution, such that \(R_{\perp}(\boldsymbol{\alpha}^{\ast})=\mathbf{0}\), and let \(\widetilde{\boldsymbol{\alpha}}\) denote an approximate coefficient vector. 
Defining \(\mathbf{u}^{\ast}=\mathbf{u}(\boldsymbol{\alpha}^{\ast})\), \(\widetilde{\mathbf{u}}=\mathbf{u}(\widetilde{\boldsymbol{\alpha}})\), and \(\delta\boldsymbol{\alpha}=\widetilde{\boldsymbol{\alpha}}-\boldsymbol{\alpha}^{\ast}\), the residual satisfies the discrete error equation
\(
R_{\perp}(\widetilde{\boldsymbol{\alpha}})
=
\mathbf{L}_{\perp}\delta\boldsymbol{\alpha}.
\)
We measure the velocity error using the discrete energy norm
\(
\left\|
\widetilde{\mathbf{u}}
-
\mathbf{u}^{\ast}
\right\|_{\mathcal{E}}^{2}
=
\delta\boldsymbol{\alpha}^{T}
\mathbf{E}
\delta\boldsymbol{\alpha}
\) with
\(
\mathbf{E}
=
\frac{1}{\Delta t}\mathbf{M}
+
\nu\mathbf{K}.
\)
\(\mathbf{M}\) and \(\mathbf{K}\) denote the assembled velocity and velocity-gradient contributions. 
For the residual, we measure it using the weighted norm 
\(
\left\|R_{\perp}\right\|_{\mathbf{W}}^{2}
=
R_{\perp}^{T}
\mathbf{W}
R_{\perp},
\)
where \(\mathbf{W}\) contains the VCV-based weights. 
Defining the corresponding coefficient-space residual Gram matrix \(
\mathbf{G}_{\rm R}
=
\mathbf{L}_{\perp}^{T}
\mathbf{W}
\mathbf{L}_{\perp},
\)
we can write 
\(
\left\|
R_{\perp}(\widetilde{\boldsymbol{\alpha}})
\right\|_{\mathbf{W}}^{2}
=
\delta\boldsymbol{\alpha}^{T}
\mathbf{G}_{\rm R}
\delta\boldsymbol{\alpha}.
\)
Let
\(
\mathbf{B}_{\perp}
=
\mathbf{W}^{1/2}
\mathbf{L}_{\perp}
\mathbf{E}^{-1/2}.
\)
For the same linearized NBM system in
\eqref{eq:velocity_only_projection},
\(\mathbf L_\perp
=\mathbf P_\perp^{\mathbf W}\mathbf A_{\mathbf u}\).
Consequently,
\(\mathbf B_\perp
=\mathbf A_{N_b,N_c}\mathbf E^{-1/2}\),
so the full-column-rank conditions on
\(\mathbf B_\perp\) and
\(\mathbf A_{N_b,N_c}\) are equivalent
when \(\mathbf E\) is positive definite.

\begin{theorem}
\label{thm:projected_ns_residual_error_certificate}
Assume that \(\mathbf{E}\) is positive definite and that \(\mathbf{B}_{\perp}\) has full column rank. Then, for any NBM approximation \(\widetilde{\mathbf{u}}=\mathbf{u}(\widetilde{\boldsymbol{\alpha}})\),
\(
\sigma_{\min}(\mathbf{B}_{\perp})
\left\|
\widetilde{\mathbf{u}}
-
\mathbf{u}^{\ast}
\right\|_{\mathcal{E}}
\leq
\left\|
R_{\perp}(\widetilde{\boldsymbol{\alpha}})
\right\|_{\mathbf{W}}
\leq
\sigma_{\max}(\mathbf{B}_{\perp})
\left\|
\widetilde{\mathbf{u}}
-
\mathbf{u}^{\ast}
\right\|_{\mathcal{E}}.
\)
Equivalently,
\(
\left\|
R_{\perp}(\widetilde{\boldsymbol{\alpha}})
\right\|_{\mathbf{W}}
\simeq
\left\|
\widetilde{\mathbf{u}}
-
\mathbf{u}^{\ast}
\right\|_{\mathcal{E}}.
\)
\end{theorem}

\begin{proof}
The proof is provided in Supplementary Information.
\end{proof}

The singular values \(\sigma_{\min}(\mathbf{B}_{\perp})\) and
\(\sigma_{\max}(\mathbf{B}_{\perp})\) are discrete stability constants. 
They are independent of the approximation but may depend on the discretization, including neural basis space, VCV weights, time-step, viscosity, linearized convective field, and pressure-gradient projection. 
The theorem establishes two-sided equivalence between projected residual and discrete velocity-error, providing a certificate for characterizing parametric NBO-RC training.

\vspace{3pt}
\textbf{Certificate-induced properties for NBM operator learning}
The residual--error certificate can apply instancewise to operator learning over the parameter--time domain. 
For each \(\boldsymbol{\zeta}_i\), we identify the NBO prediction with \(\widetilde{\boldsymbol{\alpha}}(\boldsymbol{\zeta}_i)=\boldsymbol{\alpha}_{\boldsymbol{\eta}}(\boldsymbol{\zeta}_i)\) and denote the corresponding reconstructed velocity by \(\widetilde{\mathbf{u}}(\boldsymbol{\zeta}_i)\). 
If the discrete stability constants in Theorem~\ref{thm:projected_ns_residual_error_certificate} are uniformly bounded over \(\mathcal{D}_{\rm tr}\), then
\(
\mathcal{L}_{\rm RC}(\boldsymbol{\eta})
\simeq
\mathbb{E}_{\boldsymbol{\zeta}\in\mathcal{D}_{\rm tr}}
\left[
\left\|
\widetilde{\mathbf{u}}(\boldsymbol{\zeta})
-
\mathbf{u}^{\ast}(\boldsymbol{\zeta})
\right\|_{\mathcal{E}}^{2}
\right].
\)
Thus, minimizing the NBO-RC training loss \eqref{eq:nbm_ol_loss} is equivalent, up to the discrete stability constants, to minimizing the average discrete velocity-energy error over the training set.
The certificate also specifies coefficient-error weighting during NBO-RC training. 
Rewrite \eqref{eq:nbm_ol_loss} as
\(
\mathcal{L}_{\rm RC}(\boldsymbol{\eta})
=
\mathbb{E}_{\boldsymbol{\zeta}\in\mathcal{D}_{\rm tr}}
\left[
\delta\boldsymbol{\alpha}(\boldsymbol{\zeta})^{T}
\mathbf{G}_{\rm R}(\boldsymbol{\zeta})
\delta\boldsymbol{\alpha}(\boldsymbol{\zeta})
\right],
\)
where
\(
\mathbf{G}_{\rm R}(\boldsymbol{\zeta})
=
\mathbf{L}_{\perp}(\boldsymbol{\zeta})^{T}
\mathbf{W}(\boldsymbol{\zeta})
\mathbf{L}_{\perp}(\boldsymbol{\zeta})
\)
is the residual Gram matrix evaluated at each parameter--time training input \(\boldsymbol{\zeta}\), defining the certificate-induced metric in velocity-coefficient space. 
By contrast, NBO-DD minimizes
\(
\mathcal{L}_{\rm DD}(\boldsymbol{\eta})
=
\mathbb{E}_{\boldsymbol{\zeta}\in\mathcal{D}_{\rm tr}}
\left[
\delta\boldsymbol{\alpha}(\boldsymbol{\zeta})^{T}
\mathbf{I}
\delta\boldsymbol{\alpha}(\boldsymbol{\zeta})
\right],
\)
and hence measures coefficient errors using the mean-squared error (MSE) metric rather than NBO-RC's residual-induced coefficient-space metric. 

\vspace{-0pt}
\subsection*{NBM delivers residual-certified accuracy as a single-instance PDE solver}
\vspace{-10pt}
\textbf{Kovasznay flow and Taylor--Green vortex (Fig.~\ref{fig:kovasznay-taylor-green}).}
We first evaluate NBM on steady Kovasznay flow and transient Taylor--Green vortex, for which analytical solutions are available. 
NBM uses \(N_b=500\) dual-layer neural bases and \(N_c=50\times50\) structured uniform VCV collocation, OpenFOAM uses a \(50\times50\) structured finite-volume mesh, and PINN uses the same number of collocation points sampled using scrambled Sobol sequences. 
For both benchmarks, NBM achieves substantially lower errors than OpenFOAM and PINN in the velocity and pressure fields, as well as in the statistical quantities \(U_{\rm rms}\) and \(E(k)\). 
The refinement results further agree with the error structure established in Theorem~\ref{thm:nbm_error_estimate}: increasing \(N_b\) produces spectral-like decay of the solution error by reducing the neural approximation error, whereas increasing \(N_c\) produces algebraic decay by reducing the VCV residual-discretization error. 
Moreover, the residual and velocity error decrease concurrently under both refinements, echoing the residual--error equivalence in Theorem~\ref{thm:projected_ns_residual_error_certificate}.

\vspace{3pt}
\textbf{Lid-driven cavity and shear-layer flow (Fig.~\ref{fig:cavity-shear-layer}).}
We then evaluate NBM on steady lid-driven cavity flow and transient shear-layer flow, extending the assessment beyond analytically tractable benchmarks to boundary-driven recirculation and instability-dominated flow dynamics.
NBM uses \(N_b=1000\) double-layer neural bases and \(N_c=100\times100\) structured uniform VCV collocation, while OpenFOAM uses a \(100\times100\) structured finite-volume mesh and PINN uses the same number of scrambled Sobol collocation points. 
Because closed-form solutions are unavailable, the corresponding OpenFOAM solutions are used as numerical references. 
For the lid-driven cavity at ($Re=1000$), standard PINN trained solely by minimizing the PDE loss fails to converge. 
Sparse velocity anchors extracted from the OpenFOAM solution are therefore added to obtain a usable PINN solution. 
Even with these velocity anchors, PINN's discrepancies from OpenFOAM remain substantially larger than NBM's in both resolved flow fields and derived statistical quantities. 
Refinement again produces concurrent decreases in residual and solution error, as in the analytical benchmarks. 
Together, these benchmarks connect the high-probability triple-resolution error estimate and residual--error certificate to measurable quantities in computed solutions, demonstrating that NBM functions as an accurate and residual-certified single-instance neural PDE solver.

\clearpage
\begin{figure}[h!]
  \centering
  \includegraphics[width=\linewidth]{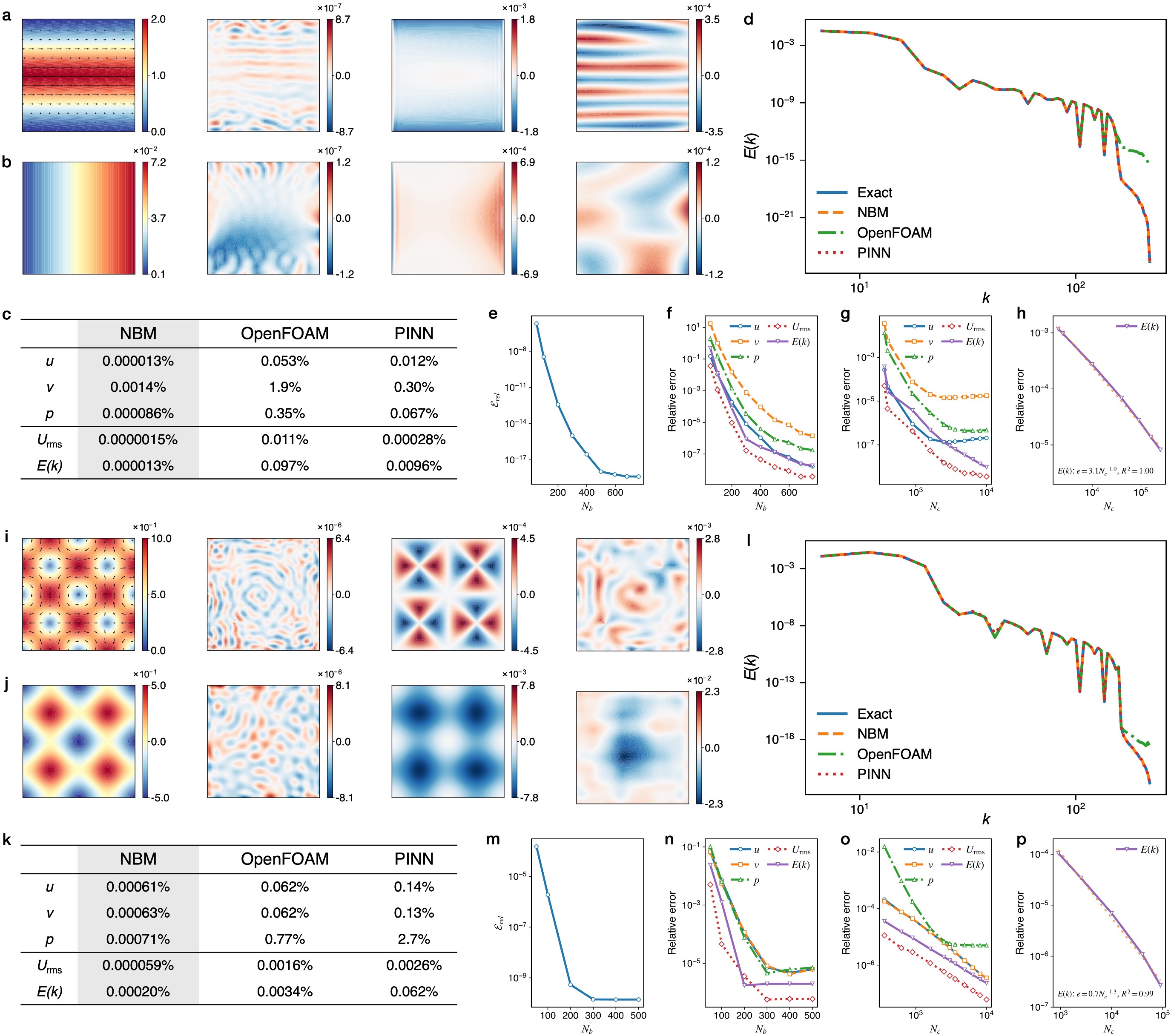}
  \caption{%
    \bfseries Kovasznay flow and Taylor--Green vortex benchmarks with analytical reference solutions.\,
    \bfseries a--h,
    \normalfont Kovasznay flow.
    \bfseries a,b,
    \normalfont Speed field with velocity quivers and pressure field, followed by absolute error fields w.r.t.\ the analytical solution for NBM, OpenFOAM, and PINN.
    \bfseries c,
    \normalfont Relative errors.
    \bfseries d,
    \normalfont Energy spectrum.
    \bfseries e,
    \normalfont Residual vs. \(N_b\) (\(N_c=50\times50\)).
    \bfseries f,
    \normalfont Relative errors vs. \(N_b\) (\(N_c=50\times50\)).
    \bfseries g,
    \normalfont Relative errors vs. \(N_c\) (\(N_b=500\)).
    \bfseries h,
    \normalfont Relative error of \(E(k)\) vs.\ \(N_c\) (\(N_b=500\)).
    \bfseries i--p,
    \normalfont Taylor--Green vortex.
    \bfseries i,j,
    \normalfont Speed field with velocity quivers and pressure field, followed by absolute error fields w.r.t.\ the analytical solution for NBM, OpenFOAM, and PINN.
    \bfseries k,
    \normalfont Relative errors.
    \bfseries l,
    \normalfont Energy spectrum.
    \bfseries m,
    \normalfont Residual \(\varepsilon_{\rm rel}\) vs.\ \(N_b\), (\(N_c=50\times50\)).
    \bfseries n,
    \normalfont Relative errors vs.\ \(N_b\) (\(N_c=50\times50\)).
    \bfseries o,
    \normalfont Relative errors vs.\ \(N_c\) (\(N_b=500\)).
    \bfseries p,
    \normalfont Relative error of \(E(k)\) vs.\ \(N_c\)) (\(N_b=500\)).
    \bfseries Note:\normalfont\ Kovasznay flow: $Re=500$. Taylor--Green vortex: $Re=1000$; results are at \(t=0.01\). See further details in Methods and Supplementary Information sections.
  }
  \label{fig:kovasznay-taylor-green}
\end{figure}
\clearpage
\begin{figure}[h!]
  \centering
  \includegraphics[width=\linewidth]{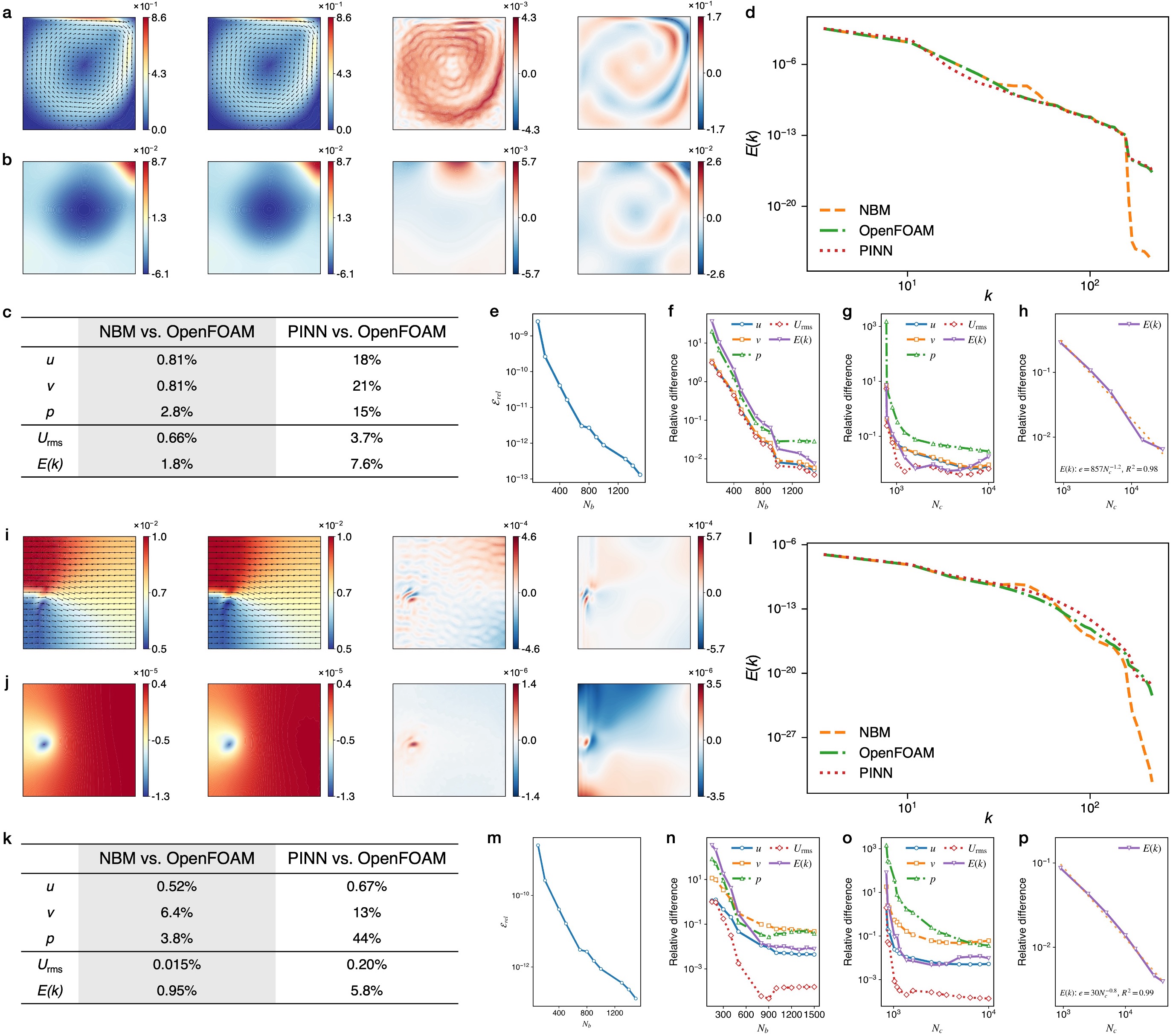}
  \caption{%
    \bfseries Lid-driven cavity and shear-layer flow benchmarks.\,
    \bfseries a--h,
    \normalfont Lid-driven cavity.
    \bfseries a,b,
    \normalfont Speed field with velocity quivers and pressure field, followed by OpenFOAM reference, NBM solution, absolute difference between NBM and OpenFOAM, and absolute difference between PINN and OpenFOAM.
    \bfseries c,
    \normalfont Relative differences using OpenFOAM as reference.
    \bfseries d,
    \normalfont Energy spectrum.
    \bfseries e,
    \normalfont Residual vs. \(N_b\) (\(N_c=100\times100\)).
    \bfseries f,
    \normalfont Relative differences vs. \(N_b\) (\(N_c=100\times100\)).
    \bfseries g,
    \normalfont Relative differences vs. \(N_c\) (\(N_b=1000\)).
    \bfseries h,
    \normalfont Relative difference of \(E(k)\) vs.\ \(N_c\) (\(N_b=800\)).
    \bfseries i--p,
    \normalfont Shear-layer flow.
    \bfseries i,j,
    \normalfont Speed field with velocity quivers and pressure field, followed by OpenFOAM reference, NBM solution, absolute difference between NBM and OpenFOAM, and absolute difference between PINN and OpenFOAM.
    \bfseries k,
    \normalfont Relative differences using OpenFOAM as reference.
    \bfseries l,
    \normalfont Energy spectrum.
    \bfseries m,
    \normalfont Residual vs.\ \(N_b\) (\(N_c=100\times100\)).
    \bfseries n,
    \normalfont Relative differences vs. \(N_b\) (\(N_c=100\times100\)).
    \bfseries o,
    \normalfont Relative differences vs. \(N_c\) (\(N_b=1000\)).
    \bfseries p,
    \normalfont Relative difference of \(E(k)\) vs. \(N_c\) (\(N_b=800\)).
    \bfseries Note:\normalfont\ Lid-driven cavity: $Re=1000$. Shear-layer flow: $Re \sim 10^{4}$; results are at \(t=30\). See further details in Methods and Supplementary Information sections.
  }
  \label{fig:cavity-shear-layer}
\end{figure}

\vspace{-0pt}
\subsection*{NBM enables residual-certified operator learning for parametric dynamics}
\vspace{-10pt}
\textbf{NBO-RC achieves substantially higher accuracy.}
We now assess parametric operator learning on Kolmogorov flow parameterized by forcing amplitude \(\chi\) and Taylor--Couette swirling flow by inner-cylinder rotation rate \(\Omega_i\) (Figs.~\ref{fig:kolmogorov} and \ref{fig:taylor_couette}); the latter is more practically relevant, combining rotating boundaries, cylindrical geometry, and a parameter range spanning the onset of Taylor-vortex instability.
We compare NBO-RC against its data-driven counterpart NBO-DD and DeepONet (DeepO) baselines trained with continuous-time physics-informed (DeepO-PI-CT), discrete-time physics-informed (DeepO-PI-DT), and data-driven (DeepO-DD) objectives.
NBO-RC achieves sub-percent velocity errors substantially below the baselines (Fig.~\ref{fig:kolmogorov}a,b; Fig.~\ref{fig:taylor_couette}a) and closely tracks NBM and OpenFOAM energy spectra, whereas the baselines show markedly larger spectral deviations (Fig.~\ref{fig:kolmogorov}c; Fig.~\ref{fig:taylor_couette}b).
DeepO-PI-CT fails to produce a meaningful Kolmogorov-flow solution and is substantially more expensive to train, so it is omitted from the Taylor--Couette study.

\begin{figure}[b!]
  \centering
  \includegraphics[width=\linewidth]{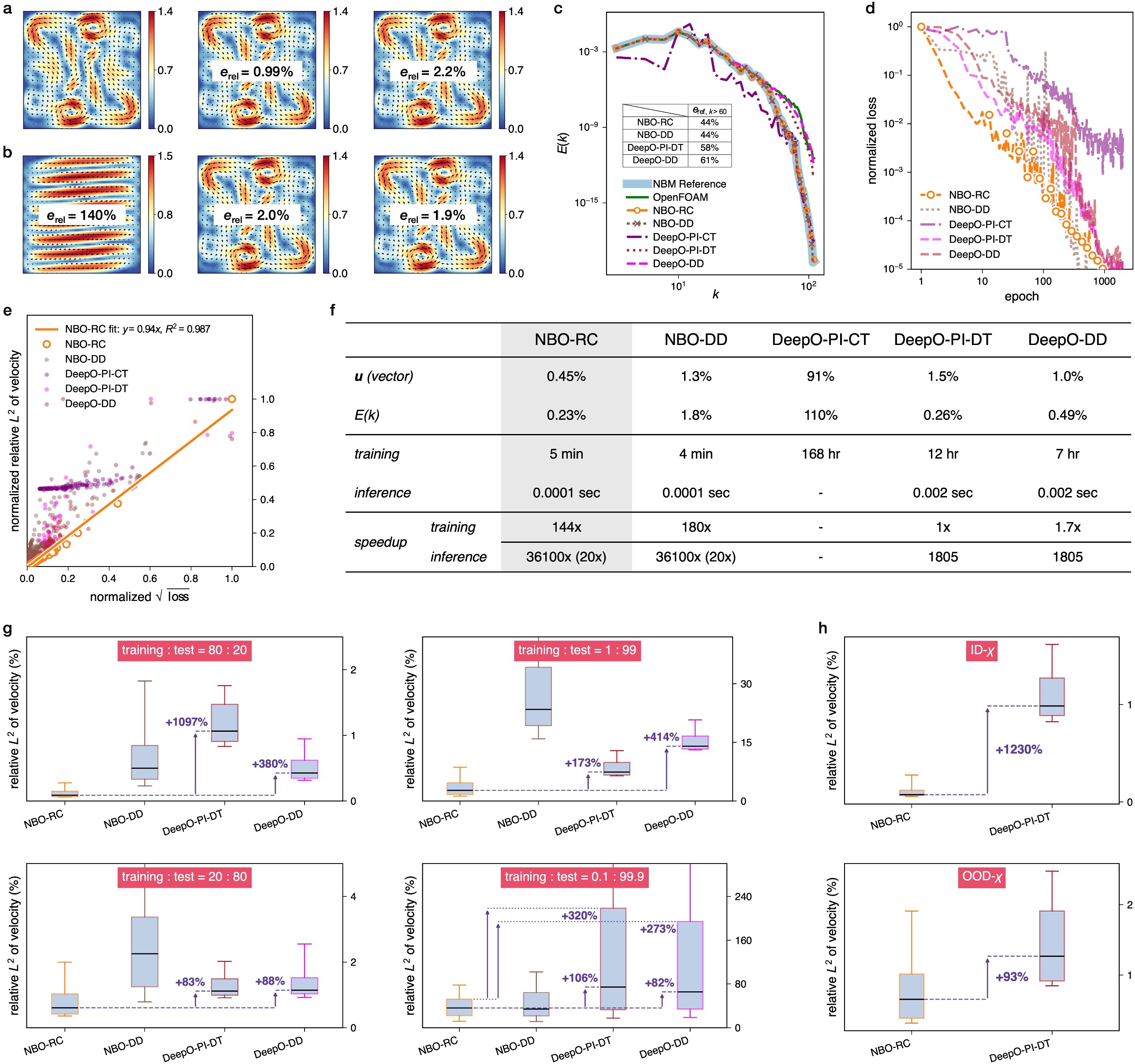}
  \caption{%
    \bfseries Residual-certified neural basis operator learning for parametric Kolmogorov flow.\,
    \bfseries a,
    \normalfont Representative speed fields with velocity quivers at \(t=3\) for NBM reference, NBO-RC, and NBO-DD.
    \bfseries b,
    \normalfont Corresponding results of DeepO-PI-CT, DeepO-PI-DT, and DeepO-DD.
    In \textbf{a,b}, \(e_{\rm rel}:\) relative \(L^2\) velocity error.
    \bfseries c,
    \normalfont Kinetic-energy spectra \(E(k)\) of NBM reference, OpenFOAM, and learned operators.
    The inset shows relative \(L^2\) errors of the high-pass-filtered velocity field for \(k>60\), using OpenFOAM as reference.
    \bfseries d,
    \normalfont Normalized training-loss histories.
    \bfseries e,
    \normalfont Normalized relative \(L^2\) velocity error vs.\ normalized square root of the training loss with linear fit for NBO-RC.
    \bfseries f,
    \normalfont Relative errors of the velocity field and energy spectrum, with training and inference timing. Training speedups are relative to DeepO-PI-DT, while inference speedups are relative to OpenFOAM; parenthetical values compare NBO with DeepO.
    \bfseries g,
    \normalfont Test relative \(L^2\) velocity-error distributions under random train--test splits of \(\zeta_{i,j}=(\chi_i,t_j)\), from \(80{:}20\) to \(0.1{:}99.9\).
    \bfseries h,
    \normalfont Relative \(L^2\) velocity-error distributions for in-distribution (ID) and out-of-distribution (OOD) magnitude values.
    \bfseries Note:\normalfont\ 
    See Methods and Supplementary Information Sections for further details.
  }
  \label{fig:kolmogorov}
\end{figure}

\begin{figure}[b!]
  \centering
  \includegraphics[width=\linewidth]{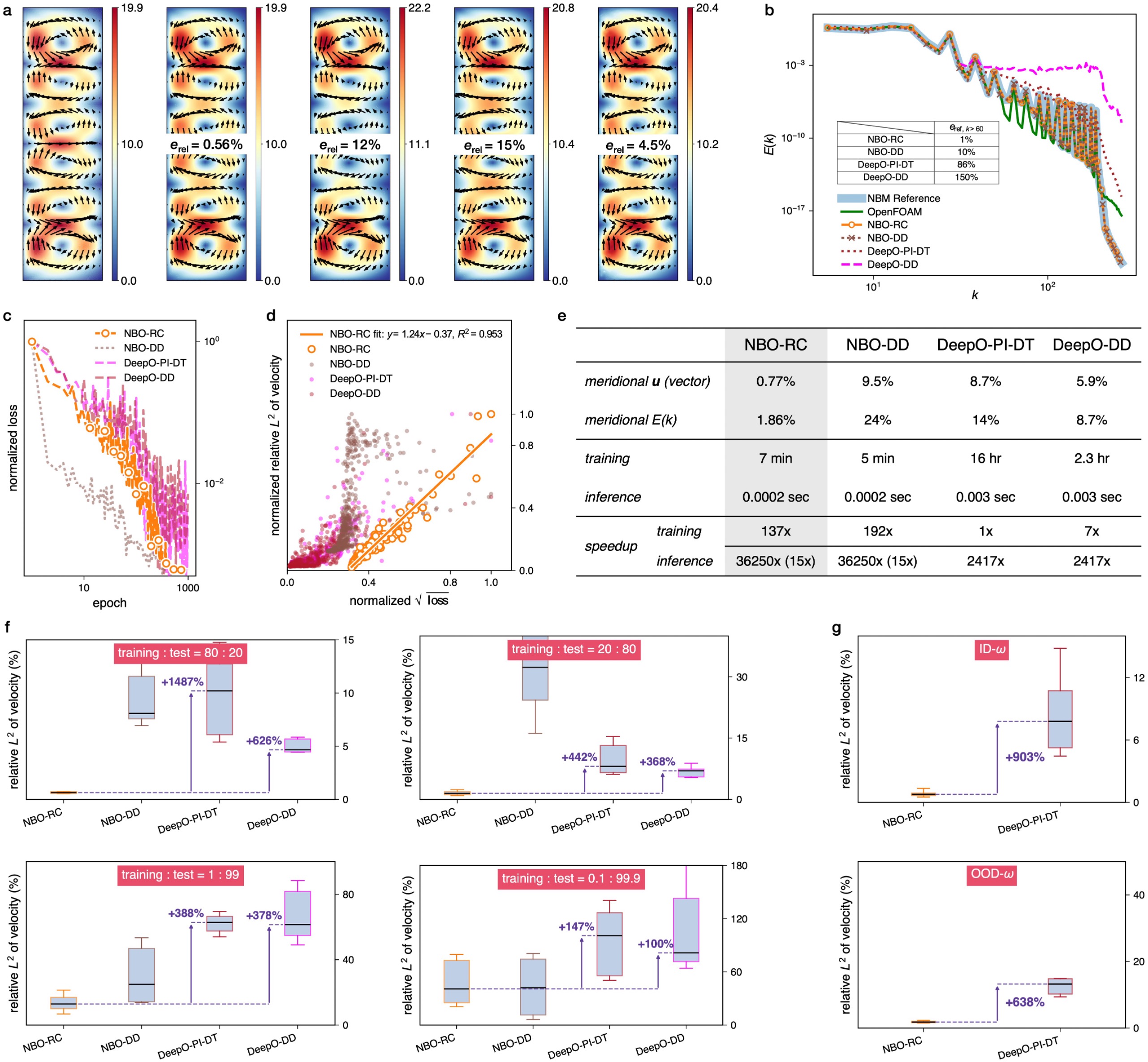}
  \caption{%
    \bfseries Residual-certified neural basis operator learning for parametric Taylor--Couette flow.\,
    \bfseries a,
    \normalfont Representative meridional-speed fields with velocity quivers at \(t=0.2\). Left to right: NBM reference, NBO-RC, NBO-DD, DeepO-PI-DT, and DeepO-DD; \(e_{\rm rel}:\) relative \(L^2\) velocity error.
    \bfseries b,
    \normalfont Meridional kinetic-energy spectra \(E(k)\) of NBM reference, OpenFOAM, and learned operators. The inset shows relative \(L^2\) errors of the high-pass-filtered meridional velocity for \(k>60\), using OpenFOAM as reference.
    \bfseries c,
    \normalfont Normalized training-loss histories.
    \bfseries d,
    \normalfont Normalized relative \(L^2\) velocity error vs.\ normalized square root of training loss with linear fit for NBO-RC.
    \bfseries e,
    \normalfont Relative errors of meridional velocity and its energy spectrum, with training and inference timing and corresponding speedups. Training speedups are relative to DeepO-PI-DT, while inference speedups are relative to OpenFOAM; parenthetical values compare NBO with DeepO.
    \bfseries f,
    \normalfont Test relative \(L^2\) meridional velocity-error distributions under random train--test splits of \(\zeta_{i,j}=(\Omega_i,t_j)\), from \(80{:}20\) to \(0.1{:}99.9\).
    \bfseries g,
    \normalfont Relative \(L^2\) meridional velocity-error distributions for in-distribution (ID) and out-of-distribution (OOD) rotation rates.
    \bfseries Note:\normalfont\ See Methods and Supplementary Information Section for further details.
  }
  \label{fig:taylor_couette}
\end{figure}

\vspace{3pt}
\textbf{NBO-RC combines error-aligned learning with superior scalability.}
NBO-DD replaces only NBO-RC's residual-based loss \eqref{eq:nbm_ol_loss} with coefficient-space MSE, isolating residual certification as the source of NBO-RC's advantage.
The normalized NBO-RC loss is nearly linear in normalized relative velocity error, whereas other losses show no comparable monotone relation(Fig.~\ref{fig:kolmogorov}e; Fig.~\ref{fig:taylor_couette}d).
Thus, \eqref{eq:nbm_ol_loss} provides an error-aligned training signal, a property sought in operator-learning but rarely achieved by existing approaches. 
Fixed neural bases also substantially reduce computational cost: spatial basis evaluations and PDE-dependent residual quantities can be decoupled and precomputed. NBO-RC thus trains only the parameter–time-to-coefficient network with a low-dimensional quadratic coefficient-space residual, whereas DeepO-PI-DT reevaluates its spatial reconstruction network and PDE-required spatial derivatives at every training update.
On a H100 GPU, NBO-RC trains approximately \(144\times\) faster for Kolmogorov flow and \(137\times\) faster for Taylor--Couette flow, while NBO inference is about \(20\times\) faster than DeepO models (Fig.~\ref{fig:kolmogorov}f; Fig.~\ref{fig:taylor_couette}e). 
These easured speedups are below the \(>10^{3}\times\) work ratio implied by an operation-count analysis (see Methods) because the much larger DeepO-PI-DT workload exploits H100 parallelism more efficiently, whereas lighter NBO-RC updates provide less work to saturate the GPU.

\vspace{3pt}
\textbf{NBO-RC generalizes robustly under sparse sampling and parameter shift.}
Across random train--test splits from \(80{:}20\) to \(0.1{:}99.9\), NBO-RC consistently yields substantially lower and generally less dispersed test errors than DeepO (Fig.~\ref{fig:kolmogorov}g; Fig.~\ref{fig:taylor_couette}f). 
Under parameter shift, it maintains markedly lower errors at both in-distribution and out-of-distribution \(\chi\) and \(\Omega_i\) values (Fig.~\ref{fig:kolmogorov}h; Fig.~\ref{fig:taylor_couette}g), demonstrating robust generalization under sparse sampling and to unseen parameters.

\vspace{-0pt}
\subsection*{Residual certification shapes velocity-aligned operator-network training}
\vspace{-10pt}
We trace how residual certification shapes operator-network training
by carrying the residual-based metric from neural basis coefficient
space into network parameter space (see Supplementary Section).
Figure~\ref{fig:mechanistic_probe}a compares the sample-averaged residual Gram matrix
\(
\mathbf{G}_{{\rm R},{\rm mean}}
=
N_s^{-1}\sum_{i=1}^{N_s}\mathbf{G}_{{\rm R},i}
\)
with the velocity-energy matrix \(\mathbf{E}\) over 50 random unit
coefficient-space perturbation directions \(\mathbf{q}\).
The nearly perfect linear and rank correlations between
\(\mathbf{q}^{T}\mathbf{G}_{{\rm R},{\rm mean}}\mathbf{q}\)
and
\(\mathbf{q}^{T}\mathbf{E}\mathbf{q}\)
show greater residual sensitivity corresponds to greater velocity-energy error, confirming the residual-certified loss encodes a velocity-aligned coefficient-space metric.
We then trace this metric through the operator-network Jacobian
into its parameter-space training geometry.
Figure~\ref{fig:mechanistic_probe}b measures the full-Hessian deviation from its Gauss--Newton component using
\(
\Delta_{\rm GN}(\mathbf{v})
=
\frac{
\left\|
(\mathbf{H}-\mathbf{H}_{\rm GN})\mathbf{v}
\right\|_2
}{
\left\|
\mathbf{H}_{\rm GN}\mathbf{v}
\right\|_2
},
\)
over 20 random unit operator-network parameter-space directions \(\mathbf{v}\), reporting the \(95\)th percentile at each selected epoch.
The rapid decrease of \(\Delta_{\rm GN}^{\rm p95}\) and final agreement of full and Gauss--Newton directional curvatures show that both NBO-RC and NBO-DD develop local training geometry; this dominance alone cannot explain NBO-RC's advantage. 
Controlled metric swaps under the same Jacobian (Fig.~\ref{fig:mechanistic_probe}c,d) isolate the residual metric: replacing \(\mathbf{G}_{\rm R}\) by MSE for NBO-RC removes the near-perfect velocity-energy-curvature correlation, whereas applying \(\mathbf{G}_{\rm R}\) to NBO-DD restores it.
Thus velocity alignment is determined by the residual-certified metric, not the network architecture, trained parameters, or Jacobian alone. 
Closely matched residual and velocity-energy contours in stiff and weak parameter-space directions (Fig.~\ref{fig:mechanistic_probe}e), together with preserved velocity-energy rank ordering under the residual metric but not MSE (Fig.~\ref{fig:mechanistic_probe}f), further support this conclusion.
Residual certification therefore transfers velocity-relevant NBM-residual structure from coefficient space into operator-network training geometry, directing learning toward enhanced velocity prediction.

\begin{figure}[t!]
  \centering
  \includegraphics[width=\linewidth]{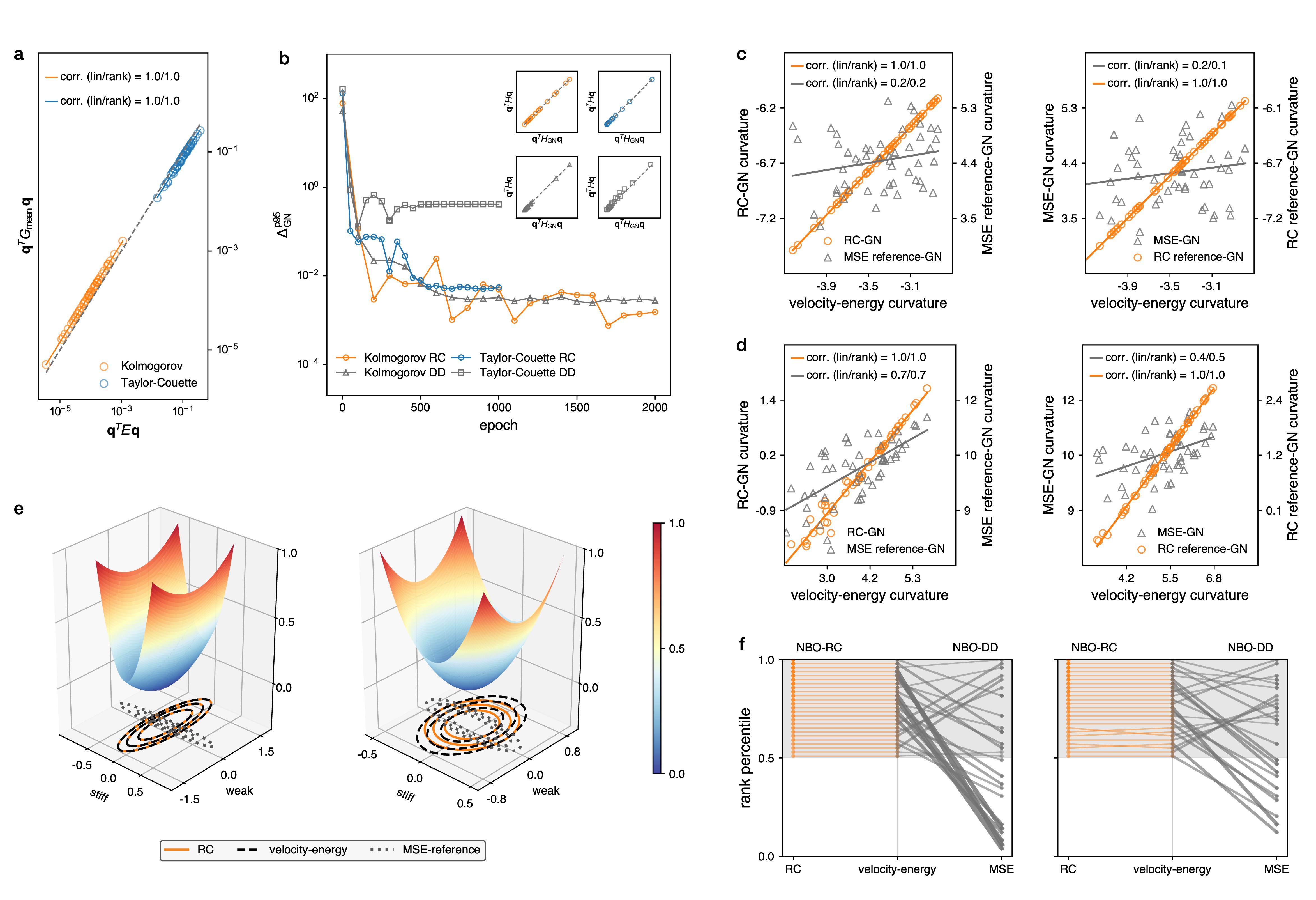}
  \caption{%
    \bfseries Mechanistic probe of residual-certified neural basis operator learning.\,
    \bfseries a,
    \normalfont Directional quadratic forms of the sample-averaged residual Gram \(\mathbf G_{{\rm R},{\rm mean}}\) and the velocity-energy metric \(\mathbf E\), evaluated along 50 random unit directions \(\mathbf q\) in the neural basis coefficient space.
    \bfseries b,
    \normalfont Evolution of the \(95\)th-percentile relative discrepancy
    \(\Delta_{\mathrm{GN}}(\mathbf v)
    =\|(\mathbf H-\mathbf H_{\mathrm{GN}})\mathbf v\|_2/
    \|\mathbf H_{\mathrm{GN}}\mathbf v\|_2\)
    between the full Hessian \(\mathbf H\) and its Gauss--Newton component \(\mathbf H_{\mathrm{GN}}\), evaluated over 20 random unit directions \(\mathbf v\) in the space of trainable network parameters.
    Insets compare the final directional curvatures
    \(\mathbf v^{T}\mathbf H\mathbf v\) and
    \(\mathbf v^{T}\mathbf H_{\mathrm{GN}}\mathbf v\).
    \bfseries c,d,
    \normalfont Velocity-energy curvature vs.\ loss-induced Gauss--Newton (GN) curvature for Kolmogorov flow (\textbf{c}) and Taylor--Couette flow (\textbf{d}).
    In each panel, the left evaluates RC-GN and MSE-reference-GN curvatures using the NBO-RC network Jacobian, whereas the right plot evaluates MSE-GN and RC-reference-GN curvatures using the NBO-DD network Jacobian.
    \bfseries e,
    \normalfont Normalized local quadratic basins in selected stiff and weak network-parameter directions for Kolmogorov flow (left) and Taylor--Couette flow (right), comparing the RC, velocity-energy, and MSE-reference metrics.
    \bfseries f,
    \normalfont Rank-percentile correspondence for directions in the upper half of velocity-energy curvature for Kolmogorov flow (left) and Taylor--Couette flow (right), comparing RC with velocity energy for NBO-RC and velocity energy with MSE for NBO-DD.
    \bfseries Note:\normalfont\ ``corr.\ (lin/rank)'' reports the Pearson linear correlation and Spearman rank correlation, respectively. See Methods and Supplementary Information Sections for the Hessian decomposition, metric construction, and probe procedure.
  }
  \label{fig:mechanistic_probe}
\end{figure}

\section*{Discussion}
\vspace{-15pt}
This work presents NBM as a new framework for solving and learning incompressible Navier--Stokes dynamics. 
Embedding inompressibility in the neural discretization and eliminating pressure from the primary solve yields a velocity-only formulation in neural basis coefficient space. 
We provide a triple-resolution convergence analysis and a posteriori error estimate linking residual to velocity error.
The certified residual also serves as a physically meaningful coefficient-space training metric for operator learning, avoiding repeated spatial computation while achieving substantially improved accuracy, generalization and scalability over DeepONets. 

Several directions can broaden and strengthen NBM.
Localized non-smoothness and nonconvex or complex geometries motivate localized formulations based on domain decomposition, partition-of-unity, or discontinuous-Galerkin-style local--global formulations. 
These could also enable adaptive enrichment by concentrating neural bases and residual resolution where needed.
As these systems grow, scalable large-scale least-squares strategies, including randomization, preconditioning, and high-throughput parallel processing will become increasingly important.
NBM also provides a promising route toward addressing a central scalability bottleneck of existing neural operators: DeepONet, FNO, and Transformer-based architectures remain explicitly coupled to spatial resolution, so computational and memory costs rise rapidly with spatial degrees of freedom or evaluation points, constraining tractable operator-learning scale~\cite{subedi2025operator,li2022transformer,hao2023gnot,cao2021choose}.
By contrast, NBO's coefficient-space construction can support more expressive operator architectures without restoring this dependence.
An important next step is therefore an NBM-powered Transformer-based operator for high-dimensional parameterizations, potentially capturing complex input interactions while avoiding severe complexity growth with spatial resolution and, more broadly, enabling multi-operator learning across wider PDE families.

\section*{Methods}
\vspace{-15pt}
In this section, we first define the six benchmark problems considered in the Results Section, expressed in nondimensional form for clarity. 
We then briefly visit the control-volume formulation and show how classical temporal- and nonlinear convective-term handling schemes integrate into the velocity-only NBM construction for incompressible Newtonian flow.
Finally, we provide concise descriptions of the baseline neural approaches
used for comparison and analyze the computational complexity and scalability
of operator learning.
Further implementation details and mechanistic-analysis procedure underlying Fig.~\ref{fig:mechanistic_probe} are provided in the Supplementary Section.

\vspace{-5pt}
\subsection*{Single-instance benchmark problems}
\vspace{-10pt}
\textbf{Kovasznay flow.} 
The Kovasznay flow is a steady two-dimensional solution of  \eqref{eq:NS} and \eqref{eq:continuity} with \(\partial_t\mathbf{u}=\mathbf{0}\) and \(\mathbf{F}=\mathbf{0}\). We consider the Cartesian domain \(\Omega=[0,1]^2\) with kinematic viscosity \(\nu=0.002\), corresponding to \(Re=500\). The exact velocity and kinematic pressure are
\[ u(x,y)=1-e^{\lambda x}\cos(2\pi y), \qquad v(x,y)=\frac{\lambda}{2\pi}e^{\lambda x}\sin(2\pi y), \qquad p(x,y)=\frac{1}{2}\left(1-e^{2\lambda x}\right), \]
where
\(\lambda = \frac{Re}{2} - \sqrt{\left(\frac{Re}{2}\right)^2 + 4\pi^2}.\)
Dirichlet boundary conditions are imposed on all boundaries from the exact solution.

\textbf{Taylor-Green vortex.}
The Taylor-Green vortex considered here has an unsteady decaying-vortex exact solution of \eqref{eq:NS} and \eqref{eq:continuity} in Cartesian domain \([0,1]^2\) (\(\mathbf F=\mathbf 0\)).
The nondimensional kinematic viscosity is \(\nu=0.001\), giving \(Re=1000\).
The exact velocity \((u,v)\) and kinematic pressure \(p\) are 
\[ u(x,y,t)=-\cos(2\pi x)\sin(2\pi y)e^{-8\pi^2\nu t}, \qquad v(x,y,t)=\sin(2\pi x)\cos(2\pi y)e^{-8\pi^2\nu t}, \] \[ p(x,y,t)=-\frac{1}{4}\left(\cos(4\pi x)+\cos(4\pi y)\right)e^{-16\pi^2\nu t}. \]
NBM and PINN impose Dirichlet boundary conditions on all boundaries from the exact solution.
OpenFOAM uses the built-in periodic boundary setup instead.
All methods use the same initial condition from the exact solution at \(t=0\).

\textbf{Lid-driven cavity.}
The lid-driven cavity is a stationary benchmark without an analytical solution. 
We consider \(\Omega=[-1,1]^2\), with \(\mathbf{F}=\mathbf{0}\), and prescribe a horizontal velocity profile on the top boundary \(y=1\): 
\(
u(x,1)=U_{\mathrm{lid}}\sin^2\!\left(\frac{\pi(x+1)}{2}\right), \, v(x,1)=0. \)
No-slip conditions are imposed on the remaining walls.
Zero pressure is enforced at one arbitrary point for pressure solution uniqueness.
For sufficiently large \(Re\), steady solutions may not exist; therefore, we here take \(U_{\mathrm{lid}}=1\) and \(\nu=0.002\), giving \(Re=1000\).

\textbf{Shear-layer flow.}
The shear-layer flow is an unsteady benchmark problem without an exact solution (\(\mathbf F=\mathbf 0\)).
It is driven by a perturbed inlet shear layer and is associated with Kelvin--Helmholtz instability phenomena.
We consider a Cartesian domain \([-1,1]^2\), with nondimensional kinematic viscosity \(\nu=10^{-6}\), giving a Reynolds number on the order of \(10^4\).
The initial velocity is prescribed as \(\mathbf u(x,y,0)=(0.75,0,0)\).
At the inlet \(x=-1\), the velocity is
\[ u(-1,y,t)=0.75+0.25\tanh\!\left(\frac{y}{0.08}\right)+0.05\sin(2\pi t)\exp\!\left[-\left(\frac{y}{0.08}\right)^2\right], \qquad v(-1,y,t)=0. \]
At the outlet \(x=1\), an open outflow condition is imposed using zero velocity gradient and zero kinematic pressure.
Slip boundary conditions are imposed on the top and bottom boundaries \(y=\pm 1\).

\vspace{-5pt}
\subsection*{Parametric benchmark problems}
\vspace{-10pt}
\textbf{Kolmogorov flow.}
The Kolmogorov flow is driven by a periodic force
\(
\mathbf{F}(y;\chi)
=
\left(
\chi\sin(4\pi y),0
\right),
\)
where \(\chi\) denotes the forcing amplitude. We consider \(\Omega=[-1,1]^2\) with \(\nu=10^{-3}\). The fluid starts from rest,
\(
\mathbf{u}(\mathbf{x},0)
=
\mathbf{0},
\)
and no-slip conditions are imposed on \(\partial\Omega\). The operator input is
\(
\boldsymbol{\zeta}
=
(\chi,t),
\,
\chi\in[0.5,1.5],
\,
t\in[0,3].
\)
The resulting parameter--time domain contains strongly nonlinear vortical dynamics generated by different forcing amplitudes.

\textbf{Taylor-Couette swirling flow.}
Taylor--Couette swirling flow describes the motion between two coaxial cylinders driven by cylinder rotation (\(\mathbf F=\mathbf 0\)).
We use an axisymmetric formulation in cylindrical coordinates \((r,\theta,z)\), with velocity
\(
\mathbf{u}
=
(u_r,u_\theta,u_z).
\)
The meridional computational domain is
\(
(r,z)
\in
[1,1.5]\times[-1.5,1.5],
\)
and the kinematic viscosity is \(\nu=1\). The incompressibility constraint becomes
\(
\frac{1}{r}
\frac{\partial(ru_r)}{\partial r}
+
\frac{\partial u_z}{\partial z}
=
0,
\)
with the corresponding axisymmetric momentum equations obtained from \eqref{eq:NS}. No penetration is imposed on the solid boundaries, with \(u_r=u_z=0\). The azimuthal wall velocities are
\(
u_\theta
=
\Omega_i r
\,\,
\text{at }r=1
\)
and
\(
u_\theta
=
\Omega_o r
\,\,
\text{at }r=1.5.
\)
We set \(\Omega_o=0\) and use the inner-cylinder rotation rate as the operator parameter:
\(
\boldsymbol{\zeta}
=
(\Omega_i,t),
\,\,
\Omega_i\in[100,200],
\,\,
t\in[0,0.2].
\)
For this nondimensional configuration, Taylor-vortex instability develops near \(\Omega_i\approx160\). 
The selected parameter interval therefore spans the transition from circular Couette flow to axisymmetric Taylor-vortex flow while remaining below the wavy-vortex regime. 

\vspace{-5pt}
\subsection*{Control-volume formulation and NBM integration}
\vspace{-10pt}
\textbf{Control-volume formulation.} 
Starting from the conservative momentum balance in \eqref{eq:NS}, integration over a control-volume \(K\) gives the semi-discrete residual used in standard finite-volume formulations:
\begin{equation}
\frac{d}{dt}\int_K \mathbf{u}\,dV
+
\sum_{f\in\partial K} F_f\mathbf{u}_f
-
\nu\sum_{f\in\partial K}
(\nabla\mathbf{u})_f\cdot\mathbf{n}_f\,|f|
+
\sum_{f\in\partial K}
p_f\mathbf{n}_f\,|f|
-
\int_K\mathbf{F}\,dV
=
\mathbf{0},
\label{eq:cv}
\end{equation}
where \(f\) denotes a face of \(K\), \(\mathbf{n}_f\) is the outward unit normal, \(|f|\) is the face measure, $F_f=(\mathbf{u}_f\cdot\mathbf{n}_f)|f|$ is the face flux, $(\nabla\mathbf{u})_f\cdot\mathbf{n}_f$ is the normal velocity-gradient flux, and $p_f$ is the face pressure.
The OpenFOAM solutions used as benchmark references are based on this formulation.

\textbf{NBM integration.}
In NBM, the VCV construction shown in Fig.~\ref{fig:overview}b allows the cell values, face values, and face gradients appearing in \eqref{eq:cv} to be evaluated analytically from the neural basis functions. 
This is a very advantageous feature of NBM, since in conventional finite-volume approaches certain of these quantities, depending on gridding design, are typically obtained through local interpolation/reconstruction and gradient approximation, procedures whose accuracy and stability often depend sensitively on grid arrangement and quality~\cite{ferziger_computational_2020}.
Because the NBM residual retains the control-volume structure of \eqref{eq:cv}, established temporal and convective discretization schemes~\cite{ferziger_computational_2020}developed for finite-volume methods can be incorporated directly into NBM.
For unsteady problems, we adopt the second-order backward differentiation scheme (BDF2) for temporal discretization with second-order accuracy such that
\(
\left.
\frac{d}{dt}
\int_K\mathbf{u}\,dV
\right|^{n+1}
\approx
|K|
\frac{
3\mathbf{u}_K^{n+1}
-
4\mathbf{u}_K^n
+
\mathbf{u}_K^{n-1}
}{
2\Delta t
},
\)
where \(|K|\) is the control-volume measure. 
The convection term is linearized by extrapolating the convecting velocity from the two preceding time levels to maintain the same level of accuracy:
\(
\widehat{F}_f^{\,n+1}
:=
\left[
\left(
2\mathbf{u}_f^{n}
-
\mathbf{u}_f^{n-1}
\right)
\cdot
\mathbf{n}_f
\right]|f|.
\)
Here the hat indicates that $\widehat{F}_f^{\,n+1}$ is an extrapolated flux evaluated from previous time levels.
The sign of $\widehat{F}_f^{\,n+1}$ determines the upwind direction, while the transported value $\mathbf{u}_f^{n+1}$ is reconstructed at the face using a convective scheme such as LUST.
Therefore, the discrete convection contribution is written as
\(\sum_{f\in\partial K}\widehat{F}_f^{\,n+1}\mathbf{u}_f^{n+1}.\)
After substituting the NBM expansion into the control-volume residual, the cell values, face values, and face gradients at \(t^{n+1}\) are linear functions of the unknown velocity-coefficient vector \(\boldsymbol{\alpha}^{n+1}\), whereas the extrapolated fluxes \(\widehat{F}_f^{\,n+1}\) are determined from previously computed states. 
The transient, diffusive, and linearized convective contributions can therefore be assembled into a linear least-squares residual in the velocity coefficients. 
The pressure contribution is eliminated using the velocity-only projection.  
The resulting projected control-volume residual, together with the boundary-condition residuals, is then used to obtain \(\boldsymbol{\alpha}^{n+1}\).

\vspace{-5pt}
\subsection*{Physics-informed neural networks}
\vspace{-10pt}
PINNs are used as neural baselines for the single-instance benchmarks. A fully connected neural network represents
\(
\mathcal{N}_{\boldsymbol{\eta}}
(x,y[,t])
=
\left(
u_{\boldsymbol{\eta}},
v_{\boldsymbol{\eta}},
p_{\boldsymbol{\eta}}
\right),
\)
where \(t\) is included as an input for unsteady problems. Spatial and temporal derivatives are obtained using automatic differentiation.
For the considered examples, the interior residuals are
\(
r_u
=
\partial_tu
+
u\partial_xu
+
v\partial_yu
+
\partial_xp
-
\nu
\left(
\partial_{xx}u
+
\partial_{yy}u
\right),
\)
\(
r_v
=
\partial_tv
+
u\partial_xv
+
v\partial_yv
+
\partial_yp
-
\nu
\left(
\partial_{xx}v
+
\partial_{yy}v
\right),
\)
\(
r_c
=
\partial_xu
+
\partial_yv,
\)
with the temporal terms omitted for steady problems. 
The PINN objective combines the governing-equation, boundary-condition, initial-condition and pressure-gauge losses:
\(
\mathcal{L}_{\rm PINN}
=
\mathcal{L}_{\rm PDE}
+
\mathcal{L}_{\rm BC}
+
\mathcal{L}_{\rm IC}
+
\mathcal{L}_{\rm gauge}.
\)
We note that, for the lid-driven cavity at \(Re=1000\), the standard PINN did not converge to a physically meaningful steady solution under the adopted training protocol, consistent with previously reported difficulties for this problem~\cite{satyadharma2024assessing}. 
We thus then add a sparse velocity-anchor term from the OpenFOAM reference solution, however, making this PINN-based neural solver no longer data-free.

\vspace{-5pt}
\subsection*{DeepONet-based neural operators}
\vspace{-10pt}
\textbf{DeepO.} 
For the parametric benchmarks, we use DeepONet-based neural operators (DeepO) as baseline operator learning methods \cite{karumuri2026physics}. The DeepONet directly represents the parameter--time-to-field map through a branch--trunk architecture. For a parameter--time instance $\zeta=(\xi,t)$ and spatial coordinate $\mathbf{x}$, the branch network maps $\zeta$ to parameter--time-dependent coefficients, while the trunk network maps $\mathbf{x}$ to spatial features. The corresponding outputs are combined through inner products to reconstruct the flow field $\mathbf{s}_{\boldsymbol{\eta}}=(u_{\boldsymbol{\eta}},v_{\boldsymbol{\eta}},p_{\boldsymbol{\eta}})$, yielding $\mathbf{s}_{\boldsymbol{\eta}}(\mathbf{x};\zeta)=\mathrm{DeepO}_{\boldsymbol{\eta}}(\zeta,\mathbf{x})$, where $\boldsymbol{\eta}$ denotes the trainable network parameters.

\textbf{Data-driven vs. physics-informed.}
The same architecture is trained in data-driven (DeepO-DD) and physics-informed (DeepO-PI) settings by changing the loss weights,
\(
\mathcal{L}_{\mathrm{DeepO}} = \lambda_{\mathrm{phys}}\mathcal{L}_{\mathrm{phys}} + \lambda_{\mathrm{BC}}\mathcal{L}_{\mathrm{BC}} + \lambda_{\mathrm{data}}\mathcal{L}_{\mathrm{data}},
\label{eq:lno-loss}
\)
where $\mathcal{L}_{\mathrm{phys}}$ is the discrete-time residual loss based on \eqref{eq:NS}-\eqref{eq:continuity}, $\mathcal{L}_{\mathrm{BC}}$ enforces the prescribed boundary conditions, and $\mathcal{L}_{\mathrm{data}}$ measures the mismatch between the reconstructed velocity field and the reference velocity snapshots. 
The two settings differ only in $\lambda_{\mathrm{phys}}$ and $\lambda_{\mathrm{data}}$: the data-driven setting sets $\lambda_{\mathrm{phys}}=0$, so $\mathcal{L}_{\mathrm{data}}$ is the sole training signal, while the physics-informed setting sets $\lambda_{\mathrm{data}}=0$, so training relies solely on the residual and boundary losses, without reference-snapshot supervision.
Implementation details are provided in the Supplementary Section.

\textbf{Continuous-time vs. discrete-time.}
For time dependence, we formulate a continuous-time version (DeepO-PI-CT), in which the predicted field is a continuous function of $\zeta$, the loss is formed from the continuous-time residual, and temporal derivatives such as $\partial_t\mathbf{u}$ are obtained by automatic differentiation with respect to $t$. 
This continuous-time formulation fails to produce physically meaningful solutions. 
We therefore adopt a discrete-time formulation (DeepO-PI-DT), in which the network predicts the target field at time level $t^{n+1}$: for $\zeta^{n+1}=(\xi,t^{n+1})$,
\(
\mathbf{s}_{\boldsymbol{\eta}}^{n+1}(\mathbf{x})
=
\mathrm{DeepO}_{\boldsymbol{\eta}}
\left(\zeta^{n+1},\mathbf{x}\right).
\)
Cached states from previous time levels are used only inside the loss.
Note: DeepO-DD is snapshot-based and therefore uses temporally discrete training data for time-dependent problems.

\vspace{-5pt}
\subsection*{Computational complexity and scalability analysis of operator learning}
\vspace{-10pt}
\textbf{Operation-count convention.} 
We analyze the computational complexity of NBO-RC and DeepO-PI-DT in terms of multiply--accumulate (MAC) operations associated with dense neural-network layers and the dominant matrix operations entering the training objectives.
For a dense layer with input and output dimensions \(d_{\mathrm{in}}\) and \(d_{\mathrm{out}}\), respectively, one forward evaluation requires \(d_{\mathrm{in}}d_{\mathrm{out}}\) MACs.
Elementwise activations, bias additions, and scalar reductions are omitted as lower-order contributions.
The operation counts below use the Kolmogorov-flow configuration as a representative example, while the subsequent expressions give a more general complexity of the two formulations.
Because NBO-RC and DeepO-PI-DT process the same number of parameter--time training records per epoch, their per-epoch complexity ratio is the same as their per-record complexity ratio.

\textbf{NBO-RC.}
The principal computational advantage of NBO-RC follows from separating spatial discretization from online operator-network training.
Once the neural bases are fixed, the discrete operators entering the projected residual can be assembled before training, and the squared residual for each parameter--time sample can be expressed as a quadratic function of the predicted coefficient vector.
Consequently, each online training update requires only evaluation of the parameter--time-to-coefficient network and the corresponding low-dimensional quadratic form.
For the Kolmogorov-flow example, the coefficient network maps the two-dimensional parameter--time input to \(d_{\mathrm{coef}}=300\) coefficients (after POD reduction) through four hidden layers of width \(768\).
Its forward cost is \(\mathcal{C}_{\mathrm{coef}}^{\mathrm{fwd}}=2(768)+3(768)^2+768(300)\simeq2 \times10^{6}\) MACs per training record.
Including reverse-mode differentiation through the coefficient network and evaluation and differentiation of the \(300\times300\) quadratic residual gives approximately
\(
\mathcal{C}_{\mathrm{NBO}}
\simeq
6\times10^{6}
\)
MAC-equivalent operations per training record.
More generally, taking advantage of the precomputed coefficient-space residual quantities, the online NBO-RC training complexity is
\(
\mathcal{C}_{\mathrm{NBO}}
=
O\!\left(
\mathcal{C}_{\mathrm{coef}}
+
d_{\mathrm{coef}}^{\,2}
\right),
\)
where \(\mathcal{C}_{\mathrm{coef}}\) denotes the forward-and-backward cost of the parameter--time-to-coefficient network.
Importantly, the dependence on the spatial discretization is confined to the offline basis construction and residual precomputation, which are performed once rather than repeated over training epochs.

\textbf{DeepO-PI-DT.}
DeepO-PI-DT retains a trainable spatial reconstruction network inside every optimization step.
For the architecture considered here, both the branch and trunk networks contain five hidden layers of width \(256\), with a DeepO rank \(128\).
For the Kolmogorov formulation, the branch maps the 2-dimensional parameter--time input to \(256\) coefficients corresponding to the streamfunction and pressure outputs, while the trunk maps the 2-dimensional spatial coordinate to \(128\) spatial features.
The forward costs are therefore \(\mathcal{C}_{\mathrm{branch}}^{\mathrm{fwd}}=2(256)+4(256)^2+256(256)\simeq3\times10^{5}\) and \(\mathcal{C}_{\mathrm{trunk}}^{\mathrm{fwd}}=2(256)+4(256)^2+256(128)\simeq3\times10^{5}\) MACs, respectively.
The branch is evaluated once for each parameter--time record, whereas the trunk must be evaluated over the spatial points entering the loss.
In the Kolmogorov physics-informed training configuration, \(m_{\mathrm{phy}}=512\) interior points and \(m_{\mathrm{BC}}=512\) boundary points are evaluated for each training record.
Including one additional spatial evaluation for the pressure gauge and the branch--trunk contractions, the forward field-reconstruction workload is approximately \(\mathcal{C}_{\mathrm{rec}}^{\mathrm{fwd}}=\mathcal{C}_{\mathrm{branch}}^{\mathrm{fwd}}+(m_{\mathrm{phy}}+m_{\mathrm{BC}}+1)\mathcal{C}_{\mathrm{trunk}}^{\mathrm{fwd}}+(m_{\mathrm{phy}}+m_{\mathrm{BC}})(256)\simeq3\times10^{8}\) MACs per training record even before accounting for the spatial derivatives required by the physics objective.
The dominant additional cost of DeepO-PI-DT arises from automatic differentiation.
The Navier--Stokes residual requires multiple first- and higher-order spatial derivatives of the reconstructed fields at every interior physics point, while the boundary loss also requires spatial differentiation of the streamfunction to obtain the velocity.
These derivative computations retain their computational graphs because the resulting physics-informed loss must subsequently be differentiated with respect to the network parameters.
Using the Kolmogorov derivative graph and accounting conservatively for the repeated reverse-mode differentiation and the final parameter backward gives an aggregate automatic-differentiation factor of approximately \(\kappa_{\mathrm{AD}}\simeq30\) relative to the base spatial reconstruction workload, yielding
\(
\mathcal{C}_{\mathrm{DeepO}}
\simeq
9.2\times10^{9}
\)
MAC-equivalent operations per physics-training record.
More generally, the DeepO-PI-DT training complexity can be expressed as
\(
\mathcal{C}_{\mathrm{DeepO}}
=
O\!\left(
\mathcal{C}_{\mathrm{branch}}
+
\kappa_{\mathrm{phy}}\,
m_{\mathrm{phy}}\,
\mathcal{C}_{\mathrm{trunk}}
+
\kappa_{\mathrm{BC}}\,
m_{\mathrm{BC}}\,
\mathcal{C}_{\mathrm{trunk}}
\right),
\)
where \(m_{\mathrm{phy}}\) and \(m_{\mathrm{BC}}\) are the numbers of spatial points used for the physics and boundary losses, respectively, and \(\kappa_{\mathrm{phy}}\) and \(\kappa_{\mathrm{BC}}\) account for the automatic-differentiation work required by the corresponding residuals.
The precise values of these factors depend on the governing equations and field representation, but the structural dependence on repeated spatial-network evaluation and differentiation is common to the physics-informed DeepO formulation.
When the number of physics points is increased proportionally with spatial resolution, \(m_{\mathrm{phy}}\propto N_x\), the dominant online training cost correspondingly scales as
\(
\mathcal{C}_{\mathrm{DeepO}}
=
O\!\left(
\kappa_{\mathrm{phy}}\,
N_x\,
\mathcal{C}_{\mathrm{trunk}}
\right),
\)
whereas the online NBO-RC complexity remains independent of \(N_x\) for a fixed coefficient dimension after the spatial operators have been precomputed.

\textbf{Complexity comparison.}
For the representative Kolmogorov-flow configuration, the resulting per-record complexity ratio is
\(
\frac{\mathcal{C}_{\mathrm{DeepO}}}
{\mathcal{C}_{\mathrm{NBO}}}
\simeq
\frac{9.2\times10^{9}}
{6.3\times10^{6}}
\simeq
1.5\times10^{3}.
\)
Thus, the operation-count analysis predicts a \(>10^3\) orders of magnitude theoretical reduction in online physics-training workload for NBO-RC relative to DeepO-PI-DT.
The same structural advantage applies to the Taylor--Couette example as well.

\vspace{-5pt}
\subsection*{Computational environment} 
\vspace{-10pt}
The single-instance experiments are performed on an Apple M3 system
with 24~GB unified memory.
The parametric benchmarks are performed on the ACES high-performance
computing system at Texas A\&M University.
Offline precomputation is carried out on a CPU compute node (dual Intel Xeon 8468 processors + 512~GB memory) utilizing 96 CPU cores, while operator-network training and inference are performed using a single NVIDIA H100 GPU with 80~GB GPU memory.
All reported training and inference timings for NBO and DeepO are measured on the same H100 GPU environment using double-precision (FP64) arithmetic.

\vspace{-5pt}
\subsection*{Code availability} 
\vspace{-10pt}
The source code and scripts required to reproduce the results will be made publicly available on GitHub upon publication.

\section*{Acknowledgments}
\vspace{-15pt} 
Portions of this research used the advanced computing resources provided by Texas A\&M High Performance Research Computing through allocation MTH260049 from the Advanced Cyberinfrastructure Coordination Ecosystem:
Services \& Support (ACCESS) program, which is supported by U.S. National
Science Foundation grants \#2138259, \#2138286, \#2138307, \#2137603,
and \#2138296~\cite{boerner2023access}. ZW was partially supported by U.S. National Science Foundation grants \#2245097 and \#2513112. MW was partially supported by U.S. National Science Foundation grants \#2603912 and \#2608740.

\putbib[reference]
\end{bibunit}

\clearpage

\makeatletter
\def\bibcite#1#2{%
  \global\@namedef{b@#1.SI}{%
    \hyper@@link[cite]{}{cite.#1.SI}{#2}%
  }%
}
\def\@extra@b@citeb{.SI}
\makeatother

\begin{bibunit}[unsrt]

\begin{center}
{\LARGE
\textbf{Supplementary Information:}\\[0.5em]
Solving and Learning Navier-Stokes Dynamics
with the Neural Basis Method\par}
\end{center}

\vspace{2em}

\setcounter{section}{0}
\renewcommand{\thesection}{S\arabic{section}}
\renewcommand{\thesubsection}{\thesection.\arabic{subsection}}
\renewcommand{\thesubsubsection}{\thesubsection.\arabic{subsubsection}}

\numberwithin{equation}{section}
\renewcommand{\theequation}{\thesection.\arabic{equation}}

\setcounter{lemma}{0}
\renewcommand{\thelemma}{S\arabic{lemma}}

\setcounter{theorem}{0}
\renewcommand{\thetheorem}{S\arabic{theorem}}
\renewcommand{\theHtheorem}{SI.\arabic{theorem}}


\captionsetup[table]{
    font=footnotesize,
    labelfont=bf
}

\setlength{\LTpre}{6pt}
\setlength{\LTpost}{6pt}


\section{Notation}

{\footnotesize
\renewcommand{\arraystretch}{1.12}
\setlength{\tabcolsep}{4pt}

\begin{longtable}{
@{}
>{\raggedright\arraybackslash}p{0.18\textwidth}
>{\raggedright\arraybackslash}p{0.40\textwidth}
>{\raggedright\arraybackslash}p{0.35\textwidth}
@{}
}

\caption{\bfseries Notation used in the main text and Supplementary Information.}
\vspace{-5pt}

\label{tab:notation}\\

\toprule
\textbf{Notation}
&
\textbf{Definition / Meaning}
&
\textbf{Remarks}
\\
\midrule
\endfirsthead

\multicolumn{3}{@{}l}{\footnotesize\textit{Table \thetable\ continued}}\\[2pt]
\toprule
\textbf{Notation}
&
\textbf{Definition / Meaning}
&
\textbf{Remarks}
\\
\midrule
\endhead

\midrule
\multicolumn{3}{r@{}}{\footnotesize\textit{Continued on next page}}
\\
\endfoot

\bottomrule
\endlastfoot

\addlinespace[5pt]
\multicolumn{3}{@{}l}{\textbf{Physical and geometric quantities}}
\\
\addlinespace[2pt]

\(\Omega,\partial\Omega\)
&
Spatial domain and its boundary
&
\(\Omega\subset\mathbb{R}^{d}\).
\\

\(d\)
&
Spatial dimension
&
\(d=1\), \(2\), or \(3\).
\\

\(\mathbf{x}\)
&
Spatial coordinate
&
\(\mathbf{x}=(x,y)\) in 2D Cartesian.
\\

\(x,y\)
&
Cartesian spatial coordinates
&

\\

\(r,\theta,z\)
&
Radial/azimuthal/axial coordinates
&
For axisymmetric Taylor--Couette flow.
\\

\(t,\Delta t,t^n\)
&
Time, time-step, and discrete level \(n\)
&
\\

\(\mathbf{u}\)
&
Velocity field
&
Vector-valued primary variable.
\\

\(u,v\)
&
Cartesian velocity components
&
\(\mathbf{u}=(u,v)\) in 2D.
\\

\(u_r,u_\theta,u_z\)
&
Radial/azimuthal/axial velocity components
&
For axisymmetric Taylor--Couette flow.
\\

\(p\)
&
Kinematic pressure
&
Scalar-valued non-primary variable.
\\

\(\mathbf{F}\)
&
Body-force vector
&
\\

\(\nu\)
&
Kinematic viscosity
&
\\

\(Re\)
&
Reynolds number
&
\\

\(K,\partial K\)
&
Control-volume and its boundary
&
Non-bold; distinguish from matrix \(\mathbf K\).
\\

\(f\)
&
Face of a control-volume
&
\(f\in\partial K\).
\\

\(\mathbf n,\mathbf n_f\)
&
Outward unit normal
&
\(\mathbf n_f\) is the normal associated with face \(f\).
\\

\(|K|,|f|\)
&
Measure of \(K\) and \(f\)
&
\\

\(\mathbf u_K^n,\mathbf u_f^n,p_f\)
&
Velocity at center/face, and pressure at face
&
\(n\) denote discrete time levels.
\\

\(F_f\)
&
Convective face flux
&
\(F_f=(\mathbf u_f\cdot\mathbf n_f)|f|\).
\\

\(\widehat F_f^{\,n+1}\)
&
Extrapolated convecting face flux for \(n+1\)
&
\\

\(\nabla^\perp\)
&
Rotated gradient
&
In 2D,
\(\nabla^\perp\phi=(\partial_y\phi,-\partial_x\phi)^{T}\).
\\

\addlinespace[10pt]
\multicolumn{3}{@{}l}{\textbf{Neural representation and neural basis spaces}}
\\
\addlinespace[2pt]

\(\phi_i\), \(\varphi_j\)
&
Scalar neural basis
&
\\

\(\boldsymbol{\Phi}_i\)
&
Divergence-free vector-valued velocity basis 
&
In 2D,
\(\boldsymbol{\Phi}_i=\nabla^\perp\phi_i\).
\\

\(\mathbf u_{\rm NN}\)
&
Neural-basis representation of velocity field
&
\(\mathbf u_{\rm NN}
=\sum_i\alpha_i\boldsymbol{\Phi}_i\).
\\

\(p_{\rm NN}\)
&
Neural-basis representation of pressure field
&
\(p_{\rm NN}
=\sum_j\beta_j\varphi_j\).
\\

\(\alpha_i,\boldsymbol{\alpha}\)
&
Velocity coefficient and coefficients vector
&
\\

\(\alpha_i^*,\boldsymbol{\alpha}^*\)
&
Velocity coefficient and coefficients vector
&
\(*\) denotes the exact minimizer.
\\

\(\widetilde{\boldsymbol{\alpha}}\)
&
Computed velocity-coefficient vector
&
\\

\(\beta_j,\boldsymbol{\beta}\)
&
Pressure coefficient and coefficients vector
&
\\

\(N_b\)
&
Number of velocity neural basis functions
&
\\

\(N_b^p\)
&
Number of pressure neural basis functions
&
Used when pressure is considered.
\\

\(V_{N_b}\)
&
Divergence-free neural velocity space
&
Spanned by the \(N_b\) sampled bases.
\\

\(\Theta\)
&
Parameter space of the neural-basis generator
&
Used in the proof of Theorem 1.
\\

\(\theta,\theta_i\)
&
Used in the proof of Theorem 1.
&
\(\theta_i\overset{\rm i.i.d.}{\sim}\pi\).
\\

\(\pi\)
&
Sampling probability distribution over \(\Theta\)
&
Used in the proof of Theorem 1.
\\

\addlinespace[10pt]
\multicolumn{3}{@{}l}{\textbf{NBM discretization and pressure-gradient elimination}}
\\
\addlinespace[2pt]

\(N_c\)
&
Number of VCVs or VCV collocation points
&
\\

\(\mathbf A_{\mathbf u}\)
&
Velocity-dependent matrix (NBM discretized)
&
May depend on parameter and time.
\\

\(\mathbf G_p\)
&
Pressure-gradient matrix (NBM discretized)
&
\\

\(\mathbf b\)
&
Known right-hand-side vectors
&
\\

\(\mathbf W\)
&
Discrete residual-weighting matrix
&
\(\mathbf W=\mathbf W^T\succ0\).
\\

\(\|\mathbf r\|_{\mathbf W}\)
&
Weighted discrete residual norm
&
\(\|\mathbf r\|_{\mathbf W}^2
=\mathbf r^T\mathbf W\mathbf r\).
\\

\(\mathbb P_\perp\)
&
Leray-type continuous projector
&
\\

\(\mathcal R(\mathbf u)\)
&
Continuous velocity-only residual operator
&
The exact solution satisfies \(\mathcal R(\mathbf u)=0\).
\\

\(\mathcal G_h\)
&
Discrete pressure-gradient space
&
Defined in the VCV residual space.
\\

\(\mathbf G\)
&
Matrix of discrete pressure-gradient space
&
\\

\(\mathbf P_\perp^{\mathbf W}\)
&
Orthogonal projector onto the complement of \(\mathcal G_h\)
&
Discrete counterpart of \(\mathbb P_\perp\).
\\

\(\mathbf I\)
&
Identity matrix
&
\\

\((\cdot)^\dagger\)
&
Moore--Penrose pseudoinverse
&
\\

\(\operatorname{Col}(\mathbf G)\)
&
Column space of \(\mathbf G\)
&
\\

\(R_K(\mathbf v)\)
&
Integrated momentum residual over VCV \(K\)
&
Before pressure-gradient elimination.
\\

\(R_{K,N_c}(\mathbf v)\)
&
Quadrature-evaluated residual over VCV \(K\)
&
Used in the VCV consistency analysis.
\\

\(\mathbf r^{N_c}(\mathbf v)\)
&
Vector collecting the VCV residuals
&
Discrete unprojected residual vector.
\\

\(\mathbf r^{N_c}_{\perp}(\mathbf v)\)
&
Discrete projected velocity-only residual
&
\(\mathcal R_{N_c}
=\mathbf P_\perp^{\mathbf W}\mathbf r^{N_c}\).
\\

\(\mathbf A_{N_b,N_c}\)
&
Weighted projected least-squares matrix
&
\(\mathbf A_{N_b,N_c}
=\mathbf W^{1/2}
\mathbf P_\perp^{\mathbf W}
\mathbf A_{\mathbf u}\).
\\

\(\mathbf b_{N_c}\)
&
Weighted projected right-hand side
&
\(\mathbf b_{N_c}
=\mathbf W^{1/2}
\mathbf P_\perp^{\mathbf W}
\mathbf b\).
\\

\addlinespace[10pt]
\multicolumn{3}{@{}l}{\textbf{Neural basis operator learning}}
\\
\addlinespace[2pt]

\(\boldsymbol{\xi}\)
&
Physical parameter vector
&
\\

\(\boldsymbol{\zeta}\)
&
Parameter--time input
&
\(\boldsymbol{\zeta}=(\boldsymbol{\xi},t)\).
\\

\(\boldsymbol{\eta}\)
&
Trainable neural-network parameters
&
\\

\(\mathcal N_{\boldsymbol{\eta}}\)
&
Trainable neural mapping
&
For NBO,
\(\boldsymbol{\zeta}
\mapsto
\boldsymbol{\alpha}_{\boldsymbol{\eta}}(\boldsymbol{\zeta})\).
\\

\(\boldsymbol{\alpha}_{\boldsymbol{\eta}}(\boldsymbol{\zeta})\)
&
NBO-predicted velocity-coefficient vector
&
\\

\(\mathcal D_{\rm tr}\)
&
Parameter--time training set
&
Contains training instances \(\boldsymbol{\zeta}\).
\\

\(N_s\)
&
Number of parameter--time samples
&
Also used in sample averaging.
\\

\(\mathbb E_{\boldsymbol{\zeta}\in\mathcal D_{\rm tr}}\)
&
Empirical average over training samples
&
Written as an expectation.
\\

\(\mathcal L_{\rm RC}\)
&
Residual-certified NBO training loss
&
Based on NBM residual.
\\

\(\mathcal L_{\rm DD}\)
&
Data-driven NBO training loss
&
Coefficient-space mean-squared-error.
\\

\(\delta\boldsymbol{\alpha}\)
&
Velocity-coefficient error
&
\\

\(\mathbf G_{\rm R}\)
&
Residual Gram matrix in coefficient space
&
\(\mathbf G_{\rm R}
=\mathbf L_\perp^T\mathbf W\mathbf L_\perp\).
\\

\(\mathbf G_{{\rm R},{\rm mean}}\)
&
Sample-averaged residual Gram matrix
&
\(\mathbf G_{{\rm R},{\rm mean}}
=N_s^{-1}\sum_i\mathbf G_{{\rm R},i}\).
\\

\addlinespace[10pt]
\multicolumn{3}{@{}l}{\textbf{NBM error estimate and discrete residual--error certificate}}
\\
\addlinespace[2pt]

\(V,\|\cdot\|_V\)
&
Velocity solution space and its norm
&
\\

\(\mathbf u_{N_b,N_c}^{*}\)
&
Exact discrete NBM solution for \(N_b\) and \(N_c\)
&
Minimizer of the discrete residual.
\\

\(\widetilde{\mathbf u}_{N_b,N_c}\)
&
Computed discrete NBM solution
&
Includes algebraic least-squares error.
\\

\(\varepsilon_{\rm LS}\)
&
Prescribed least-squares optimality tolerance
&
Controls the algebraic solution error.
\\

\(\sigma_{\min}(\mathbf A)\)
&
Smallest singular value of matrix \(\mathbf A\)
&
Appears in the algebraic-error estimate.
\\

\(C_{\rm basis}\)
&
Stability constant of the neural bases
&
Assumed independent of \(N_b\).
\\

\(q\)
&
Order of accuracy of the VCV quadrature
&
\\

\(\delta\)
&
Failure-probability parameter
&
\\

\(C_1,C_2,C_3\)
&
Positive constants in NBM error estimate
&
Independent of \(N_b\) and \(N_c\).
\\

\(\alpha,\beta\)
&
Constants governing the stretched-exponential neural approximation rate
&
Used in Theorem 1.
\\

\(R_\perp(\boldsymbol{\alpha})\)
&
Discrete projected residual for one linearized implicit time step
&
\\

\(\mathbf L_\perp\)
&
Projected discrete linear-system matrix
&
\\

\(\mathbf g_\perp\)
&
Corresponding projected right-hand side
&
Contains all known contributions.
\\

\(\mathbf u^*,\widetilde{\mathbf u}\)
&
Exact projected discrete velocity and approximate velocity
&
Used in Theorem 2.
\\

\(\mathbf M\)
&
Assembled velocity contribution to the discrete energy matrix
&
\\

\(\mathbf K\)
&
Assembled velocity-gradient contribution to the discrete energy matrix
&
\\

\(\mathbf E\)
&
Discrete velocity-energy matrix
&
\(\mathbf E
=\Delta t^{-1}\mathbf M+\nu\mathbf K\).
\\

\(\|\cdot\|_{\mathcal E}\)
&
Velocity-energy norm
&
\\

\(\mathbf B_\perp\)
&
Energy-scaled projected residual operator
&
\(\mathbf B_\perp
=\mathbf W^{1/2}\mathbf L_\perp\mathbf E^{-1/2}\).
\\

\(\sigma_{\min}(\mathbf B_\perp),
\sigma_{\max}(\mathbf B_\perp)\)
&
Min and max singular values of \(\mathbf B_\perp\)
&
\\

\addlinespace[10pt]
\multicolumn{3}{@{}l}{\textbf{Continuous residual--error certificate}}
\\
\addlinespace[2pt]

\(V_{\rm div}\)
&
Homogeneous divergence-free velocity space
&
\\

\(V_{\rm div}^{*}\)
&
Dual space of \(V_{\rm div}\)
&
\\

\(\mathcal G\)
&
Continuous pressure-gradient space
&
\(\mathcal G=\{\nabla q:q\in Q\}\).
\\

\(Q\)
&
Scalar pressure-function space
&
\\

\(\mathbf u^n\)
&
Frozen convecting velocity in a linearized implicit step
&
Divergence-free by construction.
\\

\(a(\mathbf v,\mathbf w)\)
&
Bilinear form of the projected linearized problem
&
Contains transient, convective, and diffusive contributions.
\\

\(\ell(\mathbf w)\)
&
Linear functional of the projected weak problem
&
Contains forcing, previous-state, and lifted-boundary contributions.
\\

\(r_\perp(\widetilde{\mathbf u})\)
&
Continuous weak projected residual functional
&
Element of \(V_{\rm div}^{*}\).
\\

\(\mathbf e\)
&
Continuous velocity error
&
\(\mathbf e
=\widetilde{\mathbf u}-\mathbf u^*\).
\\

\((\cdot,\cdot)_{L^2}\)
&
\(L^2\) inner product
&
For the continuous weak formulation.
\\

\(H^1,L^2,L^\infty\)
&
Sobolev and Lebesgue function spaces
&
\\

\(\Gamma_{\rm out}\)
&
Outflow portion of the domain boundary
&
Used in the outflow-boundary argument.
\\

\addlinespace[10pt]
\multicolumn{3}{@{}l}{\textbf{Neural approximation and random-coverage analysis}}
\\
\addlinespace[2pt]

\(\mathbf W^{k,\infty}(\Omega)\)
&
Vector-valued Sobolev space
&
\(\mathbf W^{k,\infty}(\Omega)
=[W^{k,\infty}(\Omega)]^d\).
\\

\(k\)
&
Sobolev derivative order
&
Used in the approximation analysis.
\\

\(C_V\)
&
Constant relating \(V\)-norm to Sobolev norm
&
Independent of the represented field.
\\

\(\mathcal D_{\rm all}\)
&
Admissible scalar neural dictionary
&
\(\mathcal D_{\rm all}
=\{\phi_\theta:\theta\in\Theta\}\).
\\

\(M\)
&
Number of terms in an ideal finite neural approximant
&
Distinct from the actual sampled basis count \(N_b\).
\\

\(\mathcal N_M(\mathcal D_{\rm all})\)
&
Scalar \(M\)-term neural approximation class
&
Used in the best-approximation analysis.
\\

\(\mathcal N_M^{\rm div}(\mathcal D_{\rm all})\)
&
Divergence-free vector-valued \(M\)-term neural approximation class
&
Constructed from \(\boldsymbol{\Phi}_{\theta_j}\).
\\

\((Q,R)\)
&
Constants characterizing analytic regularity
&
Defined through
\(|f|_{W^{s,\infty}}
\le Q R^{-s}s!\).
\\

\(s\)
&
Derivative order in the analytic-regularity
&
\(s=0,1,2,\ldots\).
\\

\(\psi\)
&
Scalar potential generating a 2D divergence-free velocity
&
\(\mathbf u=\nabla^\perp\psi\).
\\

\(A_\ell,W_\ell,P_\ell,b_\ell\)
&
Layer activation, feedforward weight, residual-branch projection, and bias quantities
&
Used in the neural-approximation proof.
\\

\(t_\ell,a_\ell,z_\ell,s_\ell\)
&
Affine preactivation, scaled preactivation, nonlinear activation, and residual-branch output
&
Layerwise architecture quantities.
\\

\(\alpha_\ell,\beta_\ell\)
&
Layerwise scaling and residual-branch coefficients
&
Architecture-specific quantities.
\\

\(\widehat f_N\)
&
Constructive two-hidden-layer tanh neural approximant
&
Used in the approximation proof.
\\

\(N\)
&
Network-width parameter in the constructive approximation result
&
Distinct from \(N_b\) and \(N_c\).
\\

\(M_N\)
&
Number of second-hidden-layer features associated with \(\widehat f_N\)
&
Bounded proportionally to \(N\).
\\

\(C_{\rm arch}\)
&
Architecture-dependent constant
&
Used to relate \(M_N\) to \(N\).
\\

\(C_{\rm app}\)
&
Neural best-approximation constant
&
\\

\(u_M^*,a_j^*,\theta_j^*\)
&
Selected \(M\)-term approximant, its coefficients, and neural-atom parameters
&
Used in the random-coverage analysis.
\\

\(B_j\)
&
Neighborhood of selected neural parameter \(\theta_j^*\)
&
\(B_j\subset\Theta\).
\\

\(c_p\)
&
Lower-bound constant for the probability mass of \(B_j\)
&
\(\pi(B_j)\ge c_p/M\).
\\

\(\varepsilon_M\)
&
Maximum neural-atom discrepancy within neighborhood \(B_j\)
&
Controls coverage error.
\\

\(A_M\)
&
Sum of absolute coefficients of the selected \(M\)-term approximant
&
\(A_M=\sum_j|a_j^*|\).
\\

\(E_j\)
&
Event that none of the \(N_b\) sampled parameters lies in \(B_j\)
&
Used in the union-bound argument.
\\

\(\widetilde\theta_j,\widetilde u_M\)
&
Sampled parameter covering the \(j\)-th target atom and corresponding approximation
&
Used in the random-coverage proof.
\\

\(C_{\rm cov},C_b\)
&
Coverage-error constant and combined approximation constant
&
\(C_b=C_{\rm app}+C_{\rm cov}\).
\\

\(M_{N_b}\)
&
Choice of the ideal approximation size expressed as a function of \(N_b\)
&
Used to obtain the \(N_b\)-dependent high-probability estimate.
\\

\(h_c\)
&
Characteristic VCV size
&
For a quasi-uniform VCV covering,
\(N_c\asymp h_c^{-d}\).
\\

\(\|\mathcal R(\cdot)\|_Y\)
&
Continuous residual norm
&
\\

\(\|\mathcal R(\cdot)\|_{Y,N_c}\)
&
VCV-discretized residual norm
&
\\

\(C_{\rm vcv}\)
&
VCV discretization consistency constant
&
Appears in the \(N_c^{-q/d}\) error term.
\\

\(L_{\rm R}\)
&
Local Lipschitz constant of the continuous residual operator
&
Used in the proof of Theorem 1.
\\

\(C_{\rm stab}\)
&
Local residual-to-solution stability constant
&
Used in the proof of Theorem 1.
\\

\(C_{\rm LS}\)
&
Least-squares algebraic-error constant
&
\\

\(\kappa(\mathbf A_{N_b,N_c})\)
&
Condition number of the least-squares matrix
&
\\

\(v_{N_b}\)
&
Trial approximation in \(V_{N_b}\)
&
Constructed in the proof of Theorem 1.
\\

\addlinespace[10pt]
\multicolumn{3}{@{}l}{\textbf{A posteriori pressure recovery}}
\\
\addlinespace[2pt]

\(g_x,g_y\)
&
Target pressure-gradient components obtained from the momentum equation
&
Used for a posteriori pressure recovery.
\\

\(\mathbf g_x,\mathbf g_y\)
&
Vectors of target pressure-gradient components
&
Evaluated over VCV centers.
\\

\(\mathbf P_x,\mathbf P_y\)
&
Matrices of pressure-basis spatial derivatives
&
Used in pressure-gradient fitting.
\\

\(\mathbf P_{\rm bc}\)
&
Matrix of pressure-basis values at boundary
&
Used in the pressure-recovery.
\\

\(\lambda_{\rm bc}\)
&
Weight of the pressure boundary/gauge term
&
Used in the pressure-recovery.
\\

\(\Delta x,\Delta y\)
&
Cartesian VCV spacings
&
For a uniform grid,
\(|K|=\Delta x\Delta y\).
\\

\(\mathbf x_K,\mathbf x_f\)
&
VCV-center and boundary-face coordinates
&
Used in pressure recovery.
\\

\addlinespace[10pt]
\multicolumn{3}{@{}l}{\textbf{Benchmark-specific and diagnostic quantities}}
\\
\addlinespace[2pt]

\(\lambda\)
&
Decay parameter
&
For the analytical Kovasznay solution.
\\

\(U_{\rm lid}\)
&
Prescribed lid-velocity scale
&
Lid-driven cavity benchmark.
\\

\(\chi\)
&
Forcing amplitude
&
Parameter for Kolmogorov flow.
\\

\(\Omega_i,\Omega_o\)
&
Inner- and outer-cylinder angular velocities
&
\(\Omega_i\) is the Taylor--Couette parameter.
\\

\(\gamma\)
&
LUST blending parameter
&
Used in the convective discretization.
\\

\(\mu\)
&
Generic physical parameter
&
\(\mu=\chi\) or \(\Omega_i\) in the generalization tests.
\\

\(U_{\rm rms}\)
&
Root-mean-square velocity statistic
&
Used in single-instance benchmarks.
\\

\(E(k)\)
&
Kinetic-energy spectrum
&
\(k\) denotes spectral wavenumber.
\\

\(e_{\rm rel}\)
&
Relative \(L^2\) velocity error
&
Used in parametric benchmark.
\\

\(\varepsilon_{\rm rel}\)
&
Relative residual measure
&
Used in single-instance refinement.
\\

\addlinespace[10pt]
\multicolumn{3}{@{}l}{\textbf{PINN and DeepONet baseline notation}}
\\
\addlinespace[2pt]

\(u_{\boldsymbol{\eta}},
v_{\boldsymbol{\eta}},
p_{\boldsymbol{\eta}}\)
&
Neural-network predictions
&
Used for PINN and DeepO baselines.
\\

\(r_u,r_v,r_c\)
&
PINN residuals of the \(u\)-momentum, \(v\)-momentum, and continuity equations
&
Interior physics residuals.
\\

\(\mathcal L_{\rm PINN}\)
&
Total PINN objective
&
\\

\(\mathcal L_{\rm PDE},\mathcal L_{\rm BC},\)
\(\mathcal L_{\rm IC},\mathcal L_{\rm gauge}\)
&
PINN PDE, boundary-condition, initial-condition, and pressure-gauge losses
&
Baseline-specific quantities.
\\

\(N_{\rm PDE},
N_{\rm BC},
N_{\rm IC},
\newline
N_{\rm anchor}\)
&
Numbers of PDE, boundary, initial, and velocity-anchor training points
&
PINN configuration quantities.
\\

\(N_z\)
&
Latent/reduced representation dimension
&
Used in NBO.
\\

\addlinespace[10pt]
\multicolumn{3}{@{}l}{\textbf{Computational-complexity notation}}
\\
\addlinespace[2pt]

\(d_{\rm in},d_{\rm out}\)
&
Input and output dimensions of a dense neural-network layer
&
A dense forward evaluation requires
\(d_{\rm in}d_{\rm out}\) MACs.
\\

\(d_{\rm coef}\)
&
Dimension of NBO coefficient
&
\\

\(N_x\)
&
Number of spatial evaluation points
&
\\

\(\mathcal C_{\rm coef}^{\rm fwd}\)
&
Forward cost of NBO network
&
Measured in MACs.
\\

\(\mathcal C_{\rm coef}\)
&
Forward-and-backward cost of NBO network
&
\\

\(\mathcal C_{\rm NBO}\)
&
Online NBO-RC training cost per record
&
\(O(\mathcal C_{\rm coef}+d_{\rm coef}^{2})\).
\\

\(\mathcal C_{\rm branch}^{\rm fwd},
\mathcal C_{\rm trunk}^{\rm fwd}\)
&
DeepO branch- and trunk forward workload
&
Per network evaluation.
\\

\(\mathcal C_{\rm rec}^{\rm fwd}\)
&
DeepO forward field-reconstruction workload
&
Before physics-derivative calculations.
\\

\(\mathcal C_{\rm DeepO}\)
&
DeepO physics-training workload
&
Reported in MAC-equivalent operations.
\\

\(m_{\rm phy},m_{\rm BC}\)
&
Numbers of physics/interior and boundary spatial points per training record
&
Enter the online DeepO cost directly.
\\

\(\kappa_{\rm AD}\)
&
Aggregate automatic-differentiation factor
&
\\

\(\kappa_{\rm phy},\kappa_{\rm BC}\)
&
Automatic-differentiation factors for the physics and boundary objectives
&
\\

\addlinespace[10pt]
\multicolumn{3}{@{}l}{\textbf{Mechanistic-analysis notation}}
\\
\addlinespace[2pt]

\(\mathbf q_j\)
&
Random unit direction in coefficient space
&
For coefficient-space metric alignment.
\\

\(\mathbf v_j\)
&
Random unit direction in operator-network parameter space
&
Used in Hessian and Gauss--Newton analyses.
\\

\(\mathbf J_i\)
&
Jacobian of predicted coefficients w.r.t. trainable network parameters for sample \(i\)
&
\(\mathbf J_i
=
\partial
\boldsymbol{\alpha}_{\boldsymbol{\eta}}(\boldsymbol{\zeta}_i)/
\partial\boldsymbol{\eta}\).
\\

\(\mathcal B,B\)
&
Mini-batch of parameter--time instances
&
\(B=|\mathcal B|\).
\\

\(\ell_i\)
&
Samplewise coefficient-space loss
&
Used in the Hessian decomposition.
\\

\(\mathbf C_i\)
&
Coefficient-space curvature matrix of sample \(i\)
&
\(\mathbf C_i
=\nabla_{\boldsymbol{\alpha}}^2\ell_i\).
\\

\(\mathbf H\)
&
Full Hessian of the training objective w.r.t. network parameters
&
\(\mathbf H
=\mathbf H_{\rm GN}+\mathbf H_{\rm nl}\).
\\

\(\mathbf H_{\rm GN}\)
&
Gauss--Newton component of the Hessian
&
\\

\(\mathbf H_{\rm nl}\)
&
Non-Gauss--Newton component of the Hessian
&
\\

\(\Delta_{\rm GN}(\mathbf v)\)
&
Relative discrepancy between the full-Hessian and Gauss--Newton actions along direction \(\mathbf v\)
&
To quantify Gauss--Newton dominance.
\\

\(\Delta_{\rm GN}^{\rm p95}\)
&
95th percentile of \(\Delta_{\rm GN}\) over sampled parameter-space directions
&
Reported during training.
\\

\(\mathbf C_i^{\rm RC},
\mathbf C_i^{\rm DD}\)
&
Coefficient-space curvature matrices for residual-certified and data-driven objectives
&
\(\mathbf C_i^{\rm RC}
=2\mathbf G_{{\rm R},i}\),
\(\mathbf C_i^{\rm DD}=2\mathbf I\).
\\

\(\mathbf H_{\rm GN}^{\rm RC},
\mathbf H_{\rm GN}^{\rm DD}\)
&
Gauss--Newton matrices induced by the RC and DD coefficient-space metrics
&
Used in controlled metric comparisons.
\\

\(\delta\boldsymbol{\alpha}_i(\mathbf v)\)
&
Coefficient perturbation induced by parameter-space direction \(\mathbf v\)
&
\(\delta\boldsymbol{\alpha}_i(\mathbf v)
=\mathbf J_i\mathbf v\).
\\

\(\kappa_{\rm RC}(\mathbf v)\)
&
Residual-certified Gauss--Newton directional curvature
&
Uses the residual Gram metric.
\\

\(\kappa_{\rm MSE}(\mathbf v)\)
&
MSE Gauss--Newton directional curvature
&
For Euclidean coefficient-space metric.
\\

\(\kappa_E(\mathbf v)\)
&
Velocity-energy directional curvature
&
Uses matrix \(\mathbf E\).
\\

\(\boldsymbol{\eta}^{*}\)
&
Trained operator-network parameter vector
&
\\

\(s_1,s_2\)
&
Scalar displacements along two selected parameter-space directions
&
Coordinates on the local two-dimensional plane.
\\

\(\mathbf v_1,\mathbf v_2\)
&
Selected stiff and weak directions
&
Used in the local quadratic analysis.
\\

\(\mathbf K_Q\)
&
\(2\times2\) local quadratic-form matrix associated with metric \(Q\)
&
Used to visualize the local basin geometry.
\\

\(\mathbf Q_i\)
&
Samplewise coefficient-space metric in the local quadratic analysis
&
\\

\(N_{\rm dir}\)
&
Number of sampled directions
&
Used to normalize percentile ranks.
\\
\end{longtable}

} 

\newpage

\section{NBM error estimates}

\subsection{Basic settings}
Let $\Omega\subset\mathbb{R}^d$ be a bounded spatial domain,
and let
\[
V=
\left\{
\mathbf{v}\in H^1(\Omega)^d:
\nabla\cdot\mathbf{v}=0
\right\}
\]
denote the divergence-free velocity space equipped
with the norm
\[
\|\mathbf{v}\|_V
=
\|\mathbf{v}\|_{H^1(\Omega)^d}.
\]

For an integer $k\ge1$, define
\[
\mathbf W^{k,\infty}(\Omega)
:=
[W^{k,\infty}(\Omega)]^d.
\]
Since $\Omega$ is bounded, we have
\[
\|\mathbf w\|_V
\le
C_V\|\mathbf w\|_{\mathbf W^{k,\infty}(\Omega)},
\qquad
\mathbf w\in V\cap\mathbf W^{k,\infty}(\Omega),
\]
where $C_V>0$ is independent of $\mathbf w$.
Let
\[
\mathcal D_{\rm all}
=
\{\phi_\theta:\theta\in\Theta\}
\]
denote the admissible scalar neural dictionary generated by the prescribed
two-hidden-layer tanh residual architecture\cite{wang2026solving,he2016deep}. For $M\ge1$, define the
corresponding scalar $M$-term scalar approximation class
\[
\mathcal N_M(\mathcal D_{\rm all})
=
\left\{
\sum_{j=1}^{M}
a_j\phi_{\theta_j}
:
a_j\in\mathbb R,\;
\theta_j\in\Theta
\right\}.
\]

For the velocity-only formulation, let
$\boldsymbol{\Phi}_\theta$ denote the divergence-free vector-valued neural basis
generated from $\phi_\theta$. In two dimensions, this basis is given by
\[
\boldsymbol{\Phi}_\theta
=
\nabla^\perp\phi_\theta
=
\begin{bmatrix}
\partial_y\phi_\theta\\
-\partial_x\phi_\theta
\end{bmatrix},
\]
so that $\nabla\cdot\boldsymbol{\Phi}_\theta=0$.
The corresponding vector-valued $M$-term approximation class is
\[
\mathcal N_M^{\rm div}(\mathcal D_{\rm all})
=
\left\{
\sum_{j=1}^{M}
a_j\boldsymbol{\Phi}_{\theta_j}
:
a_j\in\mathbb R,\;
\theta_j\in\Theta
\right\}.
\]

The random NBM velocity trial space is denoted by
\[
V_{N_b}
=
\operatorname{span}
\left\{
\boldsymbol{\Phi}_{\theta_1},\dots,
\boldsymbol{\Phi}_{\theta_{N_b}}
\right\},
\qquad
\theta_i\overset{\mathrm{i.i.d.}}{\sim}\pi,
\]
where $\pi$ is the sampling distribution on $\Theta$, specified according to the sampling strategy introduced in \cite{wang2026solving}.

Let \(Y\) denote the continuous residual space equipped with the norm
\(\|\cdot\|_Y\). The continuous velocity-only residual operator is
defined by
\[
\mathcal R:V\to Y,
\qquad
\mathcal R(\mathbf v)
:=
\mathbb P_\perp
\left[
\partial_t\mathbf v
+
\nabla\cdot(\mathbf v\otimes\mathbf v)
-
\nu\Delta\mathbf v
-
\mathbf F
\right].
\]
The exact solution \(\mathbf u\) satisfies
\(
\mathcal R(\mathbf u)=0.
\)

At the discrete level, the domain is partitioned into $N_c$ virtual control
volumes (VCVs). For each VCV $K$, let $R_K(\bv)$ denote the integrated
momentum residual associated with
\(
\partial_t \bv
+
\nabla\cdot(\bv\otimes\bv)
-
\nu\Delta\bv
-
\mathbf F,
\)
prior to the removal of the pressure-gradient component. Collecting the VCV
residuals gives the discrete residual vector
\[
\mathbf r^{N_c}(\bv)
=
\bigl(R_K(\bv)\bigr)_K.
\]

We equip the VCV residual space with the weighted $\ell^2$ norm
\(
\|\mathbf r\|_{\mathbf W}^2
:=
\mathbf r^T\mathbf W\mathbf r,
\)
where $\mathbf W=\mathbf W^T\succ0$ specifies the discrete weighting, or
equivalently the discrete measure, on the VCV residuals. Different choices
of $\mathbf W$ may incorporate VCV volume weights, nonuniform residual
weighting, or other discrete weighting strategies. The simplest choice is
$\mathbf W=\mathbf I$, which gives the standard Euclidean (discrete
$\ell^2$) residual norm.

The continuous projector $\mathbb P_\perp$ is realized in the VCV residual
space by the corresponding $\mathbf W$-orthogonal projector
$\mathbf P_\perp^{\mathbf W}$ onto the complement of the discrete
pressure-gradient space. Accordingly, the discrete velocity-only residual
is defined by
\[
\mathbf r^{N_c}_{\perp}(\bv)
:=
\mathbf P_\perp^{\mathbf W}\mathbf r^{N_c}(\bv).
\]

The exact discrete NBM solution is therefore
\[
\u_{N_b,N_c}^{*}
=
\arg\min_{\bv\in V_{N_b}}
\|\mathbf r^{N_c}_{\perp}(\bv)\|_{\mathbf W}^2,
\]
and the computed least-squares solution is denoted by
$\widetilde{\u}_{N_b,N_c}$.

\subsection{Spectral-like neural approximation}

We first establish the deterministic approximation property underlying the
random neural basis analysis. We assume that the relevant scalar solution field,
or the scalar potential generating the divergence-free velocity field,
admits an analytic extension to an open neighborhood of a rectangular
domain containing \(\Omega\). After a fixed affine transformation, this
domain may be identified with \([0,1]^d\).

A function \(f\) is said to be \((Q,R)\)-analytic if there exist constants
\(Q,R>0\) such that
\[
|f|_{W^{s,\infty}([0,1]^d)}
\le
Q R^{-s}s!,
\qquad
s=0,1,2,\ldots.
\]

The scalar neural bases employed by NBM are generated by the residual
architecture introduced in our previous work~\cite{wang2026solving}. For
each hidden layer,
\[
t_\ell
=
A_{\ell-1}W_\ell^{T}
+
\mathbf{1}b_\ell^{T},
\qquad
a_\ell
=
t_\ell\odot\alpha_\ell^{T},
\]
\[
z_\ell
=
\tanh(a_\ell),
\qquad
s_\ell
=
\beta_\ell A_{\ell-1}P_\ell^{T},
\qquad
A_\ell
=
s_\ell+z_\ell,
\]
where \(\beta_\ell\in\mathbb R\) controls the residual branch. Setting
\[
\beta_\ell=0,
\qquad
\alpha_\ell=\mathbf{1},
\]
reduces each residual layer to a standard tanh feedforward layer.
Consequently, the admissible two-hidden-layer NBM residual-network class
contains the standard two-hidden-layer tanh network class as a special
case. Therefore, constructive approximation bounds for standard
two-hidden-layer tanh networks also provide upper bounds for the best
approximation error of the admissible NBM architecture.

\begin{lemma}[Spectral-like neural approximation for analytically regular solutions]
\label{lem:S1_neural_approx}
Suppose that the exact scalar solution \(u\), or, for a two-dimensional
divergence-free velocity field $\u$, a scalar potential \(\psi\) satisfying
\[
\u=\nabla^\perp\psi,
\]
is \((Q,R)\)-analytic.

Then there exist constants
\(
C_{\rm app}>0,
\,
\alpha>0,
\,
\beta=\frac{1}{d+1},
\)
independent of \(M\), such that
\[
\inf_{v\in\mathcal N_M(\mathcal D_{\rm all})}
\|u-v\|_{H^1(\Omega)}
\le
C_{\rm app}
\exp\!\left(-\alpha M^\beta\right)
\]
for a scalar solution, and
\[
\inf_{\bv\in\mathcal N_M^{\rm div}(\mathcal D_{\rm all})}
\|\u-\bv\|_V
\le
C_{\rm app}
\exp\!\left(-\alpha M^\beta\right)
\]
for the divergence-free velocity representation.
\end{lemma}

\begin{proof}
For a \((Q,R)\)-analytic function \(f\), Corollary~5.6 of
De Ryck, Lanthaler, and Mishra
\cite{deryck2021approximation}
shows that, for every sufficiently large \(N\), there exists a
two-hidden-layer tanh neural network \(\widehat f_N\), with hidden-layer
widths bounded by \(O(N)\), satisfying
\[
\|f-\widehat f_N\|_{W^{m,\infty}([0,1]^d)}
\le
C_0
N^{m/(d+1)}
\exp\!\left[
-\alpha_0
N^{1/(d+1)}
\log N
\right],
\]
where \(C_0,\alpha_0>0\) are independent of \(N\). We take \(m=k\) for
direct scalar approximation and \(m=k+1\) for the two-dimensional
divergence-free velocity construction.

Because the standard two-hidden-layer tanh architecture is a special case
of the admissible NBM residual architecture, the same constructive network
\(\widehat f_N\) belongs to the admissible NBM network class. Moreover,
the output of a two-hidden-layer network is a linear combination of its
second-hidden-layer features. Since the second-hidden-layer width is
\(O(N)\), there exists a constant \(C_{\rm arch}>0\), independent of \(N\),
such that
\[
\widehat f_N
=
\sum_{j=1}^{M_N}
a_j\phi_{\theta_j},
\qquad
M_N\le C_{\rm arch}N.
\]
Hence
\(
\widehat f_N
\in
\mathcal N_{M_N}(\mathcal D_{\rm all}).
\)

For a prescribed sufficiently large \(M\), choose
\(
N
=
\left\lfloor
\frac{M}{C_{\rm arch}}
\right\rfloor.
\)
Then \(M_N\le M\). For a scalar solution \(u=f\),
\[
\inf_{v\in\mathcal N_M(\mathcal D_{\rm all})}
\|u-v\|_{H^1(\Omega)}
\le
\|u-\widehat f_N\|_{H^1(\Omega)}
\le
C_V
\|f-\widehat f_N\|_{W^{k,\infty}(\Omega)},
\]
and therefore after adjusting constants
\[
\inf_{v\in\mathcal N_M(\mathcal D_{\rm all})}
\|u-v\|_{H^1(\Omega)}
\le
C_1
M^{k/(d+1)}
\exp\!\left[
-\alpha_1
M^{1/(d+1)}
\log M
\right].
\]

For the two-dimensional divergence-free velocity, let
\[
\u=\nabla^\perp\psi,
\qquad
\u_M=\nabla^\perp\widehat\psi_N.
\]
Since differentiation is a bounded map from
\(W^{k+1,\infty}\) to \(W^{k,\infty}\),
\[
\|\u-\u_M\|_{W^{k,\infty}(\Omega)}
\le
C
\|\psi-\widehat\psi_N\|_{W^{k+1,\infty}(\Omega)}.
\]
Using the assumed control of the \(V\)-norm then gives
\[
\inf_{\bv\in\mathcal N_M^{\rm div}(\mathcal D_{\rm all})}
\|\u-\bv\|_V
\le
C_2
M^{(k+1)/(d+1)}
\exp\!\left[
-\alpha_2
M^{1/(d+1)}
\log M
\right].
\]

In either case, the algebraic factor can be absorbed into a weaker
stretched-exponential bound. More precisely, for every fixed \(r\ge0\),
there exist constants \(C,c>0\) such that
\[
M^r
\exp\!\left[
-\alpha_2M^{1/(d+1)}\log M
\right]
\le
C
\exp\!\left[
-cM^{1/(d+1)}
\right]
\]
for all sufficiently large \(M\). Enlarging \(C\) to cover the finitely
many remaining values of \(M\) yields
\[
\inf_{v\in\mathcal N_M(\mathcal D_{\rm all})}
\|u-v\|_{H^1(\Omega)}
\le
C_{\rm app}
\exp\!\left[
-\alpha M^{1/(d+1)}
\right],
\]
with the same conclusion for
\(\mathcal N_M^{\rm div}(\mathcal D_{\rm all})\).
Thus
\(
\beta=\frac{1}{d+1},
\)
which proves the lemma.
\end{proof}

\begin{remark}[Role of analytic regularity]
The analyticity assumption is used only to obtain the
stretched-exponential best-\(M\)-term rate in
Lemma~\ref{lem:S1_neural_approx}. If the solution has only finite
Sobolev regularity, constructive tanh-network approximation results still
yield convergence, but generally at an algebraic rate with respect to
network width~\cite{deryck2021approximation}. Thus, analytic regularity is
not required for neural approximation itself; it is the additional
regularity that produces the spectral-like term used in the present error
estimate.
\end{remark}

\subsection{Random coverage of an \texorpdfstring{\(M\)}{M}-term neural approximant}

From this point onward, we focus exclusively on the divergence-free
velocity formulation. Accordingly, all approximation spaces, random trial
spaces, and error estimates are stated directly in the vector-valued
velocity space \(V\). The scalar counterpart can be treated analogously
and is therefore omitted. Recall that, in two dimensions, the
divergence-free neural basis is constructed from the scalar neural basis
through
\[
\boldsymbol{\Phi}_{\theta}
=
\nabla^\perp \phi_{\theta}.
\]

\begin{lemma}[Random coverage of an \(M\)-term neural approximant]
\label{lem:S2_random_coverage}
Let
\(
\mathbf u_M^\ast
=
\sum_{j=1}^{M}
a_j^\ast
\boldsymbol{\Phi}_{\theta_j^\ast}
\in
\mathcal N_M^{\rm div}(\mathcal D_{\rm all})
\)
be an \(M\)-term divergence-free neural approximant satisfying
\(
\|\mathbf u-\mathbf u_M^\ast\|_V
\le
C_{\rm app}e^{-\alpha M^\beta}.
\)
For a prescribed \(\varepsilon_M>0\), assume that there exists a constant
\(c_p>0\), independent of \(M\), such that, for each
\(j=1,\ldots,M\), there exists a neighborhood
\(B_j\subset\Theta\) satisfying
\(
\pi(B_j)\ge \frac{c_p}{M},
\)
and
\(
\|\boldsymbol{\Phi}_{\theta}
-
\boldsymbol{\Phi}_{\theta_j^\ast}\|_V
\le
\varepsilon_M,
\,
\forall \theta\in B_j.
\)
Define
\(
A_M
=
\sum_{j=1}^{M}|a_j^\ast|.
\)
Then, for any \(0<\delta<1\), if
\(
N_b
\ge
\frac{M}{c_p}
\log\frac{M}{\delta},
\)
the random trial space \(V_{N_b}\)
satisfies, with probability at least \(1-\delta\),
\[
\inf_{\mathbf v\in V_{N_b}}
\|\mathbf u-\mathbf v\|_V
\le
C_{\rm app}e^{-\alpha M^\beta}
+
A_M\varepsilon_M.
\]
If, in addition,
\(
A_M\varepsilon_M
\le
C_{\rm cov}e^{-\alpha M^\beta},
\)
then
\[
\inf_{\mathbf v\in V_{N_b}}
\|\mathbf u-\mathbf v\|_V
\le
C_b e^{-\alpha M^\beta},
\]
where
\(
C_b
=
C_{\rm app}+C_{\rm cov}.
\)
\end{lemma}

\begin{proof}
For each \(j=1,\ldots,M\), define the event
\[
E_j
=
\left\{
\theta_i\notin B_j
\text{ for all }
i=1,\ldots,N_b
\right\},
\]
that is, the event that none of the \(N_b\) sampled parameters lies in
\(B_j\). Since the samples are i.i.d.,
\[
\mathbb P(E_j)
=
\left(1-\pi(B_j)\right)^{N_b}.
\]
Using
\(
\pi(B_j)\ge \frac{c_p}{M},
\)
we obtain
\[
\mathbb P(E_j)
\le
\left(1-\frac{c_p}{M}\right)^{N_b}
\le
\exp\left(-\frac{c_pN_b}{M}\right).
\]
Therefore, by the union bound,
\[
\mathbb P\left(
\bigcup_{j=1}^{M}E_j
\right)
\le
M
\exp\left(-\frac{c_pN_b}{M}\right).
\]
If
\(
N_b
\ge
\frac{M}{c_p}
\log\frac{M}{\delta},
\) 
then
\[
M
\exp\left(-\frac{c_pN_b}{M}\right)
\le
\delta.
\]
Hence, with probability at least \(1-\delta\), every neighborhood
\(B_j\) contains at least one sampled parameter.

On this event, for each \(j=1,\ldots,M\), choose an index
\(i_j\in\{1,\ldots,N_b\}\) such that
\(
\theta_{i_j}\in B_j.
\)
Define
\[
\widetilde{\mathbf u}_M
=
\sum_{j=1}^{M}
a_j^\ast
\boldsymbol{\Phi}_{\theta_{i_j}}.
\]
Since
\(
\boldsymbol{\Phi}_{\theta_{i_j}}\in V_{N_b},
\)
we have
\(
\widetilde{\mathbf u}_M\in V_{N_b}.
\)
Therefore,
\[
\inf_{\mathbf v\in V_{N_b}}
\|\mathbf u-\mathbf v\|_V
\le
\|\mathbf u-\widetilde{\mathbf u}_M\|_V.
\]
By the triangle inequality,
\[
\|\mathbf u-\widetilde{\mathbf u}_M\|_V
\le
\|\mathbf u-\mathbf u_M^\ast\|_V
+
\|\mathbf u_M^\ast-\widetilde{\mathbf u}_M\|_V.
\]
The first term satisfies
\[
\|\mathbf u-\mathbf u_M^\ast\|_V
\le
C_{\rm app}e^{-\alpha M^\beta}.
\]
For the second term,
\[
\mathbf u_M^\ast-\widetilde{\mathbf u}_M
=
\sum_{j=1}^{M}
a_j^\ast
\left(
\boldsymbol{\Phi}_{\theta_j^\ast}
-
\boldsymbol{\Phi}_{\theta_{i_j}}
\right).
\]
Hence,
\[
\|\mathbf u_M^\ast-\widetilde{\mathbf u}_M\|_V
\le
\sum_{j=1}^{M}
|a_j^\ast|
\,
\left\|
\boldsymbol{\Phi}_{\theta_j^\ast}
-
\boldsymbol{\Phi}_{\theta_{i_j}}
\right\|_V.
\]
Since \(\theta_{i_j}\in B_j\),
\(
\left\|
\boldsymbol{\Phi}_{\theta_j^\ast}
-
\boldsymbol{\Phi}_{\theta_{i_j}}
\right\|_V
\le
\varepsilon_M.
\)
Therefore,
\[
\|\mathbf u_M^\ast-\widetilde{\mathbf u}_M\|_V
\le
\left(
\sum_{j=1}^{M}|a_j^\ast|
\right)
\varepsilon_M
=
A_M\varepsilon_M.
\]
Combining the two estimates gives
\[
\inf_{\mathbf v\in V_{N_b}}
\|\mathbf u-\mathbf v\|_V
\le
C_{\rm app}e^{-\alpha M^\beta}
+
A_M\varepsilon_M.
\]
If
\(
A_M\varepsilon_M
\le
C_{\rm cov}e^{-\alpha M^\beta},
\)
then
\[
\inf_{\mathbf v\in V_{N_b}}
\|\mathbf u-\mathbf v\|_V
\le
\left(
C_{\rm app}+C_{\rm cov}
\right)
e^{-\alpha M^\beta}
=
C_b e^{-\alpha M^\beta}.
\]
This proves the lemma.
\end{proof}

\subsection{Expressing the coverage estimate in terms of \texorpdfstring{\(N_b\)}{Nb}}
\label{sec:coverage}

The coverage condition in Lemma~\ref{lem:S2_random_coverage} is
\(
N_b
\ge
\frac{M}{c_p}
\log\frac{M}{\delta}.
\)
For a given $0<\delta<1$, let
\[
M_{N_b}
:=
\left\lfloor
\frac{c_p}{2}
\frac{N_b}{\log(N_b/\delta)}
\right\rfloor .
\]
It can be shown that there exist sufficiently large \(N_b\), such that \(1\leq M_{N_b}\le N_b\), and
\[
\log\frac{M_{N_b}}{\delta}
\le
\log\frac{N_b}{\delta}.
\]
Therefore,
\[
\frac{M_{N_b}}{c_p}
\log\frac{M_{N_b}}{\delta}
\le
\frac{1}{c_p}
\left(
\frac{c_p}{2}
\frac{N_b}{\log(N_b/\delta)}
\right)
\log\frac{N_b}{\delta}
=
\frac{N_b}{2}
\le
N_b.
\]
Thus \(N_b\) satisfies the coverage condition for \(M=M_{N_b}\). 

Moreover, for all sufficiently large \(N_b\),
\(
\frac{c_p}{2}
\frac{N_b}{\log(N_b/\delta)}
\ge 2.
\)
Hence,
\[
M_{N_b}
=
\left\lfloor
\frac{c_p}{2}
\frac{N_b}{\log(N_b/\delta)}
\right\rfloor
\ge
\frac{c_p}{4}
\frac{N_b}{\log(N_b/\delta)}.
\]
Since \(\beta>0\), it follows that
\[
e^{-\alpha M_{N_b}^{\beta}}
\le
\exp\left[
-\widetilde{\alpha}
\left(
\frac{N_b}{\log(N_b/\delta)}
\right)^\beta
\right],
\]
where
\[
\widetilde{\alpha}
=
\alpha
\left(\frac{c_p}{4}\right)^\beta.
\]
Renaming \(\widetilde{\alpha}\) as \(\alpha\), we obtain
\[
\inf_{\mathbf v\in V_{N_b}}
\|\mathbf u-\mathbf v\|_V
\le
C_b
\exp\left[
-\alpha
\left(
\frac{N_b}{\log(N_b/\delta)}
\right)^\beta
\right].
\]

\subsection{Consistency of the VCV residual discretization}
Let the VCVs form a quasi-uniform covering of \(\Omega\) with
characteristic size \(h_c\), so that
\[
N_c\asymp h_c^{-d}.
\]
Let
\(
\mathcal V_{N_b}\subset V_{N_b}
\)
be a bounded admissible set of divergence-free velocity fields.
For \(\mathbf v\in\mathcal V_{N_b}\), let
\(
\mathbf r_{\rm ex}^{N_c}(\mathbf v)
=
\bigl(R_{K,\rm ex}(\mathbf v)\bigr)_K
\)
denote the VCV momentum residual vector obtained using exact cell and face
integration, prior to the removal of the pressure-gradient component, and
define
\[
\mathbf r_{\perp,\rm ex}^{N_c}(\mathbf v)
=
\mathbf P_\perp^{\mathbf W}
\mathbf r_{\rm ex}^{N_c}(\mathbf v).
\]
Recall that the quadrature-evaluated residual used by NBM is
\[
\mathbf r_\perp^{N_c}(\mathbf v)
=
\mathbf P_\perp^{\mathbf W}
\mathbf r^{N_c}(\mathbf v).
\]

\begin{lemma}[Consistency of the VCV residual discretization]
\label{lem:S3_vcv_consistency}

Assume that the momentum-residual integrands are sufficiently regular on
the VCVs and that the VCV quadrature is uniformly \(q\)-th order
consistent on \(\mathcal V_{N_b}\), in the sense that
\[
\left|
R_{K,\rm ex}(\mathbf v)-R_K(\mathbf v)
\right|
\le
C_q |K| h_c^q,
\qquad
\forall \mathbf v\in\mathcal V_{N_b},
\]
where \(C_q>0\) is independent of \(K\) and \(N_c\).
Assume further that the discrete weighting matrix satisfies
\[
\lambda_{\max}(\mathbf W)
\le
C_W h_c^{-d},
\]
where \(C_W>0\) is independent of \(N_c\).
Then there exists a constant \(C_{\rm vcv}>0\), independent of \(N_c\),
such that
\[
\left|
\left\|
\mathbf r_{\perp,\rm ex}^{N_c}(\mathbf v)
\right\|_{\mathbf W}
-
\left\|
\mathbf r_\perp^{N_c}(\mathbf v)
\right\|_{\mathbf W}
\right|
\le
C_{\rm vcv} h_c^q,
\qquad
\forall \mathbf v\in\mathcal V_{N_b}.
\]
Since \(N_c\asymp h_c^{-d}\), this is equivalently
\[
\left|
\left\|
\mathbf r_{\perp,\rm ex}^{N_c}(\mathbf v)
\right\|_{\mathbf W}
-
\left\|
\mathbf r_\perp^{N_c}(\mathbf v)
\right\|_{\mathbf W}
\right|
\le
C_{\rm vcv} N_c^{-q/d},
\qquad
\forall \mathbf v\in\mathcal V_{N_b}.
\]
For the second-order midpoint-type VCV discretization considered here,
\(q=2\) in smooth regions.
\end{lemma}

\begin{proof}
For any \(\mathbf v\in\mathcal V_{N_b}\), define
\[
\mathbf e^{N_c}(\mathbf v)
:=
\mathbf r_{\rm ex}^{N_c}(\mathbf v)
-
\mathbf r^{N_c}(\mathbf v).
\]
By the assumed local VCV consistency,
\[
|e_K^{N_c}(\mathbf v)|
=
\left|
R_{K,\rm ex}(\mathbf v)-R_K(\mathbf v)
\right|
\le
C_q |K| h_c^q.
\]
Hence,
\[
\begin{aligned}
\left\|
\mathbf e^{N_c}(\mathbf v)
\right\|_{\mathbf W}^2
&=
\bigl(\mathbf e^{N_c}(\mathbf v)\bigr)^T
\mathbf W
\mathbf e^{N_c}(\mathbf v)
\\
&\le
\lambda_{\max}(\mathbf W)
\left\|
\mathbf e^{N_c}(\mathbf v)
\right\|_2^2
\\
&\le
C_W h_c^{-d}
\sum_K
\left|
R_{K,\rm ex}(\mathbf v)-R_K(\mathbf v)
\right|^2
\\
&\le
C h_c^{-d+2q}
\sum_K |K|^2.
\end{aligned}
\]
By quasi-uniformity,
\(
\max_K |K|
\le
C h_c^d,
\)
and, since the VCVs cover \(\Omega\),
\(
\sum_K |K|=|\Omega|.
\)
Therefore,
\[
\sum_K |K|^2
\le
\left(\max_K |K|\right)
\sum_K |K|
\le
C h_c^d |\Omega|.
\]
It follows that
\(
\left\|
\mathbf e^{N_c}(\mathbf v)
\right\|_{\mathbf W}^2
\le
C h_c^{2q},
\)
and hence
\[
\left\|
\mathbf r_{\rm ex}^{N_c}(\mathbf v)
-
\mathbf r^{N_c}(\mathbf v)
\right\|_{\mathbf W}
\le
C h_c^q.
\]

Since \(\mathbf P_\perp^{\mathbf W}\) is a
\(\mathbf W\)-orthogonal projector, it is non-expansive in the
\(\mathbf W\)-norm. Thus,
\[
\begin{aligned}
\left\|
\mathbf r_{\perp,\rm ex}^{N_c}(\mathbf v)
-
\mathbf r_\perp^{N_c}(\mathbf v)
\right\|_{\mathbf W}
&=
\left\|
\mathbf P_\perp^{\mathbf W}
\left(
\mathbf r_{\rm ex}^{N_c}(\mathbf v)
-
\mathbf r^{N_c}(\mathbf v)
\right)
\right\|_{\mathbf W}
\\
&\le
\left\|
\mathbf r_{\rm ex}^{N_c}(\mathbf v)
-
\mathbf r^{N_c}(\mathbf v)
\right\|_{\mathbf W}
\\
&\le
C h_c^q.
\end{aligned}
\]
The reverse triangle inequality therefore gives
\[
\left|
\left\|
\mathbf r_{\perp,\rm ex}^{N_c}(\mathbf v)
\right\|_{\mathbf W}
-
\left\|
\mathbf r_\perp^{N_c}(\mathbf v)
\right\|_{\mathbf W}
\right|
\le
C h_c^q.
\]
Finally, using
\(
N_c\asymp h_c^{-d},\)
we obtain
\(
h_c^q\asymp N_c^{-q/d},
\)
and hence
\[
\left|
\left\|
\mathbf r_{\perp,\rm ex}^{N_c}(\mathbf v)
\right\|_{\mathbf W}
-
\left\|
\mathbf r_\perp^{N_c}(\mathbf v)
\right\|_{\mathbf W}
\right|
\le
C_{\rm vcv} N_c^{-q/d}.
\]
This proves the lemma.
\end{proof}

\subsection{Proof of Theorem 1}

\begin{proof}
By Lemmas~\ref{lem:S1_neural_approx} and
\ref{lem:S2_random_coverage}, together with the choice
\(M=M_{N_b}\) in Section~\ref{sec:coverage}, with probability at least
\(1-\delta\), there exists
\(
\mathbf v_{N_b}\in V_{N_b}
\)
such that
\[
\|\mathbf u-\mathbf v_{N_b}\|_V
\le
C_b
\exp\left[
-\alpha
\left(
\frac{N_b}{\log(N_b/\delta)}
\right)^\beta
\right].
\]
Since
\(
\mathcal R(\mathbf u)=0,
\)
the local Lipschitz continuity of \(\mathcal R\) gives
\[
\|\mathcal R(\mathbf v_{N_b})\|_Y
=
\|\mathcal R(\mathbf v_{N_b})
-
\mathcal R(\mathbf u)\|_Y
\le
L_{\rm R}
\|\mathbf v_{N_b}-\mathbf u\|_V.
\]
Hence,
\[
\|\mathcal R(\mathbf v_{N_b})\|_Y
\le
L_{\rm R} C_b
\exp\left[
-\alpha
\left(
\frac{N_b}{\log(N_b/\delta)}
\right)^\beta
\right].
\]

By the VCV consistency estimate of
Lemma~\ref{lem:S3_vcv_consistency},
\[
\|\mathbf r_\perp^{N_c}(\mathbf v_{N_b})\|_{\mathbf W}
\le
\|\mathcal R(\mathbf v_{N_b})\|_Y
+
C_{\rm vcv}N_c^{-q/d}.
\]
Therefore,
\[
\|\mathbf r_\perp^{N_c}(\mathbf v_{N_b})\|_{\mathbf W}
\le
L_{\rm R} C_b
\exp\left[
-\alpha
\left(
\frac{N_b}{\log(N_b/\delta)}
\right)^\beta
\right]
+
C_{\rm vcv}N_c^{-q/d}.
\]

Recall that the ideal discrete NBM minimizer is
\[
\mathbf u_{N_b,N_c}^{\ast}
=
\arg\min_{\mathbf v\in V_{N_b}}
\|\mathbf r_\perp^{N_c}(\mathbf v)\|_{\mathbf W}^2.
\]
Since minimizing the squared norm is equivalent to minimizing the norm,
\[
\|\mathbf r_\perp^{N_c}
(\mathbf u_{N_b,N_c}^{\ast})\|_{\mathbf W}
\le
\|\mathbf r_\perp^{N_c}
(\mathbf v_{N_b})\|_{\mathbf W}.
\]
Consequently,
\[
\|\mathbf r_\perp^{N_c}
(\mathbf u_{N_b,N_c}^{\ast})\|_{\mathbf W}
\le
L_{\rm R} C_b
\exp\left[
-\alpha
\left(
\frac{N_b}{\log(N_b/\delta)}
\right)^\beta
\right]
+
C_{\rm vcv}N_c^{-q/d}.
\]

Applying Lemma~\ref{lem:S3_vcv_consistency} once more gives
\[
\|\mathcal R(\mathbf u_{N_b,N_c}^{\ast})\|_Y
\le
\|\mathbf r_\perp^{N_c}
(\mathbf u_{N_b,N_c}^{\ast})\|_{\mathbf W}
+
C_{\rm vcv}N_c^{-q/d}.
\]
Thus,
\[
\|\mathcal R(\mathbf u_{N_b,N_c}^{\ast})\|_Y
\le
L_{\rm R} C_b
\exp\left[
-\alpha
\left(
\frac{N_b}{\log(N_b/\delta)}
\right)^\beta
\right]
+
2C_{\rm vcv}N_c^{-q/d}.
\]

For sufficiently large \(N_b\) and \(N_c\), the preceding approximation
lies in the neighborhood in which the local residual stability estimate
holds. Therefore,
\[
\|\mathbf u-\mathbf u_{N_b,N_c}^{\ast}\|_V
\le
C_{\rm stab}
\|\mathcal R(\mathbf u_{N_b,N_c}^{\ast})\|_Y,
\]
and hence
\[
\|\mathbf u-\mathbf u_{N_b,N_c}^{\ast}\|_V
\le
C_{\rm stab}L_{\rm R} C_b
\exp\left[
-\alpha
\left(
\frac{N_b}{\log(N_b/\delta)}
\right)^\beta
\right]
+
2C_{\rm stab}C_{\rm vcv}N_c^{-q/d}.
\]

It remains to estimate the algebraic error associated with the
least-squares solve. Write
\[
\mathbf u_{N_b,N_c}^{\ast}
=
\sum_{i=1}^{N_b}
\alpha_i^\ast
\boldsymbol{\Phi}_{\theta_i},
\qquad
\widetilde{\mathbf u}_{N_b,N_c}
=
\sum_{i=1}^{N_b}
\widetilde{\alpha}_i
\boldsymbol{\Phi}_{\theta_i}.
\]
By the assumed uniform stability of the neural basis representation,
\[
\begin{aligned}
\left\|
\widetilde{\mathbf u}_{N_b,N_c}
-
\mathbf u_{N_b,N_c}^{\ast}
\right\|_V
&=
\left\|
\sum_{i=1}^{N_b}
(\widetilde{\alpha}_i-\alpha_i^\ast)
\boldsymbol{\Phi}_{\theta_i}
\right\|_V
\\
&\le
C_{\rm basis}
\left\|
\widetilde{\boldsymbol{\alpha}}
-
\boldsymbol{\alpha}^{\ast}
\right\|_2.
\end{aligned}
\]
Since $\balpha^*$ satisfies the normal equations
and $\mathbf A_{N_b,N_c}$ has full column rank,
the prescribed least-squares optimality tolerance yields
\[
\left\|
\widetilde{\boldsymbol{\alpha}}
-
\boldsymbol{\alpha}^{\ast}
\right\|_2
\le
\frac{\varepsilon_{\rm LS}}
{\left[
\sigma_{\min}
\left(
A_{N_b,N_c}
\right)
\right]^2},
\]
we obtain
\[
\left\|
\widetilde{\mathbf u}_{N_b,N_c}
-
\mathbf u_{N_b,N_c}^{\ast}
\right\|_V
\le
C_{\rm basis}
\frac{\varepsilon_{\rm LS}}
{\left[
\sigma_{\min}
\left(
A_{N_b,N_c}
\right)
\right]^2}.
\]

Finally, by the triangle inequality,
\[
\begin{aligned}
\|\mathbf u-\widetilde{\mathbf u}_{N_b,N_c}\|_V
&\le
\|\mathbf u-\mathbf u_{N_b,N_c}^{\ast}\|_V
+
\|\mathbf u_{N_b,N_c}^{\ast}
-
\widetilde{\mathbf u}_{N_b,N_c}\|_V
\\
&\le
C_{\rm stab}L_{\rm R} C_b
\exp\left[
-\alpha
\left(
\frac{N_b}{\log(N_b/\delta)}
\right)^\beta
\right]
\\
&\quad
+
2C_{\rm stab}C_{\rm vcv}N_c^{-q/d}
+
C_{\rm basis}
\frac{\varepsilon_{\rm LS}}
{\left[
\sigma_{\min}
\left(
A_{N_b,N_c}
\right)
\right]^2}.
\end{aligned}
\]
Setting
\[
C_1=C_{\rm stab}L_{\rm R}C_b,
\qquad
C_2=2C_{\rm stab}C_{\rm vcv},
\qquad
C_3=C_{\rm basis},
\]
yields
\[
\boxed{
\|\mathbf u-\widetilde{\mathbf u}_{N_b,N_c}\|_V
\le
C_1
\exp\left[
-\alpha
\left(
\frac{N_b}{\log(N_b/\delta)}
\right)^\beta
\right]
+
C_2N_c^{-q/d}
+
C_3
\frac{\varepsilon_{\rm LS}}
{\left[
\sigma_{\min}
\left(
A_{N_b,N_c}
\right)
\right]^2}
}.
\]
This proves the theorem.
\end{proof}

\newpage

\section{NBM residual--error certificate at neural discrete level}

\subsection{Proof of Theorem~2}
\begin{proof}
\(\boldsymbol{\alpha}^\ast\) is the exact solution of a discrete step, thus
\[
  R_\perp(\widetilde{\boldsymbol{\alpha}})
  =
  \mathbf L_\perp\widetilde{\boldsymbol{\alpha}}-\mathbf g_\perp
  =
  \mathbf L_\perp
  (\widetilde{\boldsymbol{\alpha}}-\boldsymbol{\alpha}^\ast)
  =
  \mathbf L_\perp\delta\boldsymbol{\alpha} .
\]
Since \(\mathbf E\) is positive definite, let
\[
  \mathbf q:=\mathbf E^{1/2}\delta\boldsymbol{\alpha}.
\]
Then, by the definition of the discrete energy norm,
\[
  \|\mathbf q\|_2
  =
  \|\mathbf e_{\mathbf u}\|_{\mathcal E}.
\]
Moreover, we have
\[
  \|R_\perp(\widetilde{\boldsymbol{\alpha}})\|_{W}
  =
  \|\mathbf W^{1/2}\mathbf L_\perp\delta\boldsymbol{\alpha}\|_2
  =
  \|\mathbf B_\perp\mathbf q\|_2 .
\]
Since \(\mathbf B_\perp\) has full column rank, its smallest singular value is
positive. Hence
\[
  \sigma_{\min}(\mathbf B_\perp)\|\mathbf q\|_2
  \le
  \|\mathbf B_\perp\mathbf q\|_2
  \le
  \sigma_{\max}(\mathbf B_\perp)\|\mathbf q\|_2 .
\]
Substituting
\[
  \|\mathbf q\|_2
  =
  \|\mathbf e_{\mathbf u}\|_{\mathcal E},
  \qquad
  \|\mathbf B_\perp\mathbf q\|_2
  =
  \|R_\perp(\widetilde{\boldsymbol{\alpha}})\|_{W},
\]
gives
\[
  \sigma_{\min}(\mathbf B_\perp)
  \|\mathbf e_{\mathbf u}\|_{\mathcal E}
  \le
  \|R_\perp(\widetilde{\boldsymbol{\alpha}})\|_{W}
  \le
  \sigma_{\max}(\mathbf B_\perp)
  \|\mathbf e_{\mathbf u}\|_{\mathcal E}.
\]
This proves the theorem.
\end{proof}
\newpage

\section{NBM residual--error certificate at continuous level}
We provide the continuous weak counterpart of the discrete certificate stated in Theorem~2.  
Let \(\Omega\subset\mathbb R^d\), \(d=2\) or \(3\), be a bounded Lipschitz
domain. After lifting any nonhomogeneous velocity boundary data, define the
homogeneous divergence-free velocity space
\[
  V_{\rm div}
  :=
  \left\{
  \mathbf v\in H^1(\Omega)^d:
  \nabla\cdot\mathbf v=0,\;
  \mathbf v \text{ satisfies the homogeneous velocity boundary conditions}
  \right\}.
\]
The space $V_{\rm div}$ is a subspace of $V$,
with homogeneous velocity boundary conditions imposed.
The NBM neural basis space $V_{N_b}\subset V$
is divergence-free by construction, while velocity
boundary conditions are enforced through the boundary
residuals rather than by the individual neural basis functions.
Let
\(
  \mathcal G
  :=
  \{\nabla q:\ q\in Q\}
\)
denote the continuous pressure-gradient space.
The continuous Leray projection removes the component of the momentum residual in \(\mathcal G\). 
In weak form,
this is equivalently obtained by restricting the residual to test functions in
\(V_{\rm div}\), because pressure-gradient components vanish against
divergence-free test fields.

Consider one linearized implicit Navier--Stokes step with frozen convective velocity \(\mathbf u^n\).
In the NBM formulation applied here, \(\mathbf u^n\) is represented in the same divergence-free space, and therefore
\(
  \nabla\cdot\mathbf u^n=0
\)
holds by construction. We assume that \(\mathbf u^n\in L^\infty(\Omega)^d\) and
that the boundary conditions make the convection term skew-symmetric on
\(V_{\rm div}\):
\[
  ((\mathbf u^n\cdot\nabla)\mathbf v,\mathbf v)_{L^2}=0,
  \qquad
  \forall \mathbf v\in V_{\rm div}.
\]
This holds for periodic boundaries, no-slip or slip walls, and prescribed-velocity boundaries after lifting the prescribed data.

Define the continuous energy norm
\(
  \|\mathbf v\|_{\mathcal E}^{2}
  :=
  \frac{1}{\Delta t}\|\mathbf v\|_{L^2(\Omega)}^2
  +
  \nu\|\nabla\mathbf v\|_{L^2(\Omega)}^2 .
\)
Let \(\mathbf u^\ast\in V_{\rm div}\) denote the exact solution of the
continuous projected problem:
\[
  a(\mathbf u^\ast,\mathbf w)
  =
  \ell(\mathbf w),
  \qquad
  \forall \mathbf w\in V_{\rm div} ,
\]
where
\[
  a(\mathbf v,\mathbf w)
  :=
  \frac{1}{\Delta t}(\mathbf v,\mathbf w)_{L^2}
  +
  ((\mathbf u^n\cdot\nabla)\mathbf v,\mathbf w)_{L^2}
  +
  \nu(\nabla\mathbf v,\nabla\mathbf w)_{L^2},
\]
and
\[
  \ell(\mathbf w)
  :=
  (\mathbf F,\mathbf w)_{L^2}
  +
  \frac{1}{\Delta t}(\mathbf u^n,\mathbf w)_{L^2}.
\]
Additional known terms generated by boundary lifting may be absorbed into
\(\ell\).
For an approximation \(\widetilde{\mathbf u}\in V_{\rm div}\), define the weak projected residual \(r_\perp(\widetilde{\mathbf u})\in V_{\rm div}^\ast\) by
\[
  \langle r_\perp(\widetilde{\mathbf u}),\mathbf w\rangle
  :=
  a(\widetilde{\mathbf u},\mathbf w)-\ell(\mathbf w),
  \qquad
  \forall \mathbf w\in V_{\rm div}.
\]
Here \(V_{\rm div}^{\ast}\) denotes the dual space of \(V_{\rm div}\); hence
\(r_\perp(\widetilde{\mathbf u})\) is a weak residual functional acting on
divergence-free test functions.
The corresponding dual norm is measured with respect to the energy norm:
\[
  \|r_\perp(\widetilde{\mathbf u})\|_{V_{\rm div}^\ast}
  :=
  \sup_{\mathbf w\in V_{\rm div},\,\mathbf w\neq 0}
  \frac{
  |\langle r_\perp(\widetilde{\mathbf u}),\mathbf w\rangle|
  }{
  \|\mathbf w\|_{\mathcal E}
  }.
\]
The following continuous estimate applies to velocity approximations satisfying the prescribed boundary conditions after boundary lifting.
\begin{theorem}
\label{thm:si_continuous_projected_oseen_residual_error}
Let 
\(
  \mathbf e
  :=
  \widetilde{\mathbf u}-\mathbf u^\ast.
\)
Then
\(
  \|\mathbf e\|_{\mathcal E}
  \le
  \|r_\perp(\widetilde{\mathbf u})\|_{V_{\rm div}^\ast}
  \le
  \left(
    1+
    \|\mathbf u^n\|_{L^\infty(\Omega)}
    \sqrt{\frac{\Delta t}{\nu}}
  \right)
  \|\mathbf e\|_{\mathcal E}.
\)
\end{theorem}

\begin{proof}
Since \(\mathbf u^\ast\) satisfies the projected weak problem,
\[
  a(\mathbf u^\ast,\mathbf w)
  =
  \ell(\mathbf w),
  \qquad
  \forall \mathbf w\in V_{\rm div}.
\]
Therefore,
\[
  \langle r_\perp(\widetilde{\mathbf u}),\mathbf w\rangle
  =
  a(\widetilde{\mathbf u},\mathbf w)
  -
  a(\mathbf u^\ast,\mathbf w)
  =
  a(\mathbf e,\mathbf w),
  \qquad
  \forall \mathbf w\in V_{\rm div}.
\]
Hence
\[
  \langle r_\perp(\widetilde{\mathbf u}),\mathbf w\rangle
  =
  \frac{1}{\Delta t}(\mathbf e,\mathbf w)_{L^2}
  +
  \nu(\nabla\mathbf e,\nabla\mathbf w)_{L^2}
  +
  ((\mathbf u^n\cdot\nabla)\mathbf e,\mathbf w)_{L^2}.
\]
Taking \(\mathbf w=\mathbf e\), we obtain
\[
  \langle r_\perp(\widetilde{\mathbf u}),\mathbf e\rangle
  =
  \frac{1}{\Delta t}\|\mathbf e\|_{L^2}^{2}
  +
  \nu\|\nabla\mathbf e\|_{L^2}^{2}
  +
  ((\mathbf u^n\cdot\nabla)\mathbf e,\mathbf e)_{L^2}.
\]
By the skew-symmetry assumption,
\[
  ((\mathbf u^n\cdot\nabla)\mathbf e,\mathbf e)_{L^2}=0.
\]
Therefore,
\[
  \langle r_\perp(\widetilde{\mathbf u}),\mathbf e\rangle
  =
  \frac{1}{\Delta t}\|\mathbf e\|_{L^2}^{2}
  +
  \nu\|\nabla\mathbf e\|_{L^2}^{2}
  =
  \|\mathbf e\|_{\mathcal E}^{2}.
\]
Using the definition of the dual norm gives
\[
  \|r_\perp(\widetilde{\mathbf u})\|_{V_{\rm div}^\ast}
  \ge
  \frac{
  |\langle r_\perp(\widetilde{\mathbf u}),\mathbf e\rangle|
  }{
  \|\mathbf e\|_{\mathcal E}
  }
  =
  \|\mathbf e\|_{\mathcal E}.
\]
This proves the lower bound.

For the upper bound, let \(\mathbf w\in V_{\rm div}\). 
The time and diffusion parts satisfy 
\[
  \left|
  \frac{1}{\Delta t}(\mathbf e,\mathbf w)_{L^2}
  +
  \nu(\nabla\mathbf e,\nabla\mathbf w)_{L^2}
  \right|
  \le
  \|\mathbf e\|_{\mathcal E}
  \|\mathbf w\|_{\mathcal E}
\]
by Cauchy--Schwarz in the energy inner product.
Similarly, for the convective part,
\[
  |((\mathbf u^n\cdot\nabla)\mathbf e,\mathbf w)_{L^2}|
  \le
  \|\mathbf u^n\|_{L^\infty}
  \|\nabla\mathbf e\|_{L^2}
  \|\mathbf w\|_{L^2}.
\]
Since
\[
  \|\nabla\mathbf e\|_{L^2}
  \le
  \frac{1}{\sqrt{\nu}}\|\mathbf e\|_{\mathcal E},
  \qquad
  \|\mathbf w\|_{L^2}
  \le
  \sqrt{\Delta t}\|\mathbf w\|_{\mathcal E},
\]
we have
\[
  |((\mathbf u^n\cdot\nabla)\mathbf e,\mathbf w)_{L^2}|
  \le
  \|\mathbf u^n\|_{L^\infty}
  \sqrt{\frac{\Delta t}{\nu}}
  \|\mathbf e\|_{\mathcal E}
  \|\mathbf w\|_{\mathcal E}.
\]
Combining the time, diffusion and convection estimates yields
\[
  |\langle r_\perp(\widetilde{\mathbf u}),\mathbf w\rangle|
  \le
  \left(
    1+
    \|\mathbf u^n\|_{L^\infty(\Omega)}
    \sqrt{\frac{\Delta t}{\nu}}
  \right)
  \|\mathbf e\|_{\mathcal E}
  \|\mathbf w\|_{\mathcal E}.
\]
Taking the supremum over all nonzero \(\mathbf w\in V_{\rm div}\) gives
\[
  \|r_\perp(\widetilde{\mathbf u})\|_{V_{\rm div}^\ast}
  \le
  \left(
    1+
    \|\mathbf u^n\|_{L^\infty(\Omega)}
    \sqrt{\frac{\Delta t}{\nu}}
  \right)
  \|\mathbf e\|_{\mathcal E}.
\]
Together with the lower bound, this proves the theorem.
\end{proof}

For general outflow boundaries, the theorem also applies because the normal velocity is outward.
The self-convection contribution on the outflow boundary is
\[
  ((\mathbf u^n\cdot\nabla)\mathbf v,\mathbf v)_{L^2(\Omega)}
  =
  \frac12
  \int_{\Gamma_{\rm out}}
  (\mathbf u^n\cdot\mathbf n)|\mathbf v|^2\,ds
  \ge 0 .
\]
Hence this boundary contribution does not weaken the positive time-diffusion energy term, so the lower bound remains valid. The upper bound is also
unchanged, because its proof only uses the boundedness estimate
\(
  |((\mathbf u^n\cdot\nabla)\mathbf e,\mathbf w)_{L^2}|
  \le
  \|\mathbf u^n\|_{L^\infty}
  \|\nabla\mathbf e\|_{L^2}
  \|\mathbf w\|_{L^2} ,
\)
together with the energy norm. 
For a neutral outflow/stress boundary, the stress boundary contribution is zero in the weak formulation. 
Therefore both bounds in
Theorem~\ref{thm:si_continuous_projected_oseen_residual_error}
hold. 
Together with the boundary settings mentioned above, the continuous
residual--error equivalence applies broadly to the boundary configurations commonly used in practice.
\newpage

\section{Additional details of velocity-only construction}

\subsection{A posteriori pressure recovery}

Pressure is not included among the primary unknowns of the
velocity-only solve. 
When needed for further calculations, it can be computed a posteriori from the obtained velocity field.
Different approaches can be designed for this purpose. 
Here we show one example for global recovery. 
We evaluate pressure-gradient target from the momentum equation at the VCV centers:
\begin{equation*}
g_x
=
\rho
\left(
\nu\Delta u
-
\partial_t u
-
u\partial_xu
-
v\partial_yu
\right),
\end{equation*}
and
\begin{equation*}
g_y
=
\rho
\left(
\nu\Delta v
-
\partial_t v
-
u\partial_xv
-
v\partial_yv
\right).
\end{equation*}
For the recovery, we can represent the pressure using a scalar neural basis:
\(
p_{\rm NN}(\mathbf{x})
=
\sum_{j=1}^{N_b^{p}}
\beta_j\varphi_j(\mathbf{x}).
\)
Let
\(
\mathbf{P}_{x}
=
\left[
\partial_x\varphi_j(\mathbf{x}_K)
\right]_{K,j}
\) 
and 
\(\mathbf{P}_{y}
=
\left[
\partial_y\varphi_j(\mathbf{x}_K)
\right]_{K,j},
\)
and let
\(
\mathbf{P}_{\rm bc}
=
\left[
\varphi_j(\mathbf{x}_f)
\right]_{f,j}
\)
contain the pressure-basis values at the boundary faces. 
The pressure coefficients are determined from the gradient-fitting problem
\begin{equation*}
\boldsymbol{\beta}^{*}
=
\arg\min_{\boldsymbol{\beta}}
\left\{
|K|
\left\|
\mathbf{P}_{x}\boldsymbol{\beta}
-
\mathbf{g}_{x}
\right\|_{2}^{2}
+
|K|
\left\|
\mathbf{P}_{y}\boldsymbol{\beta}
-
\mathbf{g}_{y}
\right\|_{2}^{2}
+
\lambda_{\rm bc}|f|
\left\|
\mathbf{P}_{\rm bc}\boldsymbol{\beta}
\right\|_{2}^{2}
\right\},
\end{equation*}
where \(|K|=\Delta x\Delta y\) is the uniform VCV volume, 
\(|f|=\Delta y\) is the boundary-face measure, and \(\lambda_{bc}\) is the boundary weight. 
\newpage

\section{Benchmark configurations}
\subsection{Single-instance benchmark configurations}
\subsubsection{NBM/OpenFOAM configurations}
\begin{table}[h!]
\centering
\caption{\bfseries NBM/OpenFOAM configurations for the single-instance benchmarks.}
\label{tab:single_instance_config}
\vspace{10pt}

{a. neural basis and spatial configurations}

\vspace{3pt}

\begin{tabular}{lcccc}
\hline
Benchmark
& \(Re\)
& \(N_b\)
& NBM \(N_c\)
& OpenFOAM mesh \\
\hline
Kovasznay
& \(500\)
& \(500\)
& \(50\times50\)
& \(50\times50\) \\

Taylor--Green
& \(1000\)
& \(500\)
& \(50\times50\)
& \(50\times50\) \\

Lid-driven cavity
& \(1000\)
& \(1000\)
& \(100\times100\)
& \(100\times100\) \\

Shear-layer
& \(O(10^{4})\)
& \(1000\)
& \(100\times100\)
& \(100\times100\) \\
\hline
\end{tabular}

\vspace{25pt}

{b. temporal and convective configurations}

\vspace{3pt}

\begin{tabular}{lccc}
\hline
Benchmark
& Temporal discretization
& Convective discretization
& Reported state \\
\hline
Kovasznay
& Steady
& LUST, \(\gamma=0.25\)
& Steady solution \\

Taylor--Green
& BDF2, \(\Delta t=0.001\)
& LUST, \(\gamma=0.25\)
& \(t=0.01\) \\

Lid-driven cavity
& Steady
& LUST, \(\gamma=0.25\)
& Steady solution \\

Shear-layer
& BDF2, \(\Delta t=1.0\)
& LUST, \(\gamma=0.25\)
& \(t=30\) \\
\hline
\end{tabular}

\end{table}

Precision is default FP64.  

\newpage

\subsubsection{PINN configurations}

\begin{table}[h!]
\centering
\caption{\bfseries PINN configurations for the single-instance benchmarks.}
\label{tab:single_instance_pinn_config}
\vspace{10pt}
{a. network configurations}

\vspace{3pt}

\begin{tabular}{lcccc}
\hline
Benchmark
& Hidden layers
& Width
& Activation \\
\hline
Kovasznay
& \(5\)
& \(96\)
& Tanh \\

Taylor--Green
& \(5\)
& \(96\)
& Tanh \\

Lid-driven cavity
& \(6\)
& \(256\)
& SiLU \\

Shear-layer
& \(8\)
& \(256\)
& SiLU \\
\hline
\end{tabular}

\vspace{25pt}

{b. collocation-point configurations}

\vspace{3pt}

\begin{tabular}{lcccc}
\hline
Benchmark
& \(N_{\rm PDE}\)
& \(N_{\rm BC}\)
& \(N_{\rm IC}\)
& \(N_{\rm anchor}\) \\
\hline
Kovasznay
& \(2500\)
& \(200\)
& --
& -- \\

Taylor--Green
& \(10000\)
& \(2000\)
& \(2500\)
& -- \\

Lid-driven cavity
& \(10000\)
& \(400\)
& --
& \(500\) \\

Shear-layer
& \(30000\)
& \(12000\)
& \(10000\)
& -- \\
\hline
\end{tabular}

\vspace{25pt}

{c. collocation-point sampling}

\vspace{3pt}

\begin{tabular}{lccc}
\hline
Benchmark
& PDE points
& Boundary points
& Initial points \\
\hline
Kovasznay
& Sobol
& Uniform
& -- \\

Taylor--Green
& Sobol
& Sobol
& Sobol \\

Lid-driven cavity
& Mixed Sobol
& Uniform
& -- \\

Shear-layer
& Mixed Sobol
& Sobol
& Sobol \\
\hline
\end{tabular}

\vspace{25pt}

{d. optimization configurations (lr: learning rate)}

\vspace{3pt}

\setlength{\tabcolsep}{4pt}
\begin{tabular}{
    l
    p{0.37\textwidth}
    p{0.37\textwidth}
    p{0.37\textwidth}
}
\hline
Benchmark
& Adam training
& L-BFGS \\
\hline

Kovasznay
&
\(20000\) iterations,
lr = \(10^{-3}\)
&
\(2000\) iterations,
lr = \(1.0\) \\

Taylor--Green
&
\(5000\)-iteration pretraining at lr = \(10^{-3}\),
followed by \(15000\) iterations at lr = \(5\times10^{-4}\)
&
\(2000\) iterations,
lr = \(1.0\) \\

Lid-driven cavity
&
\(20000\) iterations
at lr = \(2\times10^{-5}\)
&
\(2000\) iterations,
lr = \(0.5\) \\

Shear-layer
&
\(2000\)-iteration pretraining at lr = \(5\times10^{-4}\),
followed by \(50000\) iterations at lr = \(10^{-4}\)
&
\(2000\) iterations,
lr = \(0.5\) \\
\hline
\end{tabular}

\end{table}

Precision is default FP64.  

\newpage

\subsection{Parametric benchmark configurations}
\subsubsection{NBO configurations}
\begin{table}[h!]
\centering
\caption{\bfseries NBO configurations for the parametric benchmarks. 
}
\label{tab:parametric_nbo_config}
\vspace{10pt}
{a. neural basis and spatial configurations}

\vspace{3pt}

\begin{tabular}{lccc}
\hline
Benchmark
& \(N_b\)
& NBM \(N_c\)
& Latent dimension \(N_z\) \\
\hline
Kolmogorov
& \(1000\)
& \(50\times50\)
& \(300\) \\

Taylor--Couette
& \(1000\)
& \(30\times180\)
& \(300\) \\
\hline
\end{tabular}

\vspace{25pt}

{b. parameter--time dataset configurations}

\vspace{3pt}

\begin{tabular}{lcccc}
\hline
Benchmark
& Parameter range
& Time interval
& \(\Delta t\)
& \(N_s\) \\
\hline
Kolmogorov
& \(\chi\in[0.5,1.5]\)
& \(t\in[0,3]\)
& \(0.01\)
& \(4186\) \\

Taylor--Couette
& \(\Omega_i\in[104.082,195.918]\)
& \(t\in[0,0.2]\)
& \(0.002\)
& \(4312\) \\
\hline
\end{tabular}

\vspace{25pt}

{c. network configurations}

\vspace{3pt}

\begin{tabular}{lcccc}
\hline
Benchmark
& Hidden layers
& Width
& Activation \\
\hline
Kolmogorov
& \(4\)
& \(768\)
& Tanh \\

Taylor--Couette
& \(4\)
& \(768\)
& Tanh \\
\hline
\end{tabular}

\vspace{25pt}

{d. optimization configurations}

\vspace{3pt}

\begin{tabular}{lcccc}
\hline
Benchmark
& Epochs
& Batch size
& Learning rate
& Default split \\
\hline
Kolmogorov
& \(1000\)
& \(64\)
& \(10^{-3}\)
& \(80{:}20\) \\

Taylor--Couette
& \(1000\)
& \(16\)
& \(10^{-3}\)
& \(80{:}20\) \\
\hline
\end{tabular}

\end{table}

NBO-RC and NBO-DD use identical configurations and differ only in their training objectives.
Precision is default FP64.  

\newpage

\subsubsection{DeepO configurations}

\begin{table}[h!]
\centering
\caption{\bfseries DeepO configurations for the parametric benchmarks.}
\label{tab:parametric_DeepO_config}

{a. parameter--time modules}

\vspace{3pt}

\begin{tabular}{lcc}
\hline
Configuration
& Input
& Parameter--time module \\
\hline
Kolmogorov DeepO-PI-CT
& \((\chi,t)\)
& Two \(4\times128\) MLPs \\

Kolmogorov DeepO-PI-DT
& \((\chi,t)\)
& One \(4\times128\) MLP \\

Kolmogorov DeepO-DD
& \((\chi,t)\)
& One \(4\times128\) MLP \\

Taylor--Couette DeepO-PI-DT
& \((\Omega_i,t)\)
& One \(4\times128\) MLP \\

Taylor--Couette DeepO-DD
& \((\Omega_i,t)\)
& One \(4\times128\) MLP \\
\hline
\end{tabular}

\vspace{25pt}

{b. field-reconstruction DeepONet architectures}

\vspace{3pt}

\begin{tabular}{lcccc}
\hline
Benchmark
& Branch network
& Trunk network
& Rank \\
\hline
Kolmogorov
& \(5\times256\)
& \(5\times256\)
& \(128\) \\

Taylor--Couette
& \(5\times256\)
& \(5\times256\)
& \(128\) \\
\hline
\end{tabular}

\vspace{25pt}

{c. mini-batch and spatial-sampling configurations. The parameter–time batch size denotes the total number of parameter–time states in each mini-batch. For Kolmogorov DeepO-PI-CT, the batch consists of 4 parameter samples and 8 time samples.}

\vspace{3pt}

\begin{tabular}{lccc}
\hline
Configuration
& Parameter--time batch
& Interior points
& Boundary points \\
\hline
Kolmogorov DeepO-PI-CT
& \(32\) 
& \(512\) per state
& \(128\) per side \& state \\

Kolmogorov DeepO-PI-DT
& \(64\) 
& \(512\) per state
& \(128\) per side \& state \\

Kolmogorov DeepO-DD
& \(64\) 
& \(512\) per state
& \(128\) per side \& state \\

Taylor--Couette DeepO-PI-DT
& \(32\) 
& \(512\) per state
& \(128\) per side \& state \\

Taylor--Couette DeepO-DD
& \(32\) 
& \(512\) per state
& \(128\) per side \& state \\
\hline
\end{tabular}

\vspace{25pt}

{d. optimization configurations}

\vspace{3pt}

\begin{tabular}{lccc}
\hline
Configuration
& Training amount
& Optimizer
& Learning rate \\
\hline
Kolmogorov DeepO-PI-CT
& \(500000\) iterations
& AdamW
& \(10^{-3}\) \\

Kolmogorov DeepO-PI-DT
& \(2000\) epochs
& AdamW
& \(10^{-3}\) \\

Kolmogorov DeepO-DD
& \(2000\) epochs
& AdamW
& \(10^{-3}\) \\

Taylor--Couette DeepO-PI-DT
& \(1000\) epochs
& AdamW
& \(10^{-3}\) \\

Taylor--Couette DeepO-DD
& \(1000\) epochs
& AdamW
& \(10^{-3}\) \\
\hline
\end{tabular}

\vspace{25pt}

{e. training-objective weights}

\vspace{3pt}

\begin{tabular}{lcccc}
\hline
Configuration
& \(\lambda_{\rm PDE}\)
& \(\lambda_{\rm wall}\)
& \(\lambda_{\rm IC}\)
& \(\lambda_{\rm data}\) \\
\hline
Kolmogorov DeepO-PI-CT
& \(1\)
& \(10\)
& \(10\)
& \(0\) \\

Kolmogorov DeepO-PI-DT
& \(1\)
& \(10\)
& \(0\)
& \(0\) \\

Kolmogorov DeepO-DD
& \(0\)
& \(10\)
& \(0\)
& \(1\) \\

Taylor--Couette DeepO-PI-DT
& \(1\)
& \(1\)
& \(0\)
& \(0\) \\

Taylor--Couette DeepO-DD
& \(0\)
& \(1\)
& \(0\)
& \(1\) \\
\hline
\end{tabular}

\end{table}

All models use Tanh activation and FP64 precision.

\subsubsection{Generalization evaluation protocols}
For the random-split generalization tests, the available
parameter--time samples
\(\boldsymbol{\zeta}_{i,j}=(\mu_i,t_j)\), with
\(\mu=\chi\) for Kolmogorov flow and
\(\mu=\Omega_i\) for Taylor--Couette flow, were randomly partitioned
using training--test ratios of
\(80{:}20\), \(20{:}80\), \(1{:}99\), and \(0.1{:}99.9\).
The same partition was used for NBO-RC, NBO-DD, DeepO-PI-DT, and DeepO-DD
at each ratio, so that all methods were trained and evaluated using
identical parameter--time samples. The test errors reported in
Figs.~\ref{fig:kolmogorov}g and \ref{fig:taylor_couette}f were computed
over all held-out samples in the corresponding test partition.
For the parameter-generalization tests, training was restricted to
the interior parameter intervals
\(\chi\in[0.6,1.3]\) for Kolmogorov flow and
\(\Omega_i\in[104.082,195.918]\) for Taylor--Couette flow.
The ID test sets contained previously unseen parameter values sampled
from these interior intervals, whereas the OOD test sets contained
parameter values sampled from the exterior portions of the full
parameter domains:
\(\chi\in[0.5,0.6)\cup(1.3,1.5]\) and
\(\Omega_i\in[100,104.082)\cup(195.918,200]\).
ID and OOD errors were evaluated at the final simulation times,
\(t=3\) for Kolmogorov flow and \(t=0.2\) for Taylor--Couette flow,
as reported in Figs.~\ref{fig:taylor_couette}h and \ref{fig:taylor_couette}g.

\newpage

\section{Mechanistic analysis of certificate-induced operator-network training geometry}
\textbf{Coefficient-space metric alignment.}
For Fig.~6a, we evaluate
\(
\mathbf{q}_j^{T}
\mathbf{G}_{{\rm R},{\rm mean}}
\mathbf{q}_j
\)
and
\(
\mathbf{q}_j^{T}
\mathbf{E}
\mathbf{q}_j
\)
along random unit directions \(\mathbf{q}_j\) in the velocity-coefficient space. These quadratic forms measure the residual and velocity-energy sensitivities of the same coefficient perturbation. Pearson correlation is computed from the base-10 logarithms of the finite positive quadratic forms, whereas Spearman correlation is computed from their rank ordering.

\textbf{Hessian decomposition.}
For either NBO-RC or NBO-DD, let
\(
\mathbf{J}_i
=
\frac{
\partial
\boldsymbol{\alpha}_{\boldsymbol{\eta}}(\boldsymbol{\zeta}_i)
}{
\partial\boldsymbol{\eta}
}
\)
denote the Jacobian of the predicted velocity coefficients with respect to the trainable operator-network parameters.
Let \(\mathcal{B}\) denote a batch of sampled parameter--time instances, with \(B=|\mathcal{B}|\).
For the corresponding samplewise coefficient-space loss \(\ell_i\), the parameter-space Hessian decomposes as
\(
\mathbf{H}
=
\mathbf{H}_{\rm GN}
+
\mathbf{H}_{\rm nl},
\)
where
\(
\mathbf{H}_{\rm GN}
=
\frac{1}{B}
\sum_{i\in\mathcal{B}}
\mathbf{J}_i^{T}
\mathbf{C}_i
\mathbf{J}_i,
\qquad
\mathbf{C}_i
=
\nabla_{\boldsymbol{\alpha}}^{2}\ell_i,
\)
is the Gauss--Newton component.
It is obtained by locally linearizing the nonlinear coefficient map
\(\boldsymbol{\eta}\mapsto
\boldsymbol{\alpha}_{\boldsymbol{\eta}}\)
and retaining the coefficient-space curvature \(\mathbf{C}_i\); consequently, it depends only on the first parameter derivatives \(\mathbf{J}_i\) of the operator-network output.
The remaining non-Gauss--Newton component is
\(
\mathbf{H}_{\rm nl}
=
\frac{1}{B}
\sum_{i\in\mathcal{B}}
\sum_{a=1}^{N_b}
\left.
\frac{\partial\ell_i}
{\partial\alpha_a}
\right|_{
\boldsymbol{\alpha}
=
\boldsymbol{\alpha}_{\boldsymbol{\eta}}(\boldsymbol{\zeta}_i)
}
\nabla_{\boldsymbol{\eta}}^{2}
\alpha_{\boldsymbol{\eta},a}(\boldsymbol{\zeta}_i).
\)
This term arises from the second derivatives of the nonlinear operator-network map with respect to its trainable parameters.
Here, \(a\) indexes the \(N_b\) velocity coefficients, and
\(\alpha_{\boldsymbol{\eta},a}(\boldsymbol{\zeta}_i)\)
denotes the \(a\)th component of
\(\boldsymbol{\alpha}_{\boldsymbol{\eta}}(\boldsymbol{\zeta}_i)\).
The coefficient-space derivatives
\(\mathbf{C}_i\)
and
\(\partial\ell_i/\partial\alpha_a\)
are evaluated at
\(\boldsymbol{\alpha}
=
\boldsymbol{\alpha}_{\boldsymbol{\eta}}(\boldsymbol{\zeta}_i)\).
The non-Gauss--Newton component is weighted by the coefficient-space gradient of the corresponding training loss.
Consequently, provided that the second parameter derivatives of the operator-network output remain bounded, this component becomes small as the coefficient-space loss gradient decreases during training.
This provides the mechanism by which the local training geometry becomes increasingly dominated by the Gauss--Newton component as optimization proceeds.
For Fig.~\ref{fig:mechanistic_probe}b, the previously defined discrepancy
\(\Delta_{\rm GN}(\mathbf{v}_j)\)
is evaluated at the recorded training epochs along sampled unit directions \(\mathbf{v}_j\) in the operator-network parameter space.
We report its \(95\)th percentile over these directions.
Full Hessian--vector products are computed using automatic differentiation, whereas Gauss--Newton--vector products are computed by applying the coefficient-space curvature \(\mathbf{C}_i\) to each Jacobian--vector product and then performing the Jacobian-transpose pullback.
The final-epoch insets compare the directional curvatures
\(
\mathbf{v}_j^{T}\mathbf{H}\mathbf{v}_j
\)
and
\(
\mathbf{v}_j^{T}\mathbf{H}_{\rm GN}\mathbf{v}_j.
\)

\textbf{Loss-induced Gauss--Newton components.}
Although NBO-RC and NBO-DD share the same Hessian decomposition, their Gauss--Newton components arise from the different coefficient-space metrics established above.
Their samplewise coefficient-space curvatures are
\(
\mathbf{C}_i^{\rm RC}
=
2\mathbf{G}_{{\rm R},i}
\)
and
\(
\mathbf{C}_i^{\rm DD}
=
2\mathbf{I},
\)
respectively.
The corresponding Gauss--Newton components are therefore
\(
\mathbf{H}_{\rm GN}^{\rm RC}
=
\frac{2}{B}
\sum_{i\in\mathcal{B}}
\mathbf{J}_i^{T}
\mathbf{G}_{{\rm R},i}
\mathbf{J}_i
\)
and
\(
\mathbf{H}_{\rm GN}^{\rm DD}
=
\frac{2}{B}
\sum_{i\in\mathcal{B}}
\mathbf{J}_i^{T}
\mathbf{J}_i.
\)
Fixed positive normalization factors used in the implemented RC and DD objectives are retained in the numerical calculations but omitted here for clarity.
Thus, even when the full Hessians of both models are well represented by their Gauss--Newton components, NBO-RC and NBO-DD retain different parameter-space curvature geometries because their Jacobians pull back different coefficient-space metrics.

\textbf{Comparison of coefficient-space metrics.}
To understand the effect of the coefficient-space loss metric, we hold the network parameters and Jacobian fixed and apply different metrics to the same Jacobian--vector products. 
For a parameter-space direction $\mathbf v$, the induced coefficient perturbation for the $i$th sample is
\(\delta\boldsymbol{\alpha}_i(\mathbf v)
=
\mathbf J_i\mathbf v.\)
We evaluate the RC Gauss--Newton curvature, MSE Gauss--Newton curvature, and velocity-energy curvature as
\(
\kappa_{\mathrm{RC}}(\mathbf v)
=
\frac{2}{B}
\sum_{i\in\mathcal B}
\delta\boldsymbol{\alpha}_i(\mathbf v)^{T}
\mathbf G_i
\delta\boldsymbol{\alpha}_i(\mathbf v)
\), 
\(\kappa_{\mathrm{MSE}}(\mathbf v)
=
\frac{2}{B}
\sum_{i\in\mathcal B}
\left\|
\delta\boldsymbol{\alpha}_i(\mathbf v)
\right\|_2^2\), and 
\(
\kappa_E(\mathbf v)
=
\frac{1}{B}
\sum_{i\in\mathcal B}
\delta\boldsymbol{\alpha}_i(\mathbf v)^{T}
\mathbf E
\delta\boldsymbol{\alpha}_i(\mathbf v).
\)
The fixed normalization factors associated with the implemented RC and MSE objectives are retained in the numerical evaluations but are omitted here for clarity.
For the NBO-RC network, $\kappa_{\mathrm{RC}}$ is the Gauss--Newton curvature of the training objective, while $\kappa_{\mathrm{MSE}}$ is evaluated as a reference curvature using the same NBO-RC Jacobian. 
For the NBO-DD network, $\kappa_{\mathrm{MSE}}$ is the Gauss--Newton curvature of the training objective, while $\kappa_{\mathrm{RC}}$ is evaluated as a reference curvature using the same NBO-DD Jacobian. 
Thus, each comparison holds the network parameters, parameter directions, and Jacobian--vector products fixed and changes only the coefficient-space metric. 
Figure~\ref{fig:mechanistic_probe}c reports this comparison for Kolmogorov flow, and Figure~\ref{fig:mechanistic_probe}d reports the corresponding comparison for Taylor--Couette flow. 
Pearson correlation is computed from the finite positive log-curvatures, while Spearman correlation is computed from their rank order.

\textbf{Local quadratic geometry.}
For Figure~\ref{fig:mechanistic_probe}e, we examine the local quadratic geometry for Kolmogorov flow (left) and Taylor--Couette flow (right). 
For each benchmark, we use the Jacobian of the corresponding trained NBO-RC network and consider a two-dimensional plane centered at the trained parameter vector \(\boldsymbol{\eta}^{*}\), defined by \(\boldsymbol{\eta}(s_1,s_2)=\boldsymbol{\eta}^{*}+s_1\mathbf{v}_1+s_2\mathbf{v}_2\), where \(s_1,s_2\in\mathbb{R}\) denote scalar displacements along the parameter-space directions \(\mathbf{v}_1\) and \(\mathbf{v}_2\), respectively.
For each metric \(Q\in\{\mathrm{RC},E,\mathrm{MSE}\}\), its local quadratic form on this plane is therefore
\(\begin{bmatrix}s_1&s_2\end{bmatrix}
\mathbf K_Q
\begin{bmatrix}s_1&s_2\end{bmatrix}^{T}\),
where the \(2\times2\) matrix \(\mathbf K_Q\) has entries
\((\mathbf K_Q)_{ab}
=\frac{1}{B}\sum_{i\in\mathcal B}
\delta\boldsymbol{\alpha}_i(\mathbf v_a)^T
\mathbf Q_i
\delta\boldsymbol{\alpha}_i(\mathbf v_b)\),
\(a,b\in\{1,2\}\).
Here, \(\mathbf Q_i=\mathbf G_i\), \(\mathbf E\), and \(\mathbf I\) for the RC, velocity-energy, and MSE-reference metrics, respectively, with the corresponding loss-normalization factors retained.
For each benchmark, the displayed plane is expressed in the eigenvector coordinates of \(\mathbf K_{\mathrm{RC}}\). 
The eigenvectors associated with the larger and smaller eigenvalues define the stiff and weak directions, respectively.
Each quadratic form is shifted and rescaled independently to \([0,1]\), so Figure~\ref{fig:mechanistic_probe}e compares the shape and orientation of the local quadratic basins rather than their absolute curvature magnitudes.

\textbf{Rank-percentile correspondence.}
For Figure~\ref{fig:mechanistic_probe}f, we evaluate the directional rank correspondence separately for Kolmogorov flow (left) and Taylor--Couette flow (right). For each benchmark, we convert each set of directional curvatures to percentile ranks,
\(
r_j
=
\frac{
\operatorname{rank}(\kappa_j)-1
}{
N_{\mathrm{dir}}-1
},
\)
where average ranks are assigned to tied values. We retain directions in the upper half of the velocity-energy ranking, $r_E\geq 0.5$. The RC-to-velocity-energy correspondence is computed using the NBO-RC network, while the velocity-energy-to-MSE correspondence is computed separately using the NBO-DD network. The displayed connections report the median rank correspondences within equal-width percentile bins. The two segments therefore represent rank relations obtained from their respective trained networks rather than three metrics evaluated using a single common network Jacobian.

\putbib[reference]
\end{bibunit}

\end{document}